\RequirePackage{amsthm}
\documentclass[sn-mathphys,Numbered]{sn-jnl}

\usepackage{graphicx}%
\usepackage{multirow}%
\usepackage{amsmath,amssymb,amsfonts}%
\usepackage{amsthm}%
\usepackage{mathrsfs}%
\usepackage[title]{appendix}%
\usepackage{xcolor}%
\usepackage{textcomp}%
\usepackage{manyfoot}%
\usepackage{booktabs}%
\usepackage{algpseudocode}%
\usepackage{listings}%

\usepackage{mathrsfs}%
\usepackage{algorithm2e,setspace}
\usepackage{bbm}
\usepackage{color}
\usepackage{mathabx}
\usepackage{subfig}

\usepackage{tikz}
\usetikzlibrary{calc}
\usetikzlibrary{3d}

\theoremstyle{thmstyleone}%
\hypersetup{
    colorlinks=true, 
    linktoc=all,     
    linkcolor=blue,
}

\newtheorem{thm}{Theorem}
\newtheorem{corollary}[thm]{Corollary}
\newtheorem{definition}[thm]{Definition}
\newtheorem{lem}[thm]{Lemma}

\newtheorem{rmk}[thm]{Remark}

\newcommand{\x}{\mathbf{x}}

\newcommand\numberthis{\addtocounter{equation}{1}\tag{\theequation}}

\newcommand{\norm}[2]{\left\|#1\right\|_{#2}}

\newcommand{\vb}[1]{\mathbf{#1}}
\newcommand{\abs}[1]{\left|#1\right|}
\newcommand{\E}{\mathbb{E}}

\newcommand{\pr}{\mathbb{P}}
\newcommand{\R}{\mathbb{R}}
\newcommand{\C}{\mathbb{C}}

\newcommand{\eps}{\epsilon}
\newcommand\sbullet[1][.5]{\mathbin{\vcenter{\hbox{\scalebox{#1}{$\bullet$}}}}}

\newcommand{\tensor}{\otimes}
\renewcommand{\E}{\mathbb{E}}

\DeclareMathOperator{\tr}{tr}

\DeclareMathOperator*{\argmin}{arg\,min}

\begin{document}
\title{Fast and Low-Memory Compressive Sensing Algorithms for Low Tucker-Rank Tensor Approximation from Streamed Measurements
}


\author*[1]{\fnm{Cullen} \sur{Haselby}}\email{haselbyc@msu.edu}
\author[1,2]{\fnm{Mark A.} \sur{Iwen}}\email{iwenmark@msu.edu}
\author[3]{\fnm{Deanna} \sur{Needell}}\email{deanna@math.ucla.edu}
\author[4]{\fnm{Elizaveta} \sur{Rebrova}}\email{elre@princeton.edu}
\author[3]{\fnm{William} \sur{Swartworth}}\email{wswartworth@math.ucla.edu}

\affil*[1]{\orgdiv{Department of Mathematics}, \orgname{Michigan State University}, \orgaddress{  \city{East Lansing}, \state{MI}, \country{USA}}}

\affil[2]{\orgdiv{Department of Computational Mathematics, Science and Engineering}, \orgname{Michigan State University}, \orgaddress{  \city{East Lansing}, \state{MI}, \country{USA}}}

\affil[3]{\orgdiv{Department of Mathematics}, \orgname{University of California}, \orgaddress{ \city{Los Angeles}, \state{CA}, \country{USA}}}

\affil[4]{\orgdiv{ Department of Operations Research \& Financial Engineering}, \orgname{Princeton University}, \orgaddress{ \city{Princeton}, \state{NJ}, \country{USA}}}


\abstract{In this paper we consider the problem of recovering a low-rank Tucker approximation to a massive tensor based solely on structured random compressive measurements.  Crucially, the proposed random measurement ensembles are both designed to be compactly represented (i.e., low-memory), and can also be efficiently computed in one-pass over the tensor.  Thus, the proposed compressive sensing approach may be used to produce a low-rank factorization of a huge tensor that is too large to store in memory with a total memory footprint on the order of the much smaller desired low-rank factorization.  In addition, the compressive sensing recovery algorithm itself (which takes the compressive measurements as input, and then outputs a low-rank factorization) also runs in a time which principally depends only on the size of the sought factorization, making its runtime sub-linear in the size of the large tensor one is approximating.  Finally, unlike prior works related to (streaming) algorithms for low-rank tensor approximation from such compressive measurements, we present a unified analysis of both Kronecker and Khatri-Rao structured measurement ensembles culminating in error guarantees comparing the error of our recovery algorithm's approximation of the input tensor to the best possible low-rank Tucker approximation error achievable for the tensor by any possible algorithm.  We further include an empirical study of the proposed approach that verifies our theoretical findings and explores various trade-offs of parameters of interest.}




\maketitle

\section{Introduction}

With the rapid increase in data acquisition and data-driven applications comes the need for efficient methods to acquire, store, reconstruct, and analyze large-scale data. In many settings, this data is multi-modal, making a tensor representation the most natural.  These settings range from medical to communications to a widespread use in data science in general \cite{kolda2009tensor,cichocki2011tensor,sun2008incremental,le2001diffusion,chen2021tensor}. Although the tensor is a higher mode analogue of the matrix, linear algebraic and computational results for matrices generally do not extend beyond two modes. However, as with matrices, tensor data often has an implicit low-rank structure that can be utilized for efficient computation. In contrast to matrices, there are several non-equivalent notions of low-rankness for tensors. Such low-rank structures structures arise from factorizations such as Tucker, CANDECOMP/PARAFAC
(CP), tubal, and tensor train \cite{kolda2009tensor}. 

Because of the large-scale nature of tensor data -- both in the number of modes and the dimensionality of each mode, computationally efficient methods are critical for computing such factorizations, as well as recovering the data from compressed measurements. An ideal framework for this recovery should reduce overall memory requirements, involve measurement schemes with fast matrix-vector multiplies, apply to a wide range of low-rank decompositions, be robust to non-exact low-rankness, and provide a computationally tractable, provably accurate recovery scheme. The framework should apply to both static and streaming settings, the latter occurring when the tensor data is being updated sequentially over time (e.g. slice by slice, entry-wise, or rank-one updates) and when one wishes to save on space by only storing a compressed version of the data.

\subsection{Prior Work}

A number of papers aim to obtain low-rank tensor decompositions in a fast an efficient way, often using randomization. Some of the Tucker-rank related works include \cite{kaya2016high,austin2016parallel,zhou2014decomposition,tsourakakis2010mach}.  Here, we assume that the tensor is received as part of a stream of data, i.e., the complete raw tensor is not stored and we minimize the number of visits to each tensor entry (to at most one or two times each, referring to the one-pass or two-pass algorithms, respectively).  In the earlier work \cite{Malik2018}, the authors describe a variant of Tucker-Alternating Least Squares (aka Higher Order Orthogonal Iteration, multi-pass scenario) that employed TensorSketch to produce the necessary measurements. However, the quality of reconstruction for iterative schemes is sensitive to the initialization used for the factors and core, as well as requiring possibly many iterations to converge or overcome ``swamps'', a well documented issue with the ALS approach.

Herein we aim to recover a low-rank tensor from oblivious linear measurements such as those employed in compressive sensing tasks. Some of these measurements satisfy the Tensor Restricted Isometry Property (TRIP) and are amenable to provable recovery by iterative algorithms, related tensor works include \cite{rauhut2017low,grotheer2020stochastic,haselby2023modewise}. One of the main goals of this work is to make related measurement maps and recovery algorithms more memory-efficient by using more structured linear measurements.  It is known to be hard to design such measurements that also satisfy the TRIP \cite{haselby2023modewise,jiang2022near}.  As a result, in this work we employ direct recovery procedures that allow for error analysis that does not rely on the TRIP.

In the prior work of Hendrikx and De Lathauwer \cite{DeLath2022}, the authors  describe an algorithm that recovers a tensor from linear combinations of its entries where the measurement operations themselves are constrained to be Kronecker-structured. In that work they give necessary and sufficient conditions for perfect recovery of exactly low-rank tensors. They also describe some relevant heuristics and provide empirical results for the performance of the recovery when the tensor or its measurements are subject to white noise. Additionally they demonstrate how to adapt their method to two other tensor formats -- CP and tensor train.  

In the work of Sun et al \cite{Sun2019}, a comparable low Tucker-rank recovery procedure to \cite{DeLath2022} was described (Algorithm 4.1, Tucker Sketch).  However, the measurement ensembles analyzed therein were both random and structurally less constrained, allowing for a tractable probabilistic analysis (i.e., the Gaussian measurements used to estimate factors in \cite{Sun2019} were not constrained to be strictly modewise as in the Kronecker-structured case considered in \cite{DeLath2022}).  As a result, quasi-optimal low-rank approximation guarantees for unstructured Gaussian linear measurements of the recovered tensor are proven in expectation in \cite{Sun2019}.  Furthermore, motivated by huge memory requirements of unstructured Gaussian measurement maps analyzed therein (which will generally take up as much storage as the data tensor itself), the authors also empirically investigate the performance of more structured measurement maps constructed via Khatri-Rao products of smaller Gaussian matrices.  In those experiments they show that these Khatri-Rao structured measurements deliver nearly as good approximations in practice as the unstructured memory-hungry maps they analyzed theoretically, with the additional benefit of also requiring significantly less space to store. 

\subsection{Contribution}
In this work, we further study recovery from the Kronecker and Khatri-Rao measurement ensembles introduced in \cite{Sun2019, DeLath2022}. Our analysis  unifies and adds to both of these works in several ways. 
First, unlike in \cite{DeLath2022}, our error analysis applies in the non-exact low-rank case, and quantifies how our low-rank approximation error bounds depend on the relevant parameters of the problem.  Second, though the overarching structure of our theoretical analysis employs a similar strategy to that in \cite{Sun2019}, it is instead applied to more structured, memory efficient, and difficult to analyze Kronecker and Khatri-Rao structured measurement ensembles.  The acquired measurements can then be used in a single unified recovery method (see Algorithm \ref{alg:loo_one_pass_prime}, which effectively matches those utilized in both \cite{Sun2019} and \cite{DeLath2022}) to produce a low-rank factorization.  

Among other differences from \cite{Sun2019}, we believe it is important to emphasize that the more structured measurement ensembles analyzed herein can operate on the tensor data in subtly different modewise manners than the measurements analyzed there.  That is, our measurements yield several smaller tensors whose entries are obtained by matrix-vector operations involving measurement matrices operating on various fibers or slices of the tensor data. This can have significant practical advantages since neither the entire original tensor nor the entire measurement operator need to be kept in working memory in order to obtain such measurements.  Additionally, besides applying the analysis to these more structured measurement ensembles, we remove the theoretical assumption in \cite{Sun2019} that the entries of the sketching matrices are drawn independently from a Gaussian distribution, and instead rely on a more general Johnson-Lindenstrauss property-based analysis in order to derive with-high-probability recovery guarantees for more general measurement ensembles.  The advantage of this is that there are several well known distributions of random matrices with, e.g., fast matrix-vector multiplies that are known to satisfy this property, and our alternate analysis makes it a straightforward exercise to derive measurement requirements and error bounds for measurement ensembles constructed from such different choices of sketching matrices. 

Theorem~\ref{thm:sketchdimonepass} summarizes our main results. In order to state it, we will first describe some tensor concepts and operations related to the measurement process. The precise statements of the results are given in Theorem~\ref{thm:onepass_kron_nondist} (for generic measurement ensembles) and Theorems~\ref{thm:onepass_kron} and \ref{thm:onepass_khat} (specialized for sub-gaussian measurement ensembles).

\subsection{Tensor and Measurement Preliminaries} 
\label{sec:tensor_notation}
Here we recall terminology and set notation that will be useful in describing the tensor operations used in stating our main result and throughout this paper.  For a more thorough introduction we refer the reader to, e.g., \cite{zare2018extension,kolda2009tensor}.

 \textbf{Tensor order, fibers, and unfoldings.}   The \emph{order} of a tensor is its number of  \emph{modes}. That is, $\mathcal{X}\in\mathbb{R}^{n_1\times \dots n_d}$ is an order $d$ tensor, or a $d$-mode tensor. \emph{Mode-$j$ fibers} are the tensor analogue of rows and columns in the matrix case. They are the vectors given by fixing all but one of the indices and varying the $j$-th coordinate. For example, the 3-mode tensor tensor $\mathcal{X}\in\mathbb{R}^{n_1 \times n_2 \times n_3}$ will have  $n_2n_3$ mode-1 fibers $\mathbf{x}_{:, j, k} \in \mathbbm{R}^{n_1}$ indexed by $j\in[n_2]$ and $k\in[n_3]$, $n_1n_3$ mode-2 fibers $\mathbf{x}_{i, :, k} \in \mathbbm{R}^{n_2}$ indexed by $i\in[n_1]$ and $k\in[n_3]$, and $n_1 n_2$ mode-3 fibers $\mathbf{x}_{i, j, :} \in \mathbbm{R}^{n_3}$ indexed by $i\in[n_1]$ and $j\in[n_2]$.
The \emph{mode-$j$ unfolding} of a tensor $\mathcal{X} \in\mathbb{R}^{n_1\times \dots n_d}$ is an $n_j \times \prod_{\substack{k=1, k\neq j}}^{d} n_k$ matrix, $X_{[j]}$, formed by arranging the mode-$j$ fibers of $\mathcal{X}$ as the columns of $X_{[j]}$. The ordering of these columns is not important so long as it is consistent across calculations. 
    
    \textbf{Inner product and norm.} The set of all $d$-mode tensors $\mathcal{X} \in \mathbbm{R}^{n_1 \times \ldots \times n_d}$ forms a vector space when equipped with component-wise addition and scalar multiplication. The \emph{standard inner product} of two tensors $X, Y \in \mathbbm{R}^{n_1 \times \ldots \times n_d}$ is
    \begin{equation*}
        \langle \mathcal{X},\mathcal{Y} \rangle := \sum_{j_1 \in [n_1],\ldots j_d \in [n_d]} \mathcal{X}_{j_1, \ldots, j_d} \mathcal{Y}_{j_1, \ldots, j_d} \,.
    \end{equation*}
     The standard tensor norm is then defined by this inner product as $\norm{\mathcal{X}}{2} := \sqrt{\langle \mathcal{X},\mathcal{X} \rangle}$.  Note that the standard tensor inner product and norm above correspond to the standard dot product and Euclidean norm as applied to vectorized tensors in $\mathbbm{R}^{N}$, where $N = \prod_{j = 1}^d n_j$.

\begin{definition}[Modewise product] \label{def:mwprod} The $j$-mode product of a $d$-mode tensor $\mathcal{X}\in\R^{n_1 \times\dots\times n_d}$ with a matrix $\Phi \in \R^{m_j \times n_j}$ is denoted as $\mathcal{X}\times_j \Phi$, and defined componentwise as
\begin{equation}
    \left( \mathcal{X} \times_j \Phi \right)_{i_1, \ldots, i_{j-1}, \ell, i_{j+1}, \ldots, i_d}=\sum\limits_{i_j=1}^{n_j} \Phi_{\ell, i_j}\mathcal{X}_{i_1, \ldots, i_j, \ldots, i_d}
\end{equation}
for $\ell \in [m_j]$ and $i_j \in [n_j]$. Equivalently, for the respective unfoldings, $\left( \mathcal{X}\times_j \Phi \right)_{[j]}= \Phi X_{[j]}$.
\end{definition}


Additionally, we let $\circ$ denote the outer product of two tensors.  That is, if $\mathcal{X}\in\R^{n_1 \times\dots\times n_d}$ and $\mathcal{Y}\in\R^{m_1 \times\dots\times m_f}$, then $\mathcal{X} \circ \mathcal{Y} \in \mathbbm{R}^{n_1 \times\dots\times n_d \times m_1 \times \dots \times m_f}$ with entries given by $\left(\mathcal{X} \circ \mathcal{Y} \right)_{j_1, \dots, j_d, i_1, \dots, i_f} := \mathcal{X}_{j_1, \dots, j_d} \mathcal{Y}_{i_1, \dots, i_f}$.

\subsubsection{The Tucker Decomposition and Low-Rank Approximation}

Next, we define the Tucker decomposition of tensors, which decomposes a tensor into a core tensor that is multiplied by a factor matrix along each of its modes. Crucially, this decomposition relates to the proposed modewise measurement operations in a convenient fashion as we detail in section~\ref{sec:MainResBRKS}.

\begin{definition}[Tucker decomposition] \label{def:tucker} Let $\mathcal{X}\in\R^{n_1\times \dots\times n_d}$.  We say that $\mathcal{X}$ has a Tucker decomposition of rank ${\bf r} = (r_1,\dots,r_d)$ if there exists $\mathcal{G}\in\R^{r_1 \times \dots \times r_d}$, and $U_j \in \R^{n_j \times r_j}$ with orthonormal columns for all $j\in [d]$, so that
\begin{equation} \label{eq:tucker}
    \mathcal{X} = \mathcal{G}\times_1 U_1 \times_2 U_2 \times_3 \dots \times_d U_d = \sum_{\vb{i} \in \{ (i_1, \dots, i_d) ~|~ i_j \in [r_j]~ \forall j \in [d] \}} \mathcal{G}_{\vb{i}} \cdot \vb{u}^{(1)}_{i_1} \circ \vb{u}^{(2)}_{i_2} \circ \dots \circ \vb{u}^{(d)}_{i_d},
\end{equation}
where $\vb{u}^{(j)}_{i_j}$ denotes the $i_j$-th column of $U_j$, and $\mathcal{G}_{\vb{i}} := \mathcal{G}_{i_1, \dots, i_d} \in \R$. We will also use the shorthand $\mathcal{X} = [\![ \mathcal{G}, U_1,\dots, U_d ]\!] := \mathcal{G}\times_1 U_1 \times_2 U_2 \times_3 \dots \times_d U_d$ when referring to the rank-${\vb r}$ Tucker decomposition of a given rank-${\vb r}$ tensor $\mathcal{X}$.
\end{definition}

There is an equivalent formulation to equation \eqref{eq:tucker} using unfoldings of $\mathcal{X}$ and of $\mathcal{G}$ in terms of matrix-matrix products:
\begin{equation}
\label{eqn:mat_tucker}
    X_{[j]} = U_j G_{[j]} \left(U_d \otimes \dots\otimes U_{j+1} \otimes U_{j-1} \otimes \dots \otimes U_1\right)^T,
\end{equation}
where $\otimes$ denotes the standard Kronecker matrix product (see \eqref{def:KronProd} and, e.g., \cite{DeLath2000}).

Given an arbitrary tensor $\mathcal{X}\in\R^{n_1\times \dots\times n_d}$ of unknown Tucker rank it is often of interest to compute a good rank-${\bf r}$ approximation of $\mathcal{X}$.  In fact, an optimal Tucker rank-${\bf r}$ minimizer $[\![\mathcal{X}]\!]^{\rm opt}_{\bf r} \in \R^{n_1\times \dots\times n_d}$ satisfying 
\begin{equation}
\norm{\mathcal{X} - [\![\mathcal{X}]\!]^{\rm opt}_{\bf r}}{2} = \inf_{\displaystyle{\substack{\mathcal{G}\in\R^{r_1 \times \dots \times r_d},\\ U_j \in \R^{n_j \times r_j}}}} \left\|\mathcal{X} - [\![ \mathcal{G}, U_1,\dots, U_d ]\!] \right \|_2
\label{Def:OptRankrApprox}
\end{equation}
always exists (see, e.g., Theorems 10.3 and Theorem 10.8 in \cite{Hackbusch2012}).  However, computing such an $[\![\mathcal{X}]\!]^{\rm opt}_{\bf r}$ (which is not unique) is generally a challenging task that is accomplished only approximately via iterative techniques (see, e.g., \cite{kolda2009tensor}).
As a result, in such situations one usually seeks to instead compute a \emph{quasi-optimal} rank-${\bf r}$ tensor  $\tilde{\mathcal{X}}$ satisfying
\begin{equation} \label{equ:quasi_optimal}
    \norm{\mathcal{X} - \tilde{\mathcal{X}}}{2} \leq C \inf_{\displaystyle{\substack{\mathcal{G}\in\R^{r_1 \times \dots \times r_d},\\ U_j \in \R^{n_j \times r_j}}}} \left\|\mathcal{X} - [\![ \mathcal{G}, U_1,\dots, U_d ]\!] \right \|_2,
\end{equation}
where $C \in \R^+$ is a positive constant independent of $\mathcal{X}$.  

The following lemma demonstrates that one can recover a quasi-optimal rank-{\bf r} approximation of an arbitrary tensor $\mathcal{X}\in\R^{n_1\times \dots\times n_d}$ by simply computing the left singular vectors of $X_{[j]}$ for all $j \in [d]$. It is a variant of \cite[Theorem 10.3]{Hackbusch2012} which we prove here for the sake of completeness.  

\begin{lem}
\label{lem:QuasiOptimalfromSVDs}
Let $\mathcal{X}\in\R^{n_1\times \dots\times n_d}$, denote the $i^{\rm th}$-largest singular value of $X_{[j]}$ by $\sigma_i \left(X_{[j]} \right)$, and define $\Delta_{r,j} := \sum_{i=r+1}^{\tilde{n}_j} \sigma_i \left( X_{[j]} \right)^2$ for all $j \in [d]$ and $r \in \left[\tilde{n}_j := \min \left\{ n_j, \prod_{j \neq k} n_k\right\} \right]$.  Fix ${\bf r} \in [\tilde{n}_1] \times \dots \times [\tilde{n}_d]$ and suppose that $[\![\mathcal{X}]\!]^{\rm opt}_{\bf r}$ satisfies \eqref{Def:OptRankrApprox}.  Then,
\begin{equation*}
    \sum_{j=1}^d \Delta_{r_j,j} \leq d \norm{   \mathcal{X} - [\![\mathcal{X}]\!]^{\rm opt}_{\bf r}}{2}^2.
\end{equation*}
\end{lem}
\begin{proof}
    Let $P_j \in \R^{n_j \times n_j}$ be a rank $r_j$ orthogonal projection matrix satisfying $\norm{ X_{[j]} - P_j X_{[j]} }{F}^2 = \Delta_{r_j,j} \leq \norm{ X_{[j]} - Y }{F}^2$ for all rank at most $r_j$ matrices $Y$.\footnote{The Eckart-Young Theorem guarantees that we can compute such a $P_j$ from the left singular values of $X_{[j]}$.}  Noting that $Y = \left( [\![\mathcal{X}]\!]^{\rm opt}_{\bf r} \right)_{[j]}$ will be a rank at most $r_j$ matrix by \eqref{eqn:mat_tucker}, we then have that
    $$\Delta_{r_j,j} \leq \norm{ X_{[j]} - \left( [\![\mathcal{X}]\!]^{\rm opt}_{\bf r} \right)_{[j]} }{F}^2 = \norm{   \mathcal{X} - [\![\mathcal{X}]\!]^{\rm opt}_{\bf r}}{2}^2~~\forall j \in [d].$$
    Summing over the $d$ values of $j$ above now finishes the proof.
\end{proof}
Note, in order to simplify notation moving forward, we will assume that $n_j = n$ and $r_j = r$ for all $j\in[d]$; that is our tensors will have equal side lengths and the rank will likewise be the same in every mode, thus suppressing the need for additional level of sub-scripting. 
\subsection{Kronecker Products, Khatri-Rao Products, and Leave-One-Out Modewise Measurement Notation and Examples}
\label{sec:The_MEASUREMENTS!}
We now will describe and provide some notation for the measurement operators used throughout the rest of this paper. The general concept we are able to specialize to several cases is the following:

\begin{definition}[Leave-One-Out Measurements] \label{def:loo_measurements} Given a $d$-mode tensor $\mathcal{X}$ with side-lengths $n$, leave-one-out measurements are the result of reducing all but one mode using linear operations on the unfolding of the tensor. That is, \[
B_j = \Omega_{(j,j)} X_{[j]} \Omega_{-j}^T
\] are leave one out measurements whenever $\Omega_{(j,j)}\in\R^{n\times n}$ and full-rank, $\Omega_{-j} \in \R^{ m^{\prime} \times n^{d-1}}$ where $m^{\prime} \leq n^{d-1}$.
\end{definition}
Any measurement process that can be equivalently described as left multiplication of an unfolded tensor by a full-rank square matrix and right multiplied by some other linear operator that reduces all other modes is a leave-one measurement process. We will consider three specific cases for how to construct the overall measurement operator $\Omega_{-j}$ in Definition \ref{def:loo_measurements}: Kronecker structured, Khatri-Rao structured and unstructured measurement operators.

We will need to recall some matrix products.
The \emph{Kronecker product} of matrices $A \in \R^{m_1 \times n_1}$ and $B\in\R^{m_2 \times n_2}$ is the matrix $A\otimes B \in\R^{m_1 m_2 \times n_1 n_2}$ and is defined by
\begin{equation}
A\otimes B := \begin{bmatrix}
a_{1,1} B & a_{1,2} B & \dots& a_{1,n_1} B \\
a_{2,1} B & a_{2,2} B & \dots& a_{2,n_1} B \\
\vdots  & \vdots & \ddots& \vdots \\
a_{m_1, 1} B & a_{m_1,2} B & \dots& a_{m_1,n_1} B \\
\end{bmatrix}.
\label{def:KronProd}
\end{equation}
The \emph{Khatri-Rao product} of two matrices is the matrix that results from computing the Kronecker product of their matching columns. That is, for $A \in \R^{m_1 \times n}$ and $B\in\R^{m_2 \times n}$, their Khatri-Rao product is the matrix $A \odot B \in \R^{m_1 m_2 \times n}$ defined by
\[
A\odot B := \begin{bmatrix}
\vb{a}_1 \otimes \vb{b}_1 &\vb{a}_2 \otimes \vb{b}_2 & \dots& \vb{a}_n \otimes \vb{b}_n 
\end{bmatrix}.
\]
In this paper, we also use a so-called row-wise Khatri-Rao product (sometimes called the \emph{face-splitting product}) of $A \in \R^{m \times n_1}$ and $B\in\R^{m \times n_2}$, denoted by $A \sbullet B \in \R^{m \times n_1 n_2}$, and defined by $(A \sbullet B)^T := A^T \odot B^T$. Given the close relationship between the Khatri-Rao and face-splitting products, we will refer to both of them as being ``Khatri-Rao structured'' below.

\subsubsection{Kronecker Structured Measurements} 
\label{sec:KronMeasIntro}

Kronecker structured leave-one-out measurements are constructed by taking the Kronecker product of several component matrices, in addition to one square measurement matrix of full-rank to be applied to the mode which is uncompressed (which may simply be the identity). We will require leave one out measurements for each mode. Collectively the matrices needed to define this ensemble will form a set $\left\{\Omega_{(i,j)}\right\}_{i,j\in [d]}$ where $\Omega_{(i,j)}\in\R^{m\times n}$ when $j\neq i$, and where $\Omega_{(i,i)}\in\R^{n\times n}$ for all $i \in [d]$. 

Our measurements then are
\begin{align}
\label{equ:BRKSfactors}
    \mathcal{B}_j := \mathcal{X} \bigtimes_{i=1}^d \Omega_{(j,i)} 
\end{align}
for all $j\in[d]$ using the matrices $\left\{\Omega_{(j,i)}\right\}_{i=1}^d$. Crucially, we can identify this tensor of measurements $\mathcal{B}_j$ with a flattened version which makes it clear that Kronecker structured modewise measurements do define leave-one-out measurements conforming to Definition \ref{def:loo_measurements}.
\begin{equation} \label{eq:kronecker-sketch}
    B_j = \Omega_{(j,j)} X_{[j]} \Omega_{-j}^T =  \Omega_{(j,j)} X_{[j]}(\Omega_{(j,1)} \otimes \Omega_{(j,2)} \otimes \dots \Omega_{j,j-1} \otimes \Omega_{j,j+1}\otimes \dots \otimes \Omega_{(j,d)})^T
\end{equation}
See Algorithm~\ref{alg:kron_sketching} for an outline of the sketching procedure. 

Conceptually, the Kronecker products of matrices can be rewritten in terms of matrix products of auxiliary matrices. This will be useful if we wish to examine via, e.g., \eqref{eqn:mat_tucker}, how several modewise products change the properties of the resulting factor matrices $U_j$ in the Tucker decomposition of a given tensor. As an illustration, consider the case of a three mode tensor $\mathcal{X} \in\R^{n_1 \times n_2 \times n_3}$. Let $\vb{x} \in \mathbbm{R}^{n_1 n_2 n_3}$ denote the vectorization of $\mathcal{X}$, and further suppose that $\Omega_j \in\R^{m_j\times n_j}$ for $j=1,2,3$ are three measurement matrices used to produce modewise measurements of $\mathcal{X}$ given by $\mathcal{X} \times_1 \Omega_1 \times \Omega_2 \times_3 \Omega_3$.  Allowing for an additional reshaping, we can identify these three modewise operations equivalently with a single matrix-vector product using a variant of \eqref{eqn:mat_tucker}. That is, one can see that $\text{vec}\left(\mathcal{X}\times_1 \Omega_1 \times \Omega_2 \times_3 \Omega_3\right) = \Omega\vb{x}$ where $\Omega = \Omega_3 \otimes \Omega_2\otimes \Omega_1 \in \R^{m_1 m_2 m_3 \times n_1 n_2 n_3}$.  Let $I_n$ denote the $n \times n$ identity matrix.  The mixed-product property of the Kronecker product can now be used to further show that in fact $\Omega =\tilde{\Omega}_3 \tilde{\Omega}_2 \tilde{\Omega}_1$, where
\begin{align*}
    \tilde{\Omega}_1 &= I_{n_3} \otimes I_{n_2} \otimes \Omega_1 \in \R^{m_1 n_2 n_3 \times n_1 n_2 n_3}, \\
    \tilde{\Omega}_2 &= I_{n_3} \otimes \Omega_2 \otimes I_{m_1} \in \R^{ m_1 m_2 n_3 \times m_1 n_2 n_3}, \\
    \tilde{\Omega}_3 &= \Omega_3 \otimes I_{m_2} \otimes I_{m_1} \in \R^{m_1 m_2 m_3 \times m_1 m_2 n_3}. 
\end{align*}
Hence, $\text{vec}\left(\mathcal{X}\times_1 \Omega_1 \times \Omega_2 \times_3 \Omega_3\right) = \tilde{\Omega}_3 \tilde{\Omega}_2 \tilde{\Omega}_1 \text{vec}(\mathcal{X})$.  This example can be easily be extended to any number of modes.

\subsubsection{Khatri-Rao Structured Measurements}
\label{sec:bound_2pass_error_khat}
Leave-one-out measurement ensembles which are Khatri-Rao products of smaller maps are also possible and have been considered before in works such \cite{Sun2019}. That is,
\begin{align*}
    B_j &= \Omega_{(j,j)} X_{[j]}\Omega^T_{-j}\\
    &= \Omega_{(j,j)}X_{[j]} \left(\Omega_{(j,1)}^T \odot \Omega_{(j,2)}^T \odot \dots  \Omega^T_{(j,j-1)} \odot \Omega^T_{(j,j+1)} \dots \odot \Omega_{(j,d)}^T \right)^T\\
    &=  \Omega_{(j,j)}X_{[j]} \left(\Omega_{(j,1)} \sbullet \Omega_{(j,2)} \sbullet \dots  \Omega^T_{(j,j-1)} \sbullet \Omega_{(j,j+1)} \dots \sbullet \Omega_{(j,d)} \right)^T \numberthis \label{eqn:loo_khat}
\end{align*}
where $\Omega_{(j,j)}$ is a full-rank square matrix, and $\Omega_{(j,i)}\in \R^{m\times n}$ for $j\neq i$. Note that in this case we can consider $\Omega_{-j} \in \R^{m \times n^{d-1}}$ as sketching sub-tensors of size $n^{d-1}$ to size $m$ (as opposed to $m^{d-1}$ as was done previously with Kronecker structured measurement maps). See Algorithm~\ref{alg:khat_sketching} for an outline of the sketching procedure.

\subsubsection{Unstructured Measurements}

In \cite{Minster2019}, a type of leave-one-out measurements are described and analyzed, however in that work their measurement ensembles are not structured, but instead are Gaussian matrices where $\Omega_{-j} \in \R^{m \times n^{d-1}}$ has entries all drawn independently and is right multiplied by an unfolding of the tensor:
\[
   B_j = \Omega_{(j,j)} X_{[j]}\Omega^T_{-j}\\
\]
where again $\Omega_{(j,j)}$ is some square, full-rank matrix.

As an important observation, the storage requirements for the sketching matrix itself in this case is large, comparable in size to the data itself, which can be undesirable. Furthermore, the error analysis in \cite{Minster2019} conducted on this type of measurement when the unstructured matrix $\Omega_{-j}$ has i.i.d. Gaussian entries relies on probabilistic bounds known for Gaussian matrices. Comparable bounds for other distributions, or for matrices constructed using Kronecker or Khatri-Rhao products are not covered in the analysis. Note, Kronecker or Khatri-Rao structured measurements alleviate to a large degree the storage problem associated with $\Omega_{-j}\in\R^{m \times n^{d-1}}$ when compared to in the unstructured case since the entire sketching matrix does not need to be maintained in memory, but rather just constituent components $\Omega_{(j,i)}\in \R^{m\times n}$ can be stored, and the action of the Kronecker or Khatri-Rao product on the tensor can computed as part of the sketching phase. This does incur additional operations in the sketching phase, and in our runtime analysis and numerical experiments we detail the trade-off between space and run-time for these various choices of measurement operators.

\subsubsection{Core Measurements} Regardless of which type leave-one-measurements are used for each of the modes, in order to recover the full factorization using only a single access to the data $\mathcal{X}$, we will require an additional set of measurements for use in estimating the core. These are in all cases, modewise, and can compress all modes.

\begin{align}
\label{equ:BRKScore}
    \mathcal{B}_c := \mathcal{X} \bigtimes_{j = 1}^d \Phi_j 
\end{align}
using the matrices $\left\{\Phi_{j}\right\}_{j=1}^d$. See Algorithm~\ref{alg:core_sketching} for an outline of the sketching procedure for the core. 

All together then, these $d+1$ measurement tensors, $d$ of the leave-one-out type and one of the type in \eqref{equ:BRKScore} will be used to recover the parts of the original full tensor. Figure~\ref{fig:brks_diagram} is a schematic rendition of the overall measurement procedure in the three mode case ($d = 3$).

Note that the $n^d$ entries of the original tensor $\mathcal{X}$ are compressed into $d n m^{d-1} + m_c^d$ entries of the $d+1$ total different measurement tensors; one leave-one-out measurement for each of the $d$ modes as well as the one measurement tensor for use in estimating the core. Recovery naturally will require the storage of the measurement operators in some fashion along with the measurement tensors. As dense matrices, the measurement operators collectively have $d((d-1)mn +n^2)+dm_c n$ total entries. However the number of random variables required to generate these matrices may be significantly fewer depending on the method employed. For example, when using sub-sampled Fourier transforms, the number of random bits required to construct an ensemble is linear in $n$. Additionally, some choices of measurement matrices allow for near-linear time matrix-vector multiplication.  This is often accomplished by choosing matrices that exploit either the Fast-Fourier transform, or fast Hadamard multiplication, and will allow us to compress our tensors faster than with Gaussian measurements for example.

As we detail, applying the measurement operators is asymptotically the most computationally intensive part of the algorithm, and thus in settings where it is useful to economize the computational effort to obtain measurements these matrices have significant advantages. Whichever choice for type of measurements are used, however, the efficiency in terms of the run-time of the recovery algorithm is largely dependent on the ratios $m/n$ and $m_c / n$. 

\begin{figure}%
\subfloat[\centering Leave-one-measurements from Algorithm \ref{alg:kron_sketching} \label{fig:loo_a}]{
\begin{tikzpicture}[>=latex,scale=1.5]
\pgfmathsetmacro{\x}{1}
\pgfmathsetmacro{\o}{3.2}
\pgfmathsetmacro{\y}{1}
\pgfmathsetmacro{\z}{1.5}
\pgfmathsetmacro{\tx}{0.5}
\pgfmathsetmacro{\ty}{0.5}
\pgfmathsetmacro{\tz}{1.5}
\path (0,0,\y) coordinate (A) (\x,0,\y) coordinate (B) (\x,0,0) coordinate (C) (0,0,0)
coordinate (D) (0,\z,\y) coordinate (E) (\x,\z,\y) coordinate (F) (\x,\z,0) coordinate (G)
(0,\z,0) coordinate (H);
\draw[red,densely dotted,thick] (A)--(B);
\draw[purple, solid] (B)--(C);
\draw[blue,dashed] (C)--(G);
\draw (G)-- (F)-- (B) (A)--(E)-- node[below,black]{$\mathcal{X}$}(F)--(G)--(H)--(E);

\draw[thick, red,densely dotted] ($(A)+(0,-4pt)$) -- node[below,black]{$\Omega_{(1,2)}$}($(B)+(0,-4pt)$);
\draw[thin, black] ($(A)+(0,-4pt)+(0,-0.5)$) -- ($(B)+(0,-4pt)+(0,-0.5)$);

\draw[thick, orange,densely dotted] ($(B)+(0,-4pt)$) -- ($(B)+(0,-4pt)+(0,-0.5)$);
\draw[thin, black] ($(A)+(0,-4pt)$) -- ($(A)+(0,-4pt)+(0,-0.5)$);
\draw[thin, white] ($(A)+(0,-4pt)+(0,-0.5)$) -- ($(A)+(0,-4pt)+(0,-0.8)$);

\draw[thin, purple, solid] ($(B)+(-45:4pt)$) -- ($(C)+(-45:4pt)$);
\draw[thin, black] ($(B)+(-45:4pt)+2*(-45:4pt)$) -- node[below,sloped,black]{$\Omega_{(1,3)}$}($(C)+(-45:4pt)+2*(-45:4pt)$);

\draw[thin, magenta,solid] ($(C)+(-45:4pt)$) -- ($(C)+(-45:4pt)+2*(-45:4pt)$);
\draw[thin, black] ($(B)+(-45:4pt)$) -- ($(B)+(-45:4pt)+2*(-45:4pt)$);

\draw[thin,blue, dashed] ($(C)+(4pt,0)$) -- node[below,sloped,black]{$\Omega_{(1,1)}$}($(G)+(4pt,0)$);
\draw[thin,black] ($(C)+(4pt,0)+(1.5,0)$) --node[right,black]{$=$}($(G)+(4pt,0)+(1.5,0)$);
\draw[thin,cyan, dashed] ($(G)+(4pt,0)$) -- ($(G)+(4pt,0)+(1.5,0)$);
\draw[thin,black] ($(C)+(4pt,0)$) -- ($(C)+(4pt,0)+(1.5,0)$);

\path (\o,0,\ty) coordinate (A) (\tx+\o,0,\ty) coordinate (B) (\tx+\o,0,0) coordinate (C) (\o,0,0)
coordinate (D) (\o,\tz,\ty) coordinate (E) (\tx+\o,\tz,\ty) coordinate (F) (\tx+\o,\tz,0) coordinate (G)
(\o,\tz,0) coordinate (H);
\draw[thick, orange,densely dotted] (A)--(B);
\draw[magenta] (B)--(C);
\draw[cyan,dashed] (C)--(G);
\draw (G)--(F)--(B) (A)--(E)--node[below,black]{$\mathcal{B}_{1}$}(F)--(G)--(H)--(E);

\end{tikzpicture}

} %
~
\subfloat[\centering Core measurements from Algorithm \ref{alg:core_sketching} \label{fig:loo_b} ]{
\begin{tikzpicture}[>=latex,scale=1.5]
\pgfmathsetmacro{\x}{1}
\pgfmathsetmacro{\o}{2.7}
\pgfmathsetmacro{\y}{1}
\pgfmathsetmacro{\z}{1.5}
\pgfmathsetmacro{\tx}{0.8}
\pgfmathsetmacro{\ty}{0.8}
\pgfmathsetmacro{\tz}{1.1}
\path (0,0,\y) coordinate (A) (\x,0,\y) coordinate (B) (\x,0,0) coordinate (C) (0,0,0)
coordinate (D) (0,\z,\y) coordinate (E) (\x,\z,\y) coordinate (F) (\x,\z,0) coordinate (G)
(0,\z,0) coordinate (H);
\draw[thick, red, densely dotted	] (A)--(B);
\draw[purple] (B)--(C);
\draw[blue, dashed] (C)--(G);
\draw (G)-- (F)-- (B) (A)--(E)-- node[below,black]{$\mathcal{X}$}(F)--(G)--(H)--(E);

\draw[thick, red, densely dotted] ($(A)+(0,-4pt)$) -- node[below,black]{$\Phi_{2}$}($(B)+(0,-4pt)$);
\draw[thin, black] ($(A)+(0,-4pt)+(0,-0.8)$) -- ($(B)+(0,-4pt)+(0,-0.8)$);

\draw[thick, orange, densely dotted	] ($(B)+(0,-4pt)$) -- ($(B)+(0,-4pt)+(0,-0.8)$);
\draw[thin, black] ($(A)+(0,-4pt)$) -- ($(A)+(0,-4pt)+(0,-0.8)$);

\draw[thin, purple] ($(B)+(-45:4pt)$) -- node[below,sloped,black]{$\Phi_3$}($(C)+(-45:4pt)$);
\draw[thin, black] ($(B)+(-45:4pt)+4*(-45:4pt)$) -- ($(C)+(-45:4pt)+4*(-45:4pt)$);

\draw[thin, magenta] ($(C)+(-45:4pt)$) -- ($(C)+(-45:4pt)+4*(-45:4pt)$);
\draw[thin, black] ($(B)+(-45:4pt)$) -- ($(B)+(-45:4pt)+4*(-45:4pt)$);

\draw[thin,blue, dashed] ($(C)+(4pt,0)$) -- node[below,sloped,black]{$\Phi_{1}$}($(G)+(4pt,0)$);
\draw[thin,black] ($(C)+(4pt,0)+(0.8,0)$) --node[right,black]{$=$}($(G)+(4pt,0)+(0.8,0)$);
\draw[thin,cyan, dashed] ($(G)+(4pt,0)$) -- ($(G)+(4pt,0)+(0.8,0)$);
\draw[thin,black] ($(C)+(4pt,0)$) -- ($(C)+(4pt,0)+(0.8,0)$);

\path (\o,0,\ty) coordinate (A) (\tx+\o,0,\ty) coordinate (B) (\tx+\o,0,0) coordinate (C) (\o,0,0)
coordinate (D) (\o,\tz,\ty) coordinate (E) (\tx+\o,\tz,\ty) coordinate (F) (\tx+\o,\tz,0) coordinate (G)
(\o,\tz,0) coordinate (H);
\draw[thick, orange,densely dotted] (A)--(B);
\draw[magenta] (B)--(C);
\draw[cyan, dashed] (C)--(G);
\draw (G)--(F)--(B) (A)--(E)--node[below,black]{$\mathcal{B}_c$}(F)--(G)--(H)--(E);

\end{tikzpicture}
} %
\caption{Schematic for the two types of measurement tensors used in recovery Algorithm \ref{alg:kron_1pass} for a three mode tensor. Measurement tensors $\mathcal{B}_{i}$ of the type shown in $(a)$ for $i=1$ will be used to estimate the factors of the Tucker decomposition, and the measurement tensor $\mathcal{B}_c$ in $(b)$ will be used to estimate the core of the Tucker decomposition.}
\label{fig:brks_diagram}

\end{figure}

\subsubsection{The Canonical One-Pass Recovery Algorithm} 

We can now state Algorithm \ref{alg:loo_one_pass_prime} as the canonical algorithm for one-pass recovery the tensor using leave-one-out measurements. The algorithm outputs an estimate in factored form $\mathcal{X}_1 = [\![ \mathcal{H}, Q_1,\dots, Q_d ]\!]$, and we say one-pass to emphasize that after measuring the tensor $\mathcal{X}$, no additional access to the data is required to obtain the estimate $\mathcal{X}_1$. The inputs to the algorithm are leave-one-out measurements $B_i$ for each of the modes of any type, as well as measurements $\mathcal{B}_c$ for use in recovering the core, some of the related measurement operators $\Omega_{(i,i)},\Phi_i$, and a target rank $\vb{r}$ parameter. Algorithm~\ref{alg:loo_one_pass_prime} consists of two loops. The first loop recovers the factor matrices $Q_i$ in the Tucker factorization. The second loop depends on the output of the first loop and computes the core $\mathcal{H}$ of the Tucker factorization. 

Note that within the pseudo-code for Algorithm~\ref{alg:loo_one_pass_prime} the ``unfold'' function in the body of Algorithm \ref{alg:loo_one_pass_prime} takes a tensor and flattens it into the given shape by arranging the specified mode's fibers as the columns, e.g. $H \gets \text{unfold}(\mathcal{B}_c,m_c\times m_c^{d-1},\text{mode}=1)$ 
is equivalent to $H \gets \left(\mathcal{B}_c\right)_{[1]}$ using the notation of Section~\ref{sec:tensor_notation}, the matrix that is obtained by flattening the $d$-mode tensor $\mathcal{B}_{c}$ along mode-$1$. The ``fold'' function is the inverse of ``unfold'', and takes a matrix that is an unfolding along the specified mode of a tensor and reshapes it into a tensor with the specified and compatible dimensions, e.g. $\mathcal{B} \gets \text{fold}(H,r \underbrace{\times m_c\times \dots \times m_c}_{d-1},\text{mode}=1)$ takes the matrix $H$ with dimensions $r\times m_c^{d-1}$ and reshapes it into a $d$ mode tensor where the first mode is of length $r$ and the remaining $d-1$ modes are of length $m_c$.

In the scenario where it is possible to access the original tensor $\mathcal{X}$ twice (we refer to this as the two-pass scenario), we can compute the optimal core $\mathcal{G}$ given our estimated factor matrices $Q_i$ to obtain an estimate $\mathcal{X}_2 = [\![ \mathcal{G}, Q_1,\dots, Q_d ]\!]$. See Algorithm~\ref{alg:kron_2pass} and Algorithm~\ref{alg:khat_2pass} in Appendix~\ref{appendix:algs} for the detailed formulation of the two-pass recovery procedure.

\RestyleAlgo{ruled}

\SetKwComment{Comment}{/* }{ */}
\begin{algorithm}[hbt!]
\SetKwComment{Comment}{\# }{}
\SetKwInOut{Input}{input}\SetKwInOut{Output}{output}%
\caption{One Pass HOSVD Recovery from Leave-One-Out Measurements}	\label{alg:loo_one_pass_prime}
 \Input {\par $B_{i}\in\R^{n\times m^{d-1}}$ for $i\in[d]$ leave-one-out measurements
  \par $\Omega_{(i,i)}\in\R^{n \times n}$ for $i\in[d]$ full-rank sensing matrix for uncompressed mode
  \par $\mathcal{B}_c$ a $d$ mode tensor of measurements with side lengths $m_c$ 
    \par $\Phi_{i}\in\R^{m_c \times n}$ for $i\in[d]$ measurement matrices for core measurements
  \par $\vb{r}=(r,r,\dots,r)$ the rank of the HOSVD
  }
\Output{$\hat{\mathcal{X}} =  [\![ \mathcal{H}, Q_1,\dots Q_d ]\!] $ }

\Comment{Factor matrix recovery}
\For{$i\in[d]$}{
\Comment{Solve $n\times n$ linear system}
\par Solve $\Omega_{(i,i)} F_{i} = B_{i} $ for $F_{i}$

\Comment{Compute SVD and keep the $r$ leading singular vectors}
\par $ U,\Sigma, V^T \gets \text{SVD}(F_{i})$
\par $Q_i \gets U_{:,:r}$
}
\Comment{Core recovery}

\For{$i\in[d]$}{
\Comment{unfold measurements, mode-$i$ fibers are columns, size $m_c\times r^{(i-1)}m_c^{d-1-(i-1)}$}
\par $H \gets \text{unfold}(\mathcal{B}_c,m_c\times r^{(i-1)}m_c^{d-1-(i-1)},\text{mode}=i)$ 

\Comment{Undo the mode-$i$ measurement operator and factor's action by finding least square solution to $m_c\times r$ over-determined linear system }
\par Solve $\Phi_i Q_i H_{\text{new}} = H$ for $H_{new}$

\Comment{reshape the flattened partially solved core into a tensor}

\par $\mathcal{B}_c \gets \text{fold}(H_{\text{new}},\underbrace{r\times r\dots \times r}_{i}\underbrace{\times m_c\times \dots \times m_c}_{d-i},\text{mode}=i)$

\Comment{Each iteration $m_c\to r$ in $i$th mode}
} 
\par $\mathcal{H}\gets \mathcal{B}_c$
\end{algorithm}
With the measurement operators described and the canonical algorithm outlined, we can now state a representative main result. The full unabridged results are found in Section \ref{sec:main_putting_it_all_together}.

\begin{thm}[A Summary Main Result] \label{thm:sketchdimonepass} Let $\mathcal{X}$ be a $d$-mode tensor of side length $n$, $\epsilon \in (0,1), \delta \in (0,\frac{1}{3})$, and $r \in [n] := \{1, \dots, n \}$.  Denote an optimal rank $\vb{r} =(r,\dots,r) \in \mathbb{R}^d$ ($d \ge 2$) Tucker approximation of $\mathcal{X}$ by $[\![\mathcal{X}]\!]^{\rm opt}_{\bf r}$.  There exists a one-pass recovery algorithm (see Algorithm~\ref{alg:loo_one_pass_prime}) that outputs a Tucker factorization $\mathcal{X}_1 = [\![ \mathcal{H}, Q_1,\dots, Q_d ]\!]$ of $\mathcal{X}$ from leave-one-out linear measurements that will satisfy 
\begin{equation}
    \norm{\mathcal{X} - \mathcal{X}_1}{2}  \leq  (1+e^{\epsilon}) \sqrt{ \frac{d(1 + \epsilon)}{1 - \epsilon}} \norm{   \mathcal{X} - [\![\mathcal{X}]\!]^{\rm opt}_{\bf r}}{2}
\end{equation}
with probability at least $1-\delta$.  The total number of linear measurements the algorithm will use is bounded above by $m'(r,d,\epsilon,n, \delta) =  \left[nd + \frac{C r d^{2}}{\epsilon^{2}} \ln \left(\frac{d n^{d}}{\delta} \right) \right] \left( \frac{C r d^{2}}{\epsilon^{2}} \ln \left(\frac{d n^{d}}{\delta}\right) \right)^{d-1}$, where $C > 0$ is an absolute constant. Furthermore, the recovery algorithm (Algorithm~\ref{alg:loo_one_pass_prime}) runs in time $o(n^d)$ for large $n \gg r$.
\end{thm}
\begin{rmk}
    More generally and more exactly, the number of measurements is always $ndm^{d-1} + m_c^{d}$ where $m$ and $m_c$ are as in Theorem~\ref{thm:onepass_kron} for Kronecker structured measurements, or  $d \tilde{m} n+m_c^d$ where $\tilde{m}$ and $m_c$ are as in Theorem~\ref{thm:onepass_khat} in the case of Khatri-Rao structured measurements.  
\end{rmk}

\begin{proof}
Combine Theorem~\ref{thm:onepass_kron} (or Theorem~\ref{thm:onepass_khat}) with Theorem \ref{thm:kron_runtime} and Lemma~\ref{lem:QuasiOptimalfromSVDs}.
\end{proof}

Although our bound almost certainly not tight, overall we demonstrate the procedure has errors that depend on the number of degrees of freedom of the sketched tensor, not its full size; as is typical in similar compressive setting scenarios. Furthermore, in our numerical experiments we demonstrate some of the main trade-offs involved when choosing between Kronecker or Khatri-Rao constructed sketches in terms of approximation error and performance. 

\section{Preliminaries}

The following preliminaries will provide the necessary definitions and basic results needed to conduct our probabilistic error analysis of the various leave-one-out measurements and Algorithm \ref{alg:loo_one_pass_prime}. We will rely heavily on some definitions in randomized numerical linear algebra and compressive sensing. 

\subsection{Randomized Numerical Linear Algebra Preliminaries}
\label{sec:RandNumAlgReview}

\begin{definition}[\textbf{$(\epsilon,\delta, p)$-Johnson-Lindenstrauss (JL) property}] \label{def:jl} 
Let $\epsilon > 0$, $\delta \in (0,1)$, and $p \in \mathbbm{N}$.  A random matrix $\Omega \in \R^{m \times n}$ has the $(\epsilon,\delta, p)$-Johnson-Lindenstrauss (JL) property\footnote{Formally, a distribution over $m \times n$ matrices has the JL property if a matrix selected from this distribution satisfies \eqref{equ:JLproperty} with probability at least $1 - \delta$. For brevity, here and in the next similar cases, the term ``random matrix" will refer to a distribution over matrices.} for an arbitrary set $S \subset \R^n$ with cardinality at most $p$ if it satisfies
\begin{equation}
(1-\epsilon) \norm{\vb{x}}{2}^2 \leq \|\Omega\vb{x}\|_2^2 \leq (1+\epsilon) \norm{\vb{x}}{2}^2~~\textrm{for all } \vb{x}\in S
\label{equ:JLproperty}
\end{equation}
with probability at least $1-\delta$.

\end{definition}
For all random matrix distributions discussed in this work, only an upper-bound on the cardinality of the set $S$ appearing in Definition~\ref{equ:JLproperty} is required in order to ensure with high probability that a given realization satisfies \eqref{equ:JLproperty}.  No other property of the set $S$ to be embedded will be required to define, generate, or apply $\Omega$. Such random matrices are referred to as \textit{data oblivious}, or simply \textit{oblivious}.  The following variant of the famous Johnson-Lindenstrauss Lemma \cite{johnson1984extensions} demonstrates the existence of random matrices with the oblivious JL property.

\begin{thm}[sub-gaussian random matrices have the JL property]\label{thm:subgisJL}
Let $S\subset\R^n$ be an arbitrary finite subset of $\R^n$. Let $\delta,\epsilon \in(0,1)$. Finally, let $\Omega \in \R^{m\times n}$ be a matrix with independent, mean zero, variance $m^{-1}$, sub-gaussian entries all admitting the same parameter $c \in \R^+$.\footnote{See, e.g., Remark 7.25 in \cite{Foucart2013} for details regarding sub-gaussian random variables are their parameters.} Then
\[
(1-\epsilon)\|\vb{x}\|_2^2 \leq \|\Omega \vb{x} \|_2^2 \leq (1+\epsilon)\|\vb{x}\|_2^2
\] 
will hold simultaneously for all $\vb{x}\in S$ with probability at least $1-\delta$, provided that 
\begin{equation}
m\geq \frac{C}{\epsilon^2} \ln\left(\frac{\abs{S}}{\delta}\right),
\end{equation}
where $C \leq 8c(16c+1)$ is an absolute constant that only depends on the sub-gaussian parameter $c$.
\end{thm}
\begin{proof} See, e.g., Lemma 9.35 in \cite{Foucart2013}.
\end{proof}

Next, we define a similar property for an infinite, yet rank constrained, set.

\subsubsection{Oblivious Subspace Embedding (OSE) Results}

In this section we will review so-called Oblivious Subspace Embedding property, and describe how to construct them from any random matrix distribution with the JL property.

\begin{definition} \label{def:ose}
[\textbf{$(\epsilon,\delta, r)$-OSE property}]  Let $\epsilon>0$, $\delta\in(0,1)$, and $r \in \mathbbm{N}$.  Fix an arbitrary rank $r$ matrix $A\in\R^{n\times N}$.  A random matrix $\Omega \in \R^{m\times n}$ has the $(\epsilon,\delta, r)$-Oblivious Subspace Embedding (OSE) property for (the column space of) $A$ if it satisfies
\begin{equation}
(1-\epsilon)\norm{A\vb{x}}{2}^2 \leq \norm{\Omega A \vb{x}}{2}^2 \leq (1+\epsilon)\norm{A\vb{x}}{2}^2 ~~\textrm{for all } \vb{x}\in\R^{N}
\label{equ:JLsubspaceEmbedding}
\end{equation}
with probability at least $1 - \delta$.
\end{definition}

Note that if $Q \in\R^{n\times r}$ is an orthonormal basis for the column space of a rank $r$ matrix $A\in\R^{n\times N}$, then the images of $A$ and $Q$ coincide, i.e. $\left\{Q\vb{y} ~|~ \forall \vb{y} \in \R^r \right\} = \left\{A\vb{x} ~|~ \forall \vb{x} \in \R^N \right\}$. Furthermore, since $Q$ has orthonormal columns so that $\norm{Q\vb{y}}{2} = \norm{\vb{y}}{2}$ holds, \eqref{equ:JLsubspaceEmbedding} is equivalent to 
\begin{equation}
(1-\epsilon)\norm{\vb{y}}{2}^2 \leq \norm{\Omega Q\vb{y}}{2}^2 \leq (1+\epsilon)\norm{\vb{y}}{2}^2 ~~ \textrm{for all } \vb{y} \in \R^r
\label{equ:JLsubspaceEmbedding2}
\end{equation}
holding in this case.  Below we will use the equivalence of \eqref{equ:JLsubspaceEmbedding} and \eqref{equ:JLsubspaceEmbedding2} often by regularly constructing matrices satisfying \eqref{equ:JLsubspaceEmbedding} by instead constructing matrices satisfying \eqref{equ:JLsubspaceEmbedding2}.


The next result shows that random matrices with the JL property also have the OSE property.  It is a standard result in the compressive sensing and randomized numerical linear algebra literature (see, e.g., \cite[Lemma 5.1]{baraniuk2007simple} or \cite[Lemma 10]{sarlos2006improved}) that can be proven using an $\epsilon$-cover of the appropriate column space.  

\begin{lem}[Subspace embeddings from finite embeddings via a cover]
	\label{lem:subspaceembed}
	 Fix $\epsilon\in(0,1)$. Let $\mathcal{L}_{\mathcal{B}}^r \subset \R^n$ be the $r$-dimensional subspace of $\R^n$ spanned by an orthonormal basis $\mathcal{B}$, and define $S_{\mathcal{B}}^r := \left\{ \frac{\x}{\| \x \|_2} ~\big|~ \x \in 
\mathcal{L}_{\mathcal{B}}^r \setminus \{ {\bf 0} \} \right\}$.  Furthermore, let $C \subset S_{\mathcal{B}}^r $ be a minimal $\left(\frac{\epsilon}{16}\right)$-cover
 of $S_{\mathcal{B}}^r \subset \mathcal{L}_{\mathcal{B}}^r$. Then if $\Omega \in \C^{m\times n}$ satisfies \eqref{equ:JLproperty} with $S \leftarrow C$ and $\epsilon \leftarrow \frac{\epsilon}{2}$, it will also satisfy
 \begin{equation}
 \label{eqn:subspaceineq}
 	(1-\epsilon)\|\vb{x}\|_2^2 \leq \|\Omega\vb{x}\|_2^2 \leq (1+\epsilon)\|\vb{x}\|_2^2 ~~\forall \vb{x}\in\mathcal{L}_{\mathcal{B}}^r.
 \end{equation}
\end{lem}
\begin{proof}
See, e.g., Lemma 3 in \cite{iwen2021lower} for this version.
\end{proof}

Using Lemma~\ref{lem:subspaceembed} one can now easily prove the following corollary of Theorem~\ref{thm:subgisJL} which demonstrates the existence of random matrices with the OSE property.

\begin{corollary}
\label{cor:subgOSE} A random matrix $\Omega\in\R^{m\times n}$ with mean zero, variance $m^{-1}$, independent sub-gaussian entries has the $(\epsilon,\delta, r)$-OSE property for an arbitrary rank $r$ matrix $A \in \R^{n \times N}$ provided that
\[
	 	m \geq C \frac{r}{\epsilon^2} \ln\left(\frac{C'}{ \epsilon \delta}\right) \geq \frac{C}{\epsilon^2} \ln\left(\frac{(\frac{47}{\epsilon})^{r}}{\delta}\right),
 \]
where $C' > 0$ is an absolute constant, and $C > 0$ is an absolute constant that only depends on the sub-gaussian norms/parameters of $\Omega$'s entries.
\end{corollary}

\begin{proof}
As is done in Lemma \ref{lem:subspaceembed}, suppose $S$ is a minimal $\frac{\epsilon}{16}$-cover of the $r$ dimensional unit sphere in the column span of $A$. The cardinaltiy of this cover is bounded by $\left(\frac{47}{\epsilon}\right)^{r}$. Apply Theorem \ref{thm:subgisJL} to the finite set $S$. 
\end{proof}

Finally, the following lemma demonstrates that matrices satisfying \eqref{equ:JLsubspaceEmbedding} for $A$ also approximately preserve the Frobenius norm of $A$.  We present its proof here as an illustration of basic notation and techniques.

\begin{lem}
	\label{lem:subembedimpliesfrobound}
Suppose $\Omega\in\R^{m\times n}$ satisfies \eqref{equ:JLsubspaceEmbedding} for a rank $r$ matrix $A \in \R^{n \times N}$. Then,
	\begin{align*}
	    \abs{\norm{A}{F}^2 - \norm{\Omega A}{F}^2 } \leq \epsilon \norm{A}{F}^2.
	\end{align*}
\end{lem}
\begin{proof}
Let $\vb{e}_i \in\R^{N}$ denote the standard basis vector where the $i$-th entry is $1$, and all others are zero (i.e., $\vb{e}_i$ is the $i$-th column of the $N \times N$ identity matrix $I_N$). Similarly let $\vb{b}_i$ denote the $i$-th column of any given matrix $B$, and set $\tilde{A} := \Omega A$. By \eqref{equ:JLsubspaceEmbedding}, we conclude that for all $i\in[N]$ we have
\[
 \abs{\norm{A\vb{e}_i}{2}^2 - \norm{\Omega A \vb{e}_i}{2}^2 } \leq \epsilon \norm{A\vb{e}_i}{2}^2.
\]
To establish the desired result, we represent the squared Frobenius norm of a matrix as the sum of the squared $\ell_2$-norms of its columns. Doing so, we see that 
\begin{align*}
     \abs{\norm{A}{F}^2 - \norm{\Omega A}{F}^2 } &= \abs{\sum_{i=1}^N \left(\norm{\vb{a}_i}{2}^2 - \norm{\tilde{\vb{a}}_i}{2}^2\right)} 
     ~\leq~ \sum_{i=1}^N \abs{ \norm{\vb{a}_i}{2}^2 - \norm{\tilde{\vb{a}}_i}{2}^2} \\
    &= \sum_{i=1}^N \abs{ \norm{A\vb{e}_i}{2}^2 - \norm{(\Omega A)\vb{e}_i}{2}^2}
    ~\leq~ \epsilon \sum_{i=1}^N \norm{A\vb{e}_i}{2}^2
    ~=~ \epsilon \norm{A}{F}^2.
\end{align*}
\end{proof}

We will now define a property of random matrices related to fast approximate matrix multiplication.  

\subsubsection{The Approximate Matrix Multiplication (AMM) Property}

First proposed in \cite[Lemma 6]{sarlos2006improved} (see also, e.g., \cite{magen2011low} for other variants), the approximate matrix multiplication property will be crucial to our analysis below.

\begin{definition}[\textbf{$(\epsilon,\delta)$-AMM property}] \label{def:amm} Let $\epsilon > 0$ and $\delta \in (0,1)$.  A random matrix $\Omega\in\R^{m\times n}$ satisfies the $(\epsilon,\delta)$-Approximate Matrix Multiplication property for two arbitrary matrices $A\in\R^{p\times n}$ and $B\in\R^{n\times q}$ if
\begin{equation}
\label{equ:AMMdetprop}
\|A\Omega^T \Omega B - AB\|_F \leq  \epsilon \|A\|_F \|B\|_F
\end{equation}
holds with probability at least $1-\delta$. 
\end{definition}

The following lemma can be used to construct random matrices with the AMM property from random matrices with the JL property.  A slightly generalized version is proven in Appendix~\ref{appendix:PCPproofs} for the sake of completeness.

\begin{lem}[The JL property provides the AMM property] \label{lem:fastmatmulbyJLReal}
Let $A\in\R^{p\times n}$ and $B\in\R^{n\times q}$.  There exists a finite set $S \subset \R^n$ with cardinality $|S| \leq 2(p + q)^2$ (determined entirely by $A$ and $B$) such that the following holds:   If a random matrix $\Omega\in\R^{m\times n}$ has the $(\epsilon/2, \delta, 2(p + q)^2 )$-JL property for $S$, then $\Omega$ will also have the $(\epsilon, \delta)$-AMM property for $A$ and $B$.
\end{lem} 

\begin{proof}
Combine Lemma~\ref{lem:fastmatmulbyJLComplex} with Remark~\ref{Rem:RealPolarizationHelpsalittle}.
\end{proof}

Using Lemma~\ref{lem:fastmatmulbyJLReal} one can now prove the following corollary of Theorem~\ref{thm:subgisJL} which demonstrates the existence of random matrices with the AMM property for any two fixed matrices.

\begin{corollary}
\label{cor:AMMsetexist} Fix $A\in\R^{p\times n}$ and $B\in\R^{n\times q}$.  A random matrix $\Omega\in\R^{m\times n}$ with mean zero, variance $\frac{1}{m}$, independent sub-gaussian entries will have the $(\epsilon,\delta)$-Approximate Matrix Multiplication property for $A$ and $B$ provided that
\[
	 	m\geq \frac{C}{\epsilon^2} \ln\left(\frac{ 2(p + q)^2}{\delta} \right),
 \]
where $C > 0$ is an absolute constant that only depends on the sub-gaussian norms/parameters of $\Omega$'s entries.
\end{corollary}

\begin{proof}
Apply Theorem \ref{thm:subgisJL} to the finite set $S$ guaranteed by Lemma~\ref{lem:fastmatmulbyJLReal}. 
\end{proof}

We now present a capstone definition for this section that will be useful in our analysis of the general measurement ensembles considered herein.

\subsubsection{Projection Cost Preserving (PCP) Sketches}

The following property first appeared in the form below in  \cite[Definition 1]{cohen2015dimensionality} (see also, however, \cite[Definition 13]{feldman2020turning} for the statement of an equivalent property in a different form that appeared earlier).

\begin{definition}[\textbf{A $(\epsilon,c,r)$-Projection Cost Preserving (PCP) sketch}] \label{def:pcp} Let $\epsilon, c > 0$, and $r \in \mathbbm{N}$.  A matrix $\tilde{X}\in \R^{n\times m}$ is a $(\epsilon,c,r)$-PCP sketch of $X\in\R^{n \times N}$ if for all orthogonal projection matrices $P\in\R^{n\times n}$ with rank at most $r$,
\begin{equation}\label{eq:pcp}
(1-\epsilon) \norm{X - PX}{F}^2 \leq \norm{\tilde{X} - P\tilde{X}}{F}^2 + c \leq (1+\epsilon) \norm{X - PX}{F}^2 
\end{equation}
holds.
\end{definition}

The next lemma can be used to construct random matrices that are PCP sketches of a given matrix $X$ with high probability.  Before the lemma can be stated, however, we will need one additional definition.

\begin{definition}[Head-Tail Split] \label{def:headtail} For any $A\in\R^{m\times n}$, we can split $A$ into his leading $r$-term and its tail $(n-r)$-term Singular Value Decomposition (SVD) components. That is, consider the SVD of $A = U\Sigma V^T$.  For any $r \leq \text{rank}(A)$, let $U_{r} \in \R^{m\times r}$ and $V_{r}\in\R^{n\times r}$ denote the first $r$ columns of $U \in \R^{m \times m}$ and $V \in \R^{n \times n}$, respectively. We then define $A_{r} := U_{r} U_{r}^T A = A V_{r} V_{r}^T$ to be $A$'s best rank $r$ approximation with respect to $\| \cdot \|_{F}$. Furthermore, we denote the tail term by $A_{\setminus r} := A - A_{r}$. 
\end{definition}

One can now see that random matrices with the OSE and AMM properties for matrices derived from $X \in \R^{n \times N}$ will also be PCP sketches of $X$ with high probability.  Variants of the following result are proven in \cite{Chowdhury2019, Musco2020}.  We include the proof in Appendix~\ref{Appsec:FixingMusco2020} for the sake of completeness.

\begin{thm}[Projection-Cost-Preservation via the AMM and OSE properties]
\label{thm:musc2}
Let $X\in \R^{n\times N}$ of rank $\tilde{r} \leq \min \{ n, N \}$ have the full SVD $X = U\Sigma V^T$, and let $V_{r'}\in\R^{N\times r'}$ denote the first $r'$ columns of $V \in \R^{N \times N}$ for all $r' \in [N]$.  Fix $r \in [n]$ and consider the head-tail split $X = X_r + X_{\setminus r}$.  If $\Omega \in\R^{m\times N}$ satisfies 
\begin{enumerate}
    \item subspace embedding property \eqref{equ:JLsubspaceEmbedding} with $\epsilon \leftarrow \frac{\epsilon}{3}$ for $A \leftarrow X_{r}^T$, \label{thm:musc2_assump1}
    \item approximate multiplication property \eqref{equ:AMMdetprop} with $\epsilon \leftarrow \frac{\epsilon}{6\sqrt{\min\{r,\tilde{r} \} } }$ for $A \leftarrow X_{\setminus r}$ and $B \leftarrow V_{\min\{r,\tilde{r} \} }$, \label{thm:musc2_assump2}
    \item JL property \eqref{equ:JLproperty} with $\epsilon \leftarrow \frac{\epsilon}{6}$ for $S \leftarrow \{$the $n$ columns of $X_{\setminus r}^T \}$, and \label{thm:musc2_assump3}
    \item approximate multiplication property \eqref{equ:AMMdetprop} with $\epsilon \leftarrow \frac{\epsilon}{6\sqrt{r}}$ for $A \leftarrow X_{\setminus r}$ and $B \leftarrow X_{\setminus r}^T$, \label{thm:musc2_assump4}
    \end{enumerate}
then $\tilde{X} := X\Omega^T $ is an $(\epsilon,0,r)$-PCP sketch of $X$.
\end{thm}

\begin{proof} 
See Appendix~\ref{Appsec:FixingMusco2020}.
\end{proof}

The following lemma can be used to construct PCP sketches from random matrices with the JL property.

\begin{lem}[The JL property provides PCP sketches] \label{lem:PCPpropbyJLReal}
Let $X\in \R^{n\times N}$ have rank $\tilde{r} \leq \min \{ n, N \}$.  Fix $r \in [n]$.  There exist finite sets $S_1, S_2 \subset \R^N$ (determined entirely by $X$) with cardinalities $|S_1| \leq \left( \frac{141}{\epsilon} \right)^{\min \{ r, \tilde{r} \}}$ and $|S_2| \leq 16 n^2 + n$ such that the following holds:   If a random matrix $\Omega\in\R^{m\times N}$ has both the $\left(\frac{\epsilon}{6}, \frac{\delta}{2}, \left( \frac{141}{\epsilon} \right)^r \right)$-JL property for $S_1$ and the $\left(\frac{\epsilon}{6 \sqrt{r}}, \frac{\delta}{2}, 16 n^2 + n \right)$-JL property for $S_2$, then $X\Omega^T$ will be an $(\epsilon,0,r)$-PCP sketch of $X$ with probability at least $1 - \delta$.
\end{lem} 

\begin{proof}
To ensure property~\ref{thm:musc2_assump1} of Theorem~\ref{thm:musc2} we can appeal to Lemma~\ref{lem:subspaceembed} to see that $\Omega$ being an $(\epsilon/ 6)$-JL-embedding of a minimal $\left(\frac{\epsilon}{48} \right)$-cover of the at most $\min \{ r, \tilde{r}\}$-dimensional unit ball in the column space of $X_{r}^T$ will suffice.  Letting $S_1$ be this aforementioned $\left(\frac{\epsilon}{48} \right)$-cover, we can further see that $|S_1| \leq (141 / \epsilon)^{\min \{ r, \tilde{r} \}}$ by the proof of Corollary~\ref{cor:subgOSE}.  Hence, if $\Omega\in\R^{m\times N}$ has the $\left(\frac{\epsilon}{6}, \frac{\delta}{2}, \left( \frac{141}{\epsilon} \right)^r \right)$-JL property for $S_1$, then property~\ref{thm:musc2_assump1} of Theorem~\ref{thm:musc2} will be satisfied with with probability at least $1 - \frac{\delta}{2}$.

Applying Lemma~\ref{lem:fastmatmulbyJLReal} one can see that there exist sets $S_2', S_2'' \subset \mathbbm{R}^N$ with $| S_2' | \leq 2( n +\min \{ r, \tilde{r}\} )^2 \leq 8n^2$ and $| S_2'' | \leq 2( n + n)^2 = 8n^2$ such that an $(\epsilon/ 6 \sqrt{r})$-JL-embedding of $S_2' \cup S_2''$ will satisfy both properties~\ref{thm:musc2_assump2} and~\ref{thm:musc2_assump4} of Theorem~\ref{thm:musc2}.  Hence, since $r \geq 1$, we can see that an $(\epsilon/ 6 \sqrt{r})$-JL-embedding of $S_2 := S_2' \cup S_2'' \cup S$ will satisfy Theorem~\ref{thm:musc2} properties~\ref{thm:musc2_assump2} -- \ref{thm:musc2_assump4}, where $S$ is defined as per property~\ref{thm:musc2_assump3}.  Noting that $|S_2| \leq |S_2'| + |S_2''| + |S| \leq 16n^2 + n$, we can now see that $\Omega$ will satisfy all of Theorem~\ref{thm:musc2}'s properties~\ref{thm:musc2_assump2} -- \ref{thm:musc2_assump4} with probability at least $1 - \frac{\delta}{2}$ if it has the $\left(\frac{\epsilon}{6 \sqrt{r}}, \frac{\delta}{2}, 16n^2 + n \right)$-JL property for $S_2$.

Concluding, the prior two paragraphs in combination with the union bound imply that all of Theorem~\ref{thm:musc2}'s properties~\ref{thm:musc2_assump1} -- \ref{thm:musc2_assump4} will hold with probability at least $1 - \delta$ if $\Omega$ has both the $\left(\frac{\epsilon}{6}, \frac{\delta}{2}, \left( \frac{141}{\epsilon} \right)^r \right)$-JL property for $S_1$ and the $\left(\frac{\epsilon}{6 \sqrt{r}}, \frac{\delta}{2}, 16 n^2 + n \right)$-JL property for $S_2$.  An application of Theorem~\ref{thm:musc2} now finishes the proof.
\end{proof}

Using Lemma~\ref{lem:PCPpropbyJLReal} one can now prove the following corollary of Theorem~\ref{thm:subgisJL} which demonstrates the existence of a PCP sketch for any fixed matrix $X$.

\begin{corollary}
\label{cor:PCPsetexist} 
Fix $X\in \R^{n\times N}$ and $r \in [n]$.  Let $\Omega\in\R^{m\times N}$ be a random matrix with mean zero, variance $\frac{1}{m}$, independent sub-gaussian entries.  Then, $X\Omega^T$ will be an $(\epsilon,0,r)$-PCP sketch of $X$ with probability at least $1 - \delta$ provided that
\[
	 	m\geq C \frac{r}{\epsilon^2} \max \left\{ \ln\left(\frac{C_1}{\epsilon \delta}\right), \ln\left(\frac{C_2 n}{\delta}\right) \right\},
 \]
where $C_1, C_2 > 0$ are absolute constants, and $C > 0$ is an absolute constant that only depends on the sub-gaussian norms/parameters of $\Omega$'s entries.
\end{corollary}

\begin{proof}
Apply Theorem \ref{thm:subgisJL} to the finite set $S_1$ guaranteed by Lemma~\ref{lem:PCPpropbyJLReal} with $\epsilon \leftarrow \frac{\epsilon}{6}$ and $\delta \leftarrow \frac{\delta}{2}$.  Similarly, apply Theorem \ref{thm:subgisJL} to the finite set $S_2$ guaranteed by Lemma~\ref{lem:PCPpropbyJLReal} with $\epsilon \leftarrow \frac{\epsilon}{6 \sqrt{r}}$ and $\delta \leftarrow \frac{\delta}{2}$.  The result now follows by Lemma~\ref{lem:PCPpropbyJLReal}.
\end{proof}

We finish here by noting that Corollary~\ref{cor:PCPsetexist} is just one example of a PCP sketching result that one can prove with relative ease using Lemma~\ref{lem:PCPpropbyJLReal}.  Indeed, Lemma~\ref{lem:PCPpropbyJLReal} can be combined with other standard results concerning more structured matrices with the JL property (see, e.g., \cite{Krahmer2011,Ward2019,doi:10.1137/21M1432491,https://doi.org/10.48550/arxiv.2302.06165}) to produce similar theorems where $\Omega$ has a fast matrix-vector multiply.

\section{The Proofs of Our Main Results}
\label{sec:MainResBRKS}

The objective is to show that we can retrieve a accurate low rank Tucker approximation of a tensor $\mathcal{X}$ via Algorithm \ref{alg:loo_one_pass_prime} from valid sets of linear leave-one-out and core measurements as described in Section~\ref{sec:The_MEASUREMENTS!}. We will denote the approximation of $\mathcal{X}$ output by Algorithm \ref{alg:loo_one_pass_prime} as $\mathcal{X}_1$ here to emphasize that a single pass over the original data tensor $\mathcal{X}$ suffices in order to compute the linear input measurements required by Algorithm \ref{alg:loo_one_pass_prime}.  Hence, Algorithm \ref{alg:loo_one_pass_prime} in this setting is an example of a \textit{streaming algorithm} which doesn't need to store a copy the original uncompressed tensor $\mathcal{X}$ in memory in order to successfully approximate it.  Nonetheless, we wish to show that this algorithm still produces a quasi-optimal approximation of $\mathcal{X}$ in the sense of \eqref{equ:quasi_optimal} with high probability when given such highly compressed linear input measurements. Particular choices for measurement ensembles will make explicit the dependence on other parameters of the problem (these choices define specializations of Algorithm \ref{alg:loo_one_pass_prime}, and are summarized as Algorithm \ref{alg:kron_1pass} and \ref{alg:khat_1pass} in Appendix~\ref{appendix:algs}). 

More specifically, in this section we will show that with high probability for a given $d$-mode tensor $\mathcal{X}$, error tolerance $\epsilon > 0$, and chosen rank truncation parameter $r$, that
\begin{align} \label{equ:1pass_main_bound}
\norm{\mathcal{X} - \mathcal{X}_1}{2} 
      & \leq (1+e^{\epsilon}) \sqrt{ \frac{1+\epsilon}{1-\epsilon} \sum_{j=1}^d \Delta_{r,j}}
\end{align}
will hold whenever Algorithm \ref{alg:loo_one_pass_prime} is provided with sufficiently informative input measurements.  Here the $\Delta_{r,j}$ are defined as per Lemma~\ref{lem:QuasiOptimalfromSVDs}, and ``sufficiently informative" means that $(i)$ the leave-one-out measurements used to form $\mathcal{X}_1$ are of sufficient size to satisfy several PCP properties, and that $(ii)$ the core measurements used to form $\mathcal{X}_1$ are of sufficient size to ensure the accurate solution of least squares problems computed as part of Algorithm \ref{alg:loo_one_pass_prime}.  Finally, we note that one can see from \eqref{equ:1pass_main_bound} together with  Lemma~\ref{lem:QuasiOptimalfromSVDs} that Algorithm~\ref{alg:loo_one_pass_prime} will perfectly recover exactly low Tucker-rank tensors if the rank parameter $r$ is made sufficiently large.

In order to prove that \eqref{equ:1pass_main_bound} holds, we will need to also consider a weaker variant of Algorithm \ref{alg:loo_one_pass_prime} which permits a \textit{second} pass of the data tensor $\mathcal{X}$.  These weaker algorithms will first compute estimates of the factors of the tensor $Q_i$ as Algorithm \ref{alg:loo_one_pass_prime} does, but thereafter will be allowed to use those factors to operate on the original tensor $\mathcal{X}$ in order to approximate its core (see Algorithms \ref{alg:kron_2pass} and \ref{alg:khat_2pass}).  We denote the estimate of the tensor that results from this procedure as $\mathcal{X}_2$ to emphasize that it requires a second accesses to the original data during core recovery.  Note that such two-pass algorithms are of less practical value in the big data and compressive sensing settings since it is often not possible to directly access the data tensor again after the initial compressed measurements have been taken in these scenarios.  Nevertheless, this two-pass estimate will be extremely useful when proving \eqref{equ:1pass_main_bound}.  In particular, our proof will result from the following triangle inequality:
\begin{align}
\label{eqn:triangle_inequlaity_main}
    \norm{\mathcal{X} - \mathcal{X}_1}{2} &=     \norm{\mathcal{X} - \mathcal{X}_1 + \mathcal{X}_2 - \mathcal{X}_2}{2} \leq   \underbrace{\norm{\mathcal{X} - \mathcal{X}_2}{2}}_{\text{Term I}} + \underbrace{ \norm{\mathcal{X}_1 - \mathcal{X}_2}{2} }_{\text{Term II}}.
\end{align}

Bounding Term I in \eqref{eqn:triangle_inequlaity_main} will be the subject of Section~\ref{sec:bounding_x2}. As we shall see, bounding Term I is straightforward if we have that leave-one-out measurements imply various PCP properties; and so the main work of Sections~\ref{sec:bound_2pass_error_kron} and \ref{sec:MainResKhatri-Rao} will be to demonstrate how, for the structured choices of measurement operators considered herein, we can ensure that the PCP property is satisfied. Bounding Term II, on the other hand, will require us to apply a bound on the error incurred by solving sketched least square problems on a 
carefully partitioned re-expression of the Term II error.  That argument is the subject of Section~\ref{sec:bounding_termII}.  Finally, we combine our analysis of these two error terms along with particular choices for measurement operators to state the full versions of our main results in Section~\ref{sec:main_putting_it_all_together}.

\subsection{Bounding \texorpdfstring{$\norm{\mathcal{X} - \mathcal{X}_2}{2}$}{}}
\label{sec:bounding_x2}

In the two pass scenario, we first compute estimates for the factor matrices, $Q_i$  (see Algorithm \ref{alg:recover_factors}), using leave-one-out measurements $B_i$ for each $i\in[d]$. Then, using these factor matrices, we wish to solve 
\[\arg \min_{\mathcal{H}}\norm{\mathcal{X} - [\![ \mathcal{H}, Q_1, Q_2, \dots, Q_d]\!]}{2}.\]
One can see that the solution will be
\[
\mathcal{G} := \mathcal{X}\times_1 Q_1^T \times_2 Q_2^T \dots \times_d Q_d^T
\]
(see, e.g., \cite{de2000best}).  Let \begin{equation}\label{eq:x2-def
}\mathcal{X}_2 := \mathcal{G} \times_1 Q_1 \times_2 Q_2 \dots \times_d Q_d
\end{equation}
denote the estimate obtained from a two-pass recovery procedure (i.e., Algorithm \ref{alg:kron_2pass} or \ref{alg:khat_2pass}). Additionally, we note the following fact about modewise products (see, e.g., \cite[Lemma 1]{zare2018extension}):
\[\mathcal{X}\times_i A \times _i B = \mathcal{X} \times_i (BA).\]
As a result, if we are permitted a second pass over $\mathcal{X}$ to compute the core we have that
\begin{align*}
\mathcal{X}_2 &= \mathcal{G} \times_1 Q_1 \times_2 Q_2 \dots \times_d Q_d \\
 &= \left(\mathcal{X}\times_1 Q_1^T \times_2 Q_2^T \dots \times_d Q_d^T\right) \times_1 Q_1 \times_2 Q_2 \dots \times_d Q_d \\
 &=  \mathcal{X}\times_1 Q_1 Q_1^T \times_2 \dots \times_d Q_d Q_d^T
\end{align*}
Using this expression we can now bound the two pass error term  $\norm{\mathcal{X} - \mathcal{X}_2}{2}$.

\begin{thm}[Error bound for Two-Pass $\norm{\mathcal{X} - \mathcal{X}_2}{2}$] \label{thm:errorboundfactors} Suppose $ \tilde{X}_{[j]} :=  X_{[j]}\Omega^T_{-j} \in \R^{n \times m^{d-1}}$ are $(\epsilon,0,r)$-PCP sketches of $ X_{[j]}$ for each $j\in[d]$. 
If $Q_j \in \R^{n\times r}$ for $r \leq m$ are factor matrices obtained from Algorithm \ref{alg:recover_factors}, then 
\begin{align}
\label{eqn:factorerrorbound}
    \norm{\mathcal{X} - \mathcal{X}_2}{2} = \norm{\mathcal{X} - \mathcal{X}\times_1 Q_1 Q_1^T \times_2 \dots \times_d Q_d Q_d^T}{2}  
    \leq \sqrt{\frac{1+\epsilon}{1-\epsilon} \sum_{j=1}^d \Delta_{r,j} } ~.
\end{align}
\end{thm}
\begin{proof}
Since the $Q_i Q_i ^T$ are orthogonal projectors, we have by \cite[Theorem 5.1]{Vann2012} that 
\begin{equation} \label{eqn:pyth_like_upper}
    \norm{\mathcal{X} - \mathcal{X}\times_1 Q_1 Q_1^T \times_2 \dots \times_d Q_d Q_d^T}{2}^2 \leq \sum_{j=1}^d \norm{\mathcal{X} - \mathcal{X} \times_j Q_j Q_j^T }{2}^2.
\end{equation}
From Algorithm \ref{alg:recover_factors}  (recalling that $\Omega_{(j,j)}^{-1} B_j =  X_{[j]}\Omega^T_{-j} = \tilde{X}_{[j]}$), we have that the $Q_j$'s are the best rank-$r$ approximations for their respective sketched problems, since
\begin{align*}
    Q_j 
    &= \argmin_{\substack{rank(Q) =r \\ Q^T Q = I_{r}}}\norm{\tilde{X}_{[j]} -QQ^T \tilde{X}_{[j]} }{F}
\end{align*}
by the Eckart–Young Theorem.

Now suppose that each $U_j\in\R^{n \times r}$ forms an optimal rank $r$ approximation to $X_{[j]}$ in the sense that
\[
U_j= \argmin_{\substack{rank(U) =r \\ U^T U = I_{r}}} \norm{X_{[j]} - U U^T X_{[j]}}{F}.
\]
By the hypothesis that $\tilde{X}_{[j]}$ is a $(\epsilon,0,r)$-PCP sketch of $X_{[j]}$, we have that
\begin{align*}
 (1-\epsilon) \norm{\mathcal{X} - \mathcal{X} \times_j Q_j Q_j^T }{2}^2 &= (1-\epsilon)\norm{X_{[j]} - Q_j Q_j^T X_{[j]}}{F}^2 \\ &\leq \norm{\tilde{X}_{[j]} -Q_jQ_j^T \tilde{X}_j }{F}^2 \\
  &\leq \norm{\tilde{X}_{[j]} -U_j U_j^T \tilde{X}_{[j]} }{F}^2 \\
  &\leq (1+\epsilon)\norm{X_{[j]} -U_j U_j^T X_{[j]} }{F}^2 = (1 + \epsilon)\Delta_{r,j},
\end{align*}
 where we have used the definition of $(\epsilon,0,r)$-PCP sketches in the first and third inequalities.
After a rearrangement of terms, substituting the above into \eqref{eqn:pyth_like_upper} now yields the inequality in \eqref{eqn:factorerrorbound}. 
\end{proof}

We have now established in Theorem~\ref{thm:errorboundfactors} that we have a quasi-optimal error bound for Term I in \eqref{eqn:triangle_inequlaity_main} whenever our leave-one-out measurement matrices $\Omega^T_{-j}$ yield $(\epsilon,0,r)$-PCP sketches of all $d$ unfoldings $X_j$. 
Next, we will demonstrate how to ensure that Kronecker structured and Khatri-Rao structured leave-one-out measurement matrices provide PCP sketches.

\subsection{PCP Sketches via Kronecker-Structured Leave-one-out Measurement Matrices }
\label{sec:bound_2pass_error_kron}

In this section we study when Kronecker-structured measurement matrices will provide the PCP property. To begin we will show that the JL and OSE properties are inherited under matrix direct sums and compositions.  These are useful facts because our overall leave-one-out matrices can be constructed using these operations. In particular, we will follow the example in the last paragraph of Section~\ref{sec:KronMeasIntro} and consider a matrix ${\Omega}_{-j} \in \R^{m^{d-1} \times n^{d-1}}$ defined as 

\begin{equation}\label{eq:omega_tilde}
\Omega_{-j} = \bigotimes_{\substack{i=1 \\ i\neq j}}^{d} \Omega_i = \prod_{i'=1}^{d-1} \tilde{\Omega}_{i'} \quad \text{ for } {\Omega}_{i} \in\R^{m\times n}
\end{equation}
where 
\begin{equation}\label{omega_tilde_j}\tilde{\Omega}_{i'} := \underbrace{I_n \otimes \dots \otimes I_n}_{d-1-i'} \otimes \,\,\Omega_{i_j(i')} \otimes \underbrace{I_m \otimes \dots \otimes I_{m}}_{i'-1} \, \in \, \R^{m^{i'}n^{d-1-i'}\times m^{i'-1} n^{d-i'}}
\end{equation}
for 
$$i_j (i') := \begin{cases}
    i' & \textrm{if}~ i' < j\\
    i' + 1 & \textrm{if}~ i' \geq j
\end{cases}.$$
Here, $I_n$ denotes an $n \times n$ identity matrix.  

The next three lemmas will be used to help show that $\Omega_{-j}$ as defined in \eqref{eq:omega_tilde} inherits both the JL and OSE properties from its component $\Omega_i$ matrices.  Having established this, we can then use, e.g., Lemma~\ref{lem:PCPpropbyJLReal} to prove PCP sketching results for such Kronecker-structured $\Omega_{-j}$.

\begin{lem} 
\label{lem:directsum1}
Suppose that $\Omega_1\in\R^{m_1 \times N_1}$ and $\Omega_2\in\R^{m_2 \times N_2}$ are two random matrices. Denote their matrix direct sum by $\Omega = \Omega_1 \oplus \Omega_2  \in \R^{(m_1 + m_2)\times (N_1 + N_2)}$.  Then,
\begin{enumerate}
    \item If $\Omega_1$ and $\Omega_2$ are $\left(\epsilon ,\delta_1,p\right)$ and $\left(\epsilon ,\delta_2,p\right)$-JLs respectively, then $\Omega$ is an $\left(\epsilon ,\delta_1 + \delta_2,p\right)$-JL.
    \item If $\Omega_1$ and $\Omega_2$ are $\left(\epsilon ,\delta_1,r\right)$ and $\left(\epsilon ,\delta_2,r\right)$-OSEs respectively, then $\Omega$ is an $\left(\epsilon,\delta_1 + \delta_2,r\right)$-OSE.
\end{enumerate}
\end{lem}
\begin{proof}
{\it Part 1.:}  Consider a set $ S \subset \R^{N_1 + N_2}$ with cardinality $p$. Let $\vb{z}\in S$.  Group the first $N_1$ coordinates of $\vb{z}$ into $\vb{x}\in\R^{N_1}$ and the last $N_2$ coordinates of $\vb{z}$ into $\vb{y}\in\R^{N_2}$. Observe that
\begin{align*}
    \norm{{\Omega} \vb{z}}{2}^2 &= \norm{\begin{bmatrix} 
\Omega_1 & 0 \\
0 & \Omega_2
\end{bmatrix}
\begin{bmatrix} 
\vb{x}\\
\vb{y}
\end{bmatrix}
}{2}^2=\norm{\Omega_1 \vb{x}}{2}^2 + \norm{\Omega_2 \vb{y}}{2}^2 \leq(1+\epsilon) \left(\norm{ \vb{x}}{2}^2 + \norm{ \vb{y}}{2}^2 \right)=(1+\epsilon)\norm{ \vb{z}}{2}^2 
\end{align*} 
will hold whenever both $\norm{\Omega_1 \vb{x}}{2}^2 \leq (1+\epsilon) \| \vb{x} \|_2^2$ and $\norm{\Omega_2 \vb{y}}{2}^2 \leq (1+\epsilon) \| \vb{y} \|_2^2$ hold.  The $(1 - \epsilon)$-distortion lower bound is similar.  As a result, we can use the union bound to see that ${\Omega}$ will have the $(\epsilon,\delta_1 + \delta_2,p)$-JL property.

{\it Part 2.:}  Suppose $X\in\R^{n \times (N_1 + N_2)}$ has rank $r$. Let $X_1$ and $X_2$ denote the sub-matrices of $X$ containing the first $N_1$ and last $N_2$ columns of $X$, respectively. Note that both $X_1$ and $X_2$ have at most rank $r$. 
Furthermore, note also that
\begin{align*}
    \norm{{\Omega}X^T \vb{y}}{2}^2 &= \norm{\left(
\begin{array}{c}
	\Omega_1 X_1^T  \\
	\Omega_2 X_2^T \\
\end{array}\right)\vb{y}}{2}^2\leq (1+\epsilon) \left(\norm{ X_1^T \vb{y}}{2}^2 +  \norm{ X_2^T \vb{y}}{2}^2 \right) = (1+\epsilon)\norm{ X^T \vb{y}}{2}^2.
\end{align*}
will hold for any arbitrary vector $\vb{y}\in\R^{n}$ whenever both $\| \Omega_1 X_1^T \vb{y} \|^2_2 \leq (1 + \epsilon) \| X_1^T \vb{y} \|^2_2$ and $\| \Omega_2 X_2^T \vb{y} \|^2_2 \leq (1 + \epsilon) \| X_2^T \vb{y} \|^2_2$ hold.  The $(1 - \epsilon)$-distortion lower bound is similar.  As a result, we can see that ${\Omega}$ will be an $(\epsilon,\delta_1 + \delta_2,r)$-OSE by the union bound.
\end{proof}

Note that there is no requirement that $\Omega_1$ and $\Omega_2$ need to be independent in Lemma \ref{lem:directsum1}. This is crucial for the next lemma, which will involve many copies of the same measurement matrix.

\begin{lem}[Direct Sums Inherit the OSE and JL Properties]\label{lem:directsums2} For some $i' \in [d-1]$, let $\tilde{\Omega}_{i'}$ be defined as in \eqref{omega_tilde_j}
 and set $\Omega := \Omega_{i_j(i')} \in\R^{m\times n}$. 
 \begin{enumerate}
     \item If $\Omega$ has the $\left(\epsilon,\frac{\delta}{m^{i'-1}n^{d-i'-1}}, r\right)$-OSE property, then $\tilde{\Omega}_{i'}$ will have the $(\epsilon,\delta,r)$-OSE property.
     \item If $\Omega$ has the $\left(\epsilon,\frac{\delta}{m^{i'-1} n^{d-i'-1}}, p\right)$-JL property, then $\tilde{\Omega}_{i'}$ will have the $(\epsilon,\delta,p)$-JL, property.
 \end{enumerate}
\end{lem}

\begin{proof}
First consider the rearrangement of $\tilde{\Omega}_{i'}$ in \eqref{omega_tilde_j} defined as follows 
\[\tilde{\Omega} :=  \underbrace{I_m \otimes I_m \otimes \dots \otimes I_m}_{i'-1} \otimes \underbrace{I_n \otimes I_n \otimes \dots \otimes I_n}_{d-1-i'} \otimes \,\, \Omega .\]
Note that the Kronecker product of two identity matrices is itself an identity matrix. Thus, we can rewrite this as simply 
\begin{equation*}
 \tilde{\Omega} = \underbrace{I_m \otimes I_m \otimes \dots \otimes I_m}_{i'-1} \otimes \underbrace{I_n \otimes I_n \otimes \dots \otimes I_n}_{d-1-i'} \otimes \,\, \Omega 
 =\begin{bmatrix} 
\Omega & 0 & \dots &0 \\
0 & \Omega & \dots &0 \\
0 & 0 & \ddots &0 \\
0 & 0 & \dots &\Omega \\
\end{bmatrix}.
\end{equation*}
That is, we have a block diagonal matrix with $\bar m = m^{i'-1}n^{d-i'-1}$ copies of $\Omega$ along its diagonal. Thus, if $\Omega$ has either the $(\epsilon,\delta/\bar{m},r)$-OSE or the $(\epsilon,\delta/\bar{m},p)$-JL property, repeated applications of Lemma \ref{lem:directsum1} will then establish the desired OSE or JL properity for $\tilde{\Omega}$. 

Now consider $\tilde{\Omega}_{i'}$ as in \eqref{omega_tilde_j}. There exist unitary (permutation) matrices $L$ and $R$ which interchange rows and columns such that 
\begin{align*}
L\tilde{\Omega}R &= L \left( \underbrace{I_m \otimes I_m \otimes \dots \otimes I_m}_{i'-1} \otimes \underbrace{I_n \otimes I_n \otimes \dots \otimes I_n}_{d-1-i'} \otimes \,\, \Omega \right) R\\ &= L \left( I_{m^{i'-1}n^{d-1-i'}} \otimes \Omega \right)R
= \underbrace{I_n \otimes I_n \otimes \dots \otimes I_n}_{d-1-i'} \otimes \,\, \Omega \otimes \underbrace{I_m \otimes I_m \otimes \dots \otimes I_{m}}_{i'-1} \\
&= I_{n^{d-1-i'}} \otimes \,\, \Omega \otimes \,\, I_{m^{i'-1}} = \tilde{\Omega}_{i'}.
\end{align*}
Noting that both the OSE and JL properties are invariant to unitary transformations of a given random matrix, one can now see that $\tilde{\Omega}_{i'} = L\tilde{\Omega}R$ will indeed have the same desired OSE or JL property as was established for $\tilde{\Omega}$. 
\end{proof}

Lemma~\ref{lem:directsums2} allows us to infer JL and OSE properties of the $\tilde{\Omega}_{i'}$ matrices in \eqref{omega_tilde_j} from the properties of the smaller random matrices $\Omega_i \in \R^{m \times n}$ appearing in \eqref{eq:omega_tilde}.  The next lemma will allow us to then use these inferred properties of the $\tilde{\Omega}_{i'}$ matrices to derive OSE and JL properties for $\Omega_{-j}$ from \eqref{eq:omega_tilde} in terms of the properties of its component $\Omega_i \in \R^{m \times n}$.

\begin{lem}[A Composition Lemma for the OSE and JL Properties]\label{lem:composition} Let $\epsilon \in (0,1)$ and $\tilde{\Omega}_{i'} \in \R^{m^{i'}n^{d-1-i'}\times m^{i'-1} n^{d-i'}}
$ 
for $i' \in[d-1]$. 
\begin{enumerate}
    \item If $\tilde{\Omega}_{i'}$ is an $\left( \frac{\epsilon}{2(d-1)},\frac{\delta}{d-1},r \right)$-OSE for all $i' \in[d-1]$, then $\tilde{\Omega}= \prod\limits_{i'=1}^{d-1} \tilde{\Omega}_{i'}$ is an $\left( \epsilon,\delta,r \right)$-OSE.
    \item If $\tilde{\Omega}_{i'}$ is an $\left(\frac{\epsilon}{2(d-1)},\frac{\delta}{d-1},p \right)$-JL for all $i' \in[d-1]$, then $\tilde{\Omega}= \prod\limits_{i'=1}^{d-1} \tilde{\Omega}_{i'}$ is an $\left( \epsilon,\delta,p \right)$-JL.
    \end{enumerate}
\end{lem}
\begin{proof}
{\it Part 1.:}  Let $Y^T \in\R^{ n^{d-1}\times n}$ be an arbitrary matrix of rank at most $r$. Denote $\tilde{Y}_{i'} =\left(\tilde{\Omega}_{i'} \dots \tilde{\Omega}_1\right) Y^T \in \R^{m^{i'} n^{d-1-i'} \times n}$ for $i' \in [d-1]$. Note that each $\tilde{Y}_{i'}$ has rank at most $r$. Fix some $\vb{z}\in\R^{n}$.  Suppose for the moment that \eqref{equ:JLsubspaceEmbedding} holds for each $i' \in [d-1]$ with $\Omega = \tilde{\Omega}_{i'}$, $A = \tilde{Y}_{i'-1}$, and $\vb{x} = \vb{z}$, we have that 
\begin{align*}
    \norm{\tilde{\Omega}Y^T \vb{z}}{2}^2 &= \norm{\tilde{\Omega}_{d-1} \left(\tilde{\Omega}_{d-2} \dots \tilde{\Omega}_1\right) Y^T \vb{z}}{2}^2\\
    &= \norm{\tilde{\Omega}_{d-1} \tilde{Y}_{d-2}\vb{z}}{2}^2\\
    &\leq \left(1+\frac{\epsilon}{2(d-1)}\right) \norm{\tilde{Y}_{d-2}\vb{z}}{2}^2 \\
    &= \left(1+\frac{\epsilon}{2(d-1)}\right) \norm{\tilde{\Omega}_{d-2}\tilde{Y}_{d-3}\vb{z}}{2}^2 \\
    &\vdots\\
    &\leq \left(1+\frac{\epsilon}{2(d-1)}\right)^{d-1} \norm{Y^T\vb{z}}{2}^2 \\
    &\leq \left(\frac{1}{1-\epsilon/2}\right) \norm{Y^T\vb{z}}{2}^2 \\
    &\leq \left(1+\epsilon \right) \norm{Y^T\vb{z}}{2}^2, 
\end{align*}
where we have used the general bound $(1+k/n)^n \leq e^{k} \leq (1-k)^{-1}$ for $k\in [0,1)$ in the second to last inequality.
Similarly, for a lower bound one can see that 
\begin{align*}
    \norm{\tilde{\Omega} Y^T \vb{z}}{2}^2 &\geq \left(1-\frac{\epsilon}{2(d-1)}\right)^{d-1} \norm{Y^T \vb{z}}{2}^2 \\ 
     &\geq (1-\epsilon) \norm{Y^T \vb{z}}{2}^2.
\end{align*}
Union bounding over the failure probability that \eqref{equ:JLsubspaceEmbedding} holds for each $i' \in [d-1]$ as supposed above now yields the desired result.

{\it Part 2.:}  An essentially identical arguments also applies to obtain the desired JL property result.
\end{proof}

We now have all the necessary results to show how the component maps $\Omega_i$ of a Kronecker structured measurement ensemble as per \eqref{eq:omega_tilde} can guarantee a Kronecker sketch with the projection cost preserving property.

\begin{thm}[Kronecker Products of JL matrices yield PCP Sketchs] \label{Thm:GeneralPCPbyKron}
    Let $\epsilon \in (0,1)$, $X\in \R^{n \times n^{d-1}}$ have rank $r \in [n]$, and ${\Omega}_{-j} \in \R^{m^{d-1} \times n^{d-1}}$ be defined as in \eqref{eq:omega_tilde} and \eqref{omega_tilde_j}.  Furthermore, suppose that the $\Omega_i \in \R^{m \times n}$ in \eqref{eq:omega_tilde} have both the 
    \begin{enumerate}
        \item $\left(\frac{\epsilon}{12(d-1)}, \frac{\delta}{2(d-1) n^{d-2}}, \left( \frac{141}{\epsilon} \right)^r \right)$-JL property, and the
        \item $\left(\frac{\epsilon}{12 \sqrt{r} (d-1)}, \frac{\delta}{2(d-1) n^{d-2}}, 16 n^2 + n \right)$-JL property
    \end{enumerate}
    for all $i' \in[d-1]$.  Then, $X\Omega_{-j}^T$ will be an $(\epsilon,0,r)$-PCP sketch of $X$ with probability at least $1 - \delta$.
\end{thm}

\begin{proof}
By Lemma~\ref{lem:PCPpropbyJLReal} we know that $X\Omega_{-j}^T$ will be an $(\epsilon,0,r)$-PCP sketch of $X$ with probability at least $1 - \delta$ if $\Omega_{-j}$ has both the $\left(\frac{\epsilon}{6}, \frac{\delta}{2}, \left( \frac{141}{\epsilon} \right)^r \right)$-JL property and the $\left(\frac{\epsilon}{6 \sqrt{r}}, \frac{\delta}{2}, 16 n^2 + n \right)$-JL property.  In fact, by Lemma~\ref{lem:composition} we can further see that it suffices to have the $\tilde{\Omega}_{i'}$ from \eqref{eq:omega_tilde} and \eqref{omega_tilde_j} have both the
\begin{enumerate}
    \item $\left(\frac{\epsilon}{12(d-1)}, \frac{\delta}{2(d-1)}, \left( \frac{141}{\epsilon} \right)^r \right)$-JL property, and the
    \item $\left(\frac{\epsilon}{12 \sqrt{r} (d-1)}, \frac{\delta}{2(d-1)}, 16 n^2 + n \right)$-JL property
\end{enumerate}
for all $i' \in[d-1]$.  Finally, looking now at Lemma~\ref{lem:directsums2} for each $i' \in[d-1]$ we can see that the assumed properties of the $\Omega_i \in \R^{m \times n}$ in \eqref{eq:omega_tilde} will guarantee both of these sufficient conditions.
\end{proof}

The following corollary of Theorem~\ref{Thm:GeneralPCPbyKron} guarantees a Kronecker sketch with the projection cost preserving property when the component matrices $\Omega_i$ in \eqref{eq:omega_tilde} are sub-gaussian random matrices.

\begin{corollary}[Kronecker Products of sub-gaussian Matrices Yield PCP Sketches]
\label{lem:measurementsdefinepcp}
Suppose $\mathcal{X}$ is a real valued $d$-mode tensor with side-lengths all equal to $n$. Let $\epsilon \in (0,1), \delta\in(0,1), r \in [n]$, $j\in [d]$. If ${\Omega}_{-j} =\bigotimes_{\substack{i=1 \\ i \neq j}}^{d} \Omega_i \in \mathbb{R}^{m^{d-1} \times n^{d-1}}$ defined as in \eqref{eq:omega_tilde} with random matrices $\Omega_i\in\R^{m\times n}$ having i.i.d centered variance $m^{-1}$, sub-gaussian entries
 such that
\[
m\geq \max \left\{\frac{C_1 r (d-1)^2}{\epsilon^2} \ln \left(\frac{n^{d}(d-1)}{\delta}\right),\frac{C_2 (d-1)^2 }{\epsilon^2} \ln \left(\left( \frac{141}{\epsilon} \right)^r \frac{n^{d-2} (d-1)}{\delta} \right)  \right\}
\]
for absolute constants $C_1, C_2 > 0$ then the sketched unfolding $\tilde{X}_{[j]} = X_{[j]} {\Omega}_{-j}^T \in \mathbb{R}^{n \times m^{d-1}}$ is an $(\epsilon,0,r)$-PCP sketch of $X_{[j]}$ with probability at least $1-\delta$. 
\end{corollary}

\begin{proof}
To obtain the first quantity maximized over we apply Theorem \ref{thm:subgisJL} with $\epsilon \leftarrow \frac{\epsilon}{12 \sqrt{r} (d-1)}$, $\delta \leftarrow \frac{\delta}{2(d-1) n^{d-2}}$, and $|S| \leftarrow 16 n^2 + n$.  Similarly, for the second quantity maximized over we apply Theorem \ref{thm:subgisJL} with $\epsilon \leftarrow \frac{\epsilon}{12(d-1)}$, $\delta \leftarrow \frac{\delta}{2(d-1) n^{d-2}}$, and $|S| \leftarrow \left( \frac{141}{\epsilon} \right)^r$.  The result now follows from Theorem~\ref{Thm:GeneralPCPbyKron}.
\end{proof}

\begin{rmk}
    Note that Theorem~\ref{Thm:GeneralPCPbyKron} is quite general, requiring only that the random matrices $\Omega_{i}$ in \eqref{eq:omega_tilde} should be drawn from some distribution having a couple JL properties.  As a result of this generality, its Corollary~\ref{lem:measurementsdefinepcp} concerning sub-gaussian component matrices turns out to be sub-optimal by (at least a) factor of $d$ in that setting.  To obtain a slightly sharper result in $d$ for sub-gaussian $\Omega_i$ we recommend replacing our implicit use of Lemma~\ref{lem:directsums2} in the proof of Corollary~\ref{lem:measurementsdefinepcp} (via Theorem~\ref{Thm:GeneralPCPbyKron}) with \cite[Lemma 14]{Woodruff2019} instead.
\end{rmk}

\subsection{PCP Sketches via Khatri-Rao Structured Leave-one-out Measurement Matrices}
\label{sec:MainResKhatri-Rao}

In this section we study how to ensure that Khatri-Rao structured leave-one-out measurement matrices will provide the PCP property.  
To start we will first show that random Khatri-Rao structured measurement maps, denoted in this section by 
$$\Omega_{-j} = \frac{1}{m} \Omega_1 \sbullet \Omega_2 \sbullet \dots \Omega_{j-1} \sbullet \Omega_{j+1}\sbullet \Omega_{d}~{\rm where}~\Omega_i \in \R^{m \times n} ~\forall i \in[d] \setminus \{j\},$$ will have the JL property whenever all their component matrices $\Omega_{i}$ have i.i.d. sub-gaussian entries.  Having established this, we can then use, e.g., Lemma~\ref{lem:PCPpropbyJLReal} to prove PCP sketching results for such Khatri-Rao Structured $\Omega_{-j}$.

\begin{thm} \label{thm:khatasJL}
Let $\epsilon > 0, 0 <\delta \leq e^{-2}$, and $\Omega_{-j} = \frac{1}{m} \Omega_1 \sbullet \Omega_2 \sbullet \dots \Omega_{j-1} \sbullet \Omega_{j+1}\sbullet \Omega_{d}$ where all the $\Omega_i \in \R^{m \times n}$ for $i\in[d] \setminus \{j\}$ have i.i.d. mean zero, variance one, sub-gaussian entries.  Then $\Omega_{-j}$ is an $(\epsilon, \delta, k)$-JL whenever
\[
m \geq C^{d-1}\max \left\{ \epsilon^{-2} \log \frac{k}{\delta} , \epsilon^{-1}\left(\log\frac{k}{\delta}\right)^{d-1} \right\}
\]for a constant $C \in \mathbbm{R}^+$ that depends only on the sub-gaussian norm of the i.i.d. $\Omega_i$-entries.
\end{thm} 


The proof of Theorem \ref{thm:khatasJL} largely follows the argument proposed in Section 2 of \cite{Ahle2019} concerning the so-called $p$-moment JL property of Khatri-Rao structured measurements.  
We note that Kronecker products of sub-gaussian vectors are not sub-gaussian in general, so the general idea is to use Markov's inequality for higher moments of the norm of $\Omega_{-j}\vb{y}$ with a  fixed $\vb{y} \in \mathbb{R}^n$ to obtain the desired result. To proceed with the argument, we will need the following two concentration results.  

\begin{lem} \label{lem: wooodruff}[Lemma 19 in \cite{Woodruff2019}]
Let $\mathcal{Y}$ be a $d-1$ mode tensors with side lengths of size $n$, $p\geq 1$, and $\Omega_j (i,:) \in \R^{n}$ for $j\in [d-1]$, $i\in[m]$ be independent random vectors each satisfying the Khintchine inequality $\norm{\langle \Omega_j (i,:) , \vb{y}\rangle }{L^p} \leq C_p \norm{\vb{y}}{2}$ for any vector $\vb{y}\in \R^{n}$ where $C_p$ a constant depending only on $p$. Then
\[
\norm{\langle \Omega_1 (i,:) \otimes  \Omega_2 (i,:) \otimes \dots \otimes \Omega_{d-1}(i,:) , \textrm{vec}(\mathcal{Y})\rangle }{L^p} \leq C_p^{d-1} \norm{\mathcal{Y}}{2}.
\]
\end{lem}

\begin{lem} \label{lem:latalya}[Corollary 2 in \cite{Latala1997}]
If $p\geq 2$ and $Z, Z_1,\dots,Z_m$ are i.i.d symmetric random variables then we have
\[
\norm{\sum_{i=1}^m Z_i}{L^p} \leq C \sup_{s \in \left[ \max\{2,\frac{p}{m}\}, p\right]} \left\{\frac{p}{s} \left(\frac{m}{p}\right)^{1/s} \norm{Z}{L^s} \right\}.
\]
Here $C > 0$ is an absolute constant. 
\end{lem}

In particular, we will utilize the following corollary of Lemma~\ref{lem:latalya}.

\begin{corollary}\label{cor:latalya}
    If under conditions of Lemma~\ref{lem:latalya}, in addition, we know that $\norm{Z}{L^s} \leq (Cs)^{d-1}$, then
$$
\norm{\frac{1}{m}\sum_{i=1}^m Z_i}{L^p} \leq C^{d-1} \max \left\{2^{d-1}\sqrt{\frac{p}{m}} , \frac{(pe)^{d-1}}{m}\right\}.
$$
\end{corollary}
\begin{proof}
    The proof of this corollary loosely follows the argument presented in \cite{Ahle2019}. Since $\norm{Z}{L^s} \leq (Cs)^{d-1}$, $Z_i \sim Z$, the expression over which we are taking the supremum in Lemma~\ref{lem:latalya} is a function whose derivative with respect to $s$ is non-decreasing. That is, the derivative 
\[
\frac{\partial }{\partial s} \left[\frac{p}{s} \left(\frac{m}{p}\right)^{1/s} s^{d-1} \right] = ps^{d-4}\left(\frac{m}{p}\right)^{1/2} \left((d-2)s - \log\frac{m}{p}\right)
\]
has at most a single root in the interval of interest at $s = \frac{\log \frac{m}{p}}{d-2}$. Noting the sign change at this root, we conclude that the maximum value must occur at the endpoints of the interval, and cannot occur at the critical point that is interior to the interval. Evaluating the function of interest at $s=2$ we obtain
$2^{d-2}\sqrt{mp}$; we will further upper-bound this by $2^{d-1}\sqrt{mp}$ in order to simplify analysis for the right endpoint. Additionally, since we are interested only in an upper bound, we need not evaluate the possible endpoint $s=\frac{p}{m}$, since if $2<\frac{p}{m}$, we are increasing the interval over which we are maximizing by instead considering $s=2$. 

We will now bound the expression at the right endpoint when $s=p \geq 2$. Clearly $(1/p)^{1/p} \leq 1$ thus
\begin{align*}
    \left(\frac{m}{p}\right)^{1/p} p^{d-1} \leq m^{1/p} p^{d-1}.
\end{align*}
If the function value at the right endpoint actually dominates $2^{d-1}  \sqrt{mp}$ (i.e., our upper bound of the function value at the left endpoint), we will have that $ m^{1/p} p^{d-1} \geq 2^{d-1}  \sqrt{mp}$ must hold. We will now use this assumption to remove the dependence on $m$ in our current upper bound for the function value at the right endpoint after noting that doing so will still yield a valid upper bound whenever the functions value at the right endpoint fails to already be bounded by $2^{d-1}  \sqrt{mp}$.  

Proceeding as planned, our assumption yields that
\begin{align*}
&m^{1/2 - 1/p} \leq \left(\frac{p}{2}\right)^{d-1} p^{-1/2}.
\end{align*}
Rearranging of terms, we get that
\begin{align*}
     m^{1/p} &\leq \left[ \left(\frac{p}{2}\right)^{d-1} p^{-1/2} \right]^{\frac{2}{p-2}} \\
     &= \left[  \left(\frac{p}{2}\right)^{\frac{2d-2}{p-2}} \right]  p^{\frac{-1}{p-2}}.
\end{align*}
The factor $p^{\frac{-1}{p-2}}$ is less than one. A tedious calculation reveals that $\left(\frac{p}{2}\right)^{\frac{2d-2}{p-2}} $ is decreasing for all $p\geq 2$, and therefore $m^{1/p}$ bounded by $\lim_{p\to 2} \left(\frac{p}{2}\right)^{\frac{2d-2}{p-2}} = e^{d-1}$. maximizing over our two upper bounds and then averaging over $m$, we obtain the desired inequality.
\end{proof}

Now we are ready to give a formal proof of Theorem~\ref{thm:khatasJL}.

\begin{proof}[Proof of Theorem~\ref{thm:khatasJL}:] 
 Note that the result is unchanged for any choice of mode to leave out, so we will shorten the notation and work with $\Omega := \Omega_{-j}$ within this proof. Let $K$ be the sub-gaussian norm of an entry in the $\Omega_j$'s. We aim to bound the probability that 
$$
 \left|\frac{1}{m} \norm{\Omega \vb{x}}{2}^2 -\norm{\vb{x}}{2}^2\right| \ge \epsilon \|\vb{x}\|_2^2 \quad \text{ for a fixed  } \vb{x} \in \R^{n^{d-1}}.
$$
Without loss of generality, assume $\|\vb{x}\|_2 = 1$. Furthermore, note that 
 \begin{equation}\label{eq:sos}
\frac{1}{m} \norm{\Omega \vb{x}}{2}^2 - 1 = \frac{1}{m}\sum_{i=1}^m \left[\langle\Omega (i,:), \vb{x}\rangle^2 -1\right].
\end{equation}

Now, since the entries of each $\Omega_{j}$ are i.i.d. mean zero and variance one sub-gaussian random variables, they satisfy Khintchine's inequality (see, e.g., \cite{Vershynin2018}) in the form
$$
 \norm{\langle \Omega_j (i,:), \vb{y} \rangle}{L^p} \leq C K\sqrt{p}\norm{\vb{y}}{2} \quad \text{ for any fixed  } \vb{y} \in \R^n.
 $$
 By Lemma~\ref{lem: wooodruff}, this implies that the rows of $\Omega$ satisfy a generalized Khintchine's inequality in the form
 \begin{equation} \label{eqn:genkhint_ineq1}
 \norm{\langle \Omega(i,:), \vb{x}\rangle}{L^p} = \norm{ \langle \Omega_1 (i,:) \tensor \Omega_2 (i,:) \tensor \dots \tensor \Omega_{d-1}(i,:), \vb{x} \rangle }{L^p} \leq (C' p)^{\frac{d-1}{2}}\norm{\vb{x}}{2},
\end{equation}
where $C'$ is a new constant that only depends on $K$.

To bound the $L^p$-norm of the sum in \eqref{eq:sos} we will now bound the $L^p$-norm of each summand.  Using the centering Lemma~\ref{lem:rawmomentsboundcentered} and continuing to estimate for one term we see that
\begin{align*}
    \norm{\langle \Omega (i,:), \vb{x}\rangle^2 -1}{L^p} &\leq 2 \norm{\langle \Omega (i,:) , \vb{x}\rangle^2}{L^p} = 2 \norm{\langle \Omega (i,:), \vb{x}\rangle }{L^{2p}}^2 \\
          &= 2 \norm{\langle \Omega_1 (i,:) \tensor \Omega_2 (i,:) \tensor \dots \tensor \Omega_{d-1}(i,:), \vb{x} \rangle}{L^{2p}}^2 \\
&\overset{\eqref{eqn:genkhint_ineq1}}{\leq} 2 ((C'2p)^{\frac{d-1}{2}})^2  \norm{\vb{x}}{2}^2 \leq (C'' p)^{d-1} \norm{\vb{x}}{2}^2 = (C'' p)^{d-1}.
\end{align*}
We would now like to apply Corollary~\ref{cor:latalya} to help bound the $L^p$-norm of the sum in \eqref{eq:sos}.  However, we need to symmetrize our random variables first.  Toward that end, define
\begin{equation}\label{zi}
Z_i =  \rho_i \left( \langle \Omega (i,:), \vb{x}\rangle^2 -1\right),
\end{equation}
where $\rho_i$ are i.i.d. Rademacher random variables.  Note that $\| Z_i \|_{L^p} = \| \left( \langle \Omega (i,:), \vb{x}\rangle^2 -1\right) \|_{L^p} \leq (C'' p)^{d-1}$.

Appealing now to Corollary~\ref{cor:latalya} we have that,
\begin{equation} \label{eqn:finallatala}
\norm{\frac{1}{m}\sum_{i=1}^m Z_i}{L^p} \leq (C'')^{d-1} \max \left\{2^{d-1}\sqrt{\frac{p}{m}} , \frac{(pe)^{d-1}}{m}\right\}
\end{equation}
for any $p \ge 2$. The upper bound in \cite[Lemma 6.3]{Talagrand2002} then further implies that 
\[
 \norm{  \frac{1}{m}\sum_{i=1}^m \left[\langle\Omega (i,:), \vb{x}\rangle^2 -1\right]}{L^p} \leq 2 \norm{\frac{1}{m}\sum_{i=1}^m Z_i}{L^p} \leq 2 (C'')^{d-1} \max \left\{2^{d-1}\sqrt{\frac{p}{m}} , \frac{(pe)^{d-1}}{m}\right\}.
\]
Employing Markov's inequality, we finally have that 
\begin{align*}   
    \pr\left\{ \left|\frac{1}{m}\sum_{i=1}^m \left[\langle\Omega (i,:), \vb{x}\rangle^2 -1\right] \right|\geq \epsilon \right\} &=     \pr\left\{ \abs{\frac{1}{m} \sum_{i=1}^m \left[\langle\Omega (i,:), \vb{x}\rangle^2 -1\right] }^p \geq \epsilon^p \right\} \\  
    &\le (C''')^{p(d-1)}\frac{\max \left\{(\frac{p}{m})^{p/2} , \frac{(pe)^{p(d-1)}}{m^p}\right\}}{\epsilon^p}.
\end{align*}
Taking $p = \log(k/\delta)$ and $m \geq \tilde{C}^{d-1}\max \left\{ \epsilon^{-2} \log \frac{k}{\delta} , \epsilon^{-1}\left(\log\frac{k}{\delta}\right)^{d-1} \right\}$, the last expression is upper bounded by $\delta/k$. 
Hence, $\Omega$ is an $(\epsilon,\delta,k)$-JL by the union bound over $k$ vectors.
\end{proof}

\begin{rmk} In order to employ Lemmas~\ref{lem: wooodruff} and \ref{lem:latalya}, it is necessary that the \emph{rows} $\Omega_j(i,:)$ have independent and identical distributions and that these rows satisfy Khintchine's inequality. Assuming that the matrices $\Omega_i$ all have i.i.d sub-gaussian \emph{entries} as we have done in Theorem~\ref{thm:khatasJL} implies both these necessary properties of the rows.  However, we note that more general (though perhaps less natural) assumptions will also suffice.  For example, the distributions of the i.i.d sub-gaussian entries of the $\Omega_i$ may also vary by column.
\end{rmk}
We can now use Theorem~\ref{thm:khatasJL} to derive row bounds that guarantee that our Khatri-Rao structured sub-gaussian measurement matrices will provide PCP sketches with high probability.

\begin{thm} \label{thm:khat_loo_is_pcp}
Let $\epsilon > 0, 0 <\delta \leq e^{-2}$, and $\mathcal{X}$ be a $d$ mode tensor with side-lengths equal to $n$.  Furthermore, suppose that $\Omega_{-j} := \frac{1}{m} \Omega_1 \sbullet \Omega_2 \sbullet \dots \Omega_{j-1} \sbullet \Omega_{j+1}\sbullet \dots \sbullet \Omega_{d}$ where the $\Omega_i \in\R^{m\times n}$ for $i\in[d] \setminus \{ j \}$ are as in Theorem \ref{thm:khatasJL} with 
\begin{equation}
\label{eqn:khatrowreq}
m\geq C^{d-1}  \max\left\{\epsilon^{-2} \log \frac{\left( \frac{141}{\epsilon}\right)^r}{\delta/2} ,~  \epsilon^{-1}\left(\log\frac{\left( \frac{141}{\epsilon}\right)^r}{ \delta/2}\right)^{d-1},~\frac{r}{\epsilon^{2}} \log \frac{17n^2}{\delta/2} ,~ \frac{r}{\epsilon}\left(\log\frac{17n^2}{ \delta/2}\right)^{d-1}
\right\}
\end{equation}
for a positive constant $C \in \mathbbm{R}^+$. Then, $\tilde{X}_{[j]} = X_{[j]}\Omega_{-j}^T$ will be an $(\epsilon,0,r)$-PCP sketch of $X_{[j]}$ with probability at least $1-\delta$.
\end{thm}


\begin{proof}  By Lemma~\ref{lem:PCPpropbyJLReal} we know that it suffices for the measurement matrix $\Omega_{-j}$ to have both the $\left(\frac{\epsilon}{6},\frac{\delta}{2}, \left( \frac{141}{\epsilon} \right)^r \right)$-JL property and $\left(\frac{\epsilon}{6 \sqrt{r}}, \frac{\delta}{2}, 16 n^2 + n \right)$-JL property.  Combining these two required JL properties with Theorem~\ref{thm:khatasJL} yields \eqref{eqn:khatrowreq} after combining and simplifying constants.
\end{proof}

We can now see that both Kronecker and Khatri-Rao structured measurements can satisfy the PCP property required by Theorem~\ref{thm:errorboundfactors}.  This demonstrates that such memory-efficient measurements may be used to bound the first term in \eqref{eqn:triangle_inequlaity_main} as desired.  Given this partial success, we will now turn our attention to the second term in \eqref{eqn:triangle_inequlaity_main}.

\subsection{Bounding \texorpdfstring{$\norm{\mathcal{X}_{1} - \mathcal{X}_2}{2}$}{}}
\label{sec:bounding_termII}

Next, we show how to bound Term II in \eqref{eqn:triangle_inequlaity_main}. Recall that $\mathcal{X}_1$ is the output of Algorithm \ref{alg:loo_one_pass_prime}, and that $\mathcal{X}_2$ is the tensor recovered by the two-pass algorithm consisting of the first ``Factor matrix recovery'' phase of Algorithm \ref{alg:loo_one_pass_prime} followed by the second pass core recovery procedure discussed in Section~\ref{sec:bounding_x2}.  That is, 
$\mathcal{X}_1 := [\![ \mathcal{H},Q_1, \dots, Q_d ]\!]$ is the single-pass estimate of the tensor $\mathcal{X}$ output by Algorithm~\ref{alg:loo_one_pass_prime}, and
$$
\mathcal{X}_2 = \mathcal{G} \times_1 Q_1 \times_2 \dots \times_d Q_d = \mathcal{X}\times_1 Q_1 Q_1^T \times_2 \dots \times_d Q_d Q_d^T,
$$
where $\mathcal{G}$ is the core estimate from the two-pass algorithm, and where the $Q_i \in \R^{n \times r}$ have $r$ orthonormal columns.

To begin, we note that the one-pass core $\mathcal{H}$ computed by Algorithm \ref{alg:loo_one_pass_prime} can be recovered from its input measurements by
\begin{equation}
\begin{aligned}[b]
    \mathcal{H} &= \mathcal{B}_c \times_1 (\Phi_1 Q_1)^{\dagger} \times_2 \dots \times_d (\Phi_d Q_d)^{\dagger} \\
    &= \left(\mathcal{X}\times_1 \Phi_1 \times_2 \dots \times_d \Phi_d \right)\times_1 (\Phi_1 Q_1)^{\dagger} \times_2 \dots \times_d (\Phi_d Q_d)^{\dagger}.
    \label{eqn:core_linear_solve}
\end{aligned}
\end{equation}
In addition, we note that the norm of the difference between the two estimates is the same as the norm of the difference of their cores since factor matrices have orthonormal columns.

\begin{lem}\label{lem:core_diffs}
In the notation outlined above,
$$
\norm{\mathcal{X}_1 - \mathcal{X}_2}{2} = \norm{\mathcal{H}  - \mathcal{G}}{2}.
$$
\end{lem}
\begin{proof}
\begin{align*}
\norm{\mathcal{X}_1 - \mathcal{X}_2}{2} &= \norm{ \mathcal{H} \times_1 Q_1 \times_2 \dots \times_d Q_d - \mathcal{G} \times_1 Q_1 \times_2 \dots \times_d Q_d}{2}\\
&= \norm{ \left( \mathcal{H}  - \mathcal{G} \right)\times_1 Q_1 \times_2 \dots \times_d Q_d}{2}\\
&= \norm{(Q_1 \otimes \dots\otimes Q_d) \text{vec}\left( \mathcal{H}  - \mathcal{G} \right)}{2}\\
&=  \norm{\mathcal{H}  - \mathcal{G}}{2}
\end{align*}
since $(Q_1 \otimes \dots\otimes Q_d)$ has orthonormal columns. 
\end{proof}

 In order to simplify the presentation of our culminating results we next state a definition for an Affine-Embedding property. In Lemma~\ref{lem:affineembdbnd} below we then describe how this property relates to the OSE and AMM properties. 
\begin{definition}[\textbf{$(\epsilon,r,N)$-AE property}] \label{def:affineemd}  Let $\epsilon>0$, and $r \in \mathbbm{N}$, Fix $Q\in \R^{n\times r}$ an arbitrary matrix with orthonormal columns and an arbitrary matrix $B\in \R^{n\times N}$. A matrix $\Phi \in \R^{m\times n}$ is an $(\epsilon,r,N)$-Affine Embedding (AE) for given matrices $Q$ and $B$ if it satisfies 
    \begin{equation} \label{equ:affineembd}
\norm{(\Phi Q)^{\dagger} \Phi B }{F} \leq (1+\epsilon) \norm{B}{F}.
\end{equation} 
\end{definition}

With this definition in hand, we are now able to prove the main theorem of this section.  It will allow us to relate Term II of \eqref{eqn:triangle_inequlaity_main} with Term I.

\begin{thm}\label{thm:error-between-two-passes}
Let $\mathcal{X}_2 = [\![\mathcal{G}, Q_1,\dots, Q_d]\!]$ denote the two-pass tensor estimate, and $\mathcal{X}_1 = [\![\mathcal{H}, Q_1,\dots, Q_d]\!]$ denote the single-pass tensor estimate for a $d$-mode tensor $\mathcal{X}$. Furthermore, let $\epsilon \in (0,1)$, and $\Phi_i \in \R^{m_{c} \times n}$ be a $(\epsilon/d,r, n^{i-1} r^{d-i})$-AE for the matrices $Q_i$ and $\left( X_{[i]} - (X_2)_{[i]}\right) \bigotimes_{j=1}^{i-1} I_n \bigotimes_{\substack{j=i+1}}^d \left(\Phi_j Q_j \right)^{\dagger}\Phi_j$ for each $i\in[d]$.
Then,
\begin{equation}
  \norm{\mathcal{X}_1 - \mathcal{X}_2}{2}  \leq e^{\epsilon} \norm{\mathcal{X} - \mathcal{X}_2}{2}. 
\end{equation}
\end{thm}

\begin{proof}
By Lemma~\ref{lem:core_diffs} it is enough to estimate the difference between the cores $\mathcal{G}$ and $\mathcal{H}$. He have that
\begin{align*}
    \mathcal{H} - \mathcal{G} &\overset{\eqref{eqn:core_linear_solve}}{=}\left(\mathcal{X}\times_1 \Phi_1 \times_2 \dots \times_d \Phi_d\right) \times_1 (\Phi_1 Q_1)^{\dagger}\times_2 \dots \times_d (\Phi_d Q_d)^{\dagger} - \mathcal{G} \\
     &= \left( (\mathcal{X}-\mathcal{X}_2)\times_1 \Phi_1 \times_2 \dots \times_d \Phi_d\right) \times_1 (\Phi_1 Q_1)^{\dagger} \times_2 \dots \times_d (\Phi_d Q_d)^{\dagger} 
        \\&\quad\quad+ \left(\mathcal{X}_2\times_1 \Phi_1 \times_2 \dots \times_d \Phi_d \right)\times_1 (\Phi_1 Q_1)^{\dagger} \times_2 \dots \times_d (\Phi_d Q_d)^{\dagger} - \mathcal{G} \\
        &= (\mathcal{X}-\mathcal{X}_2)\times_1 (\Phi_1 Q_1)^{\dagger} \Phi_1 \times_2 \dots \times_d (\Phi_d Q_d)^{\dagger}
        \\&\quad\quad+\mathcal{G}\times_1  (\Phi_1 Q_1)^{\dagger} \Phi_1 Q_1  \dots \times (\Phi_d Q_d)^{\dagger} \Phi_d Q_d  - \mathcal{G}\\
   &= (\mathcal{X}-\mathcal{X}_2)\times_1 (\Phi_1 Q_1)^{\dagger} \Phi_1 \times_2 \dots \times_d (\Phi_d Q_d)^{\dagger} \Phi_d.
\end{align*}
Now consider the following related mode-$i$ unfolding for $i\in[d]$, where $\left(\Phi_j Q_j\right)^{\dagger} \Phi_j$ for $j<i$ is replaced with an $n \times n$ identity matrix $I_n$,
\begin{equation}\label{eqn:unfolding_core_diff}
 X'_i :=  \left( X_{[i]} - (X_2)_{[i]}\right) \bigotimes_{j=1}^{i-1} I_n \bigotimes_{\substack{j=i+1}}^d \left(\Phi_j Q_j \right)^{\dagger}\Phi_j.
\end{equation}
Each $\Phi_i$ is an $(\epsilon/d,r,n^{i-1} r^{d-i})$-AE, where $Q \gets Q_i$ and $B \gets X'_i$, for $i=1,2,\dots, d$. Thus,
\begin{align*}
    \norm{\mathcal{H} - \mathcal{G}}{2} &= \norm{ \left(\Phi_1 Q_1\right)^{\dagger}\Phi_1 \left( X_{[1]} - (X_2)_{[1]}\right) \bigotimes_{j=2}^d \left(\Phi_j Q_j \right)^{\dagger}\Phi_j }{F}\\
        &\leq \left(1+\frac{\epsilon}{d}\right) \norm{ X'_1}{F}\\
        &~~~~~~\vdots \\
        &\leq \left(1+\frac{\epsilon}{d}\right)^d \norm{ X'_d}{F} \\
    &\leq e^{\epsilon} \norm{  \mathcal{X}-\mathcal{X}_2}{2},
\end{align*}
where we have used Definition~\ref{def:affineemd} $d$-times together with the bound $\left(1+\frac{\eps}{d}\right)^{d} \leq e^{\eps}$.
\end{proof}
 
\subsection{Putting it All Together with Row Bounds} \label{sec:main_putting_it_all_together}

In the previous subsections we have demonstrated that we can bound both error terms in \eqref{eqn:triangle_inequlaity_main} when the leave-one-out and core measurements satisfy certain embedding properties.  In particular, we have shown how the Johnson-Lindenstrauss property (Definition~\ref{def:jl}) can be used to obtain both the Oblivious Subspace Embedding (Definition~\ref{def:ose}) and an Approximate Matrix Multiplication properties (Definition~\ref{def:amm}). These two properties are then used with compositions and direct sums to show how a tensor unfolding will satisfy a Projection Cost Preserving property (Definition~\ref{def:pcp}), which was the essential ingredient in bounding the error term $\norm{\mathcal{X}-\mathcal{X}_2}{2}$. Next, we introduced Affine-Embeddings (Definition~\ref{def:affineemd}), which are the crucial ingredient needed for bounding the error term $\norm{\mathcal{X}_1-\mathcal{X}_2}{2}$.  In this section we will show how JL and OSE properties imply the AE property. All together then, these will enable us to verify the requirements of Theorem~\ref{thm:onepass_kron_nondist} in a straightforward manner once we have specified the type of leave-one-out measurements, and the particular type of sensing matrices.

\begin{thm}[Error bound for one-pass] \label{thm:onepass_kron_nondist} Let $\mathcal{X}_1 = [\![ \mathcal{H}, Q_1,\dots, Q_d ]\!]$ denote the single-pass tensor estimate for a $d$-mode tensor $\mathcal{X}$ with side length $n$ obtained from Algorithm~\ref{alg:loo_one_pass_prime} using leave-one-out linear measurements $B_i$ for $i\in[d]$, and core measurements $\mathcal{B}_c$. Let $\epsilon > 0$ and $\delta \in (0,1/2)$.  Then,

\begin{equation} \label{eqn:onepassbound}
    \norm{\mathcal{X}_1 - \mathcal{X}}{2}  \leq   (1+e^{\epsilon}) \sqrt{ \frac{1+\epsilon}{1-\epsilon} \sum_{j=1}^d \Delta_{r,j}} \quad \text{ where } \Delta_{r,j} := \sum_{i=r+1}^{\tilde{n}_j} \sigma_i \left( X_{[j]} \right)^2
\end{equation}
 will hold whenever \begin{enumerate}
    \item $\Omega_{(i,i)}\in\R^{n \times n}$ are full-rank matrices for all $i\in[d]$,
    \item  $\Omega_{(i,i)}^{-1} B_i =  X_{[i]}\Omega^T_{-j} \in \R^{n \times m^{d-1}}$ are $(\epsilon,0,r)$-PCP sketches of $ X_{[i]}$ for each $i\in[d]$, and
    \item $\Phi_i \in \R^{m_{c} \times n}$ is an $(\epsilon/d,r,n^{i-1} r^{d-i})$-AE for the matrices $Q_i$ and $X'_i$ as in  \eqref{eqn:unfolding_core_diff} for all $i \in [d]$. 
\end{enumerate} 
\end{thm}
\begin{proof} Recalling \eqref{eqn:triangle_inequlaity_main} we have that
\begin{align} 
    \norm{\mathcal{X}_1 - \mathcal{X}}{2} =     \norm{\mathcal{X}_1 - \mathcal{X} + \mathcal{X}_2 - \mathcal{X}_2}{2} \leq   \underbrace{\norm{\mathcal{X} - \mathcal{X}_2}{2}}_{\text{Term I}} + \underbrace{ \norm{\mathcal{X}_1 - \mathcal{X}_2}{2} }_{\text{Term II}}.
\end{align}
We know from Theorem~\ref{thm:errorboundfactors} that
\begin{equation} \label{eqn:term_i}
\norm{\mathcal{X} - \mathcal{X}_2}{2}   
    \leq \sqrt{\frac{1+\epsilon}{1-\epsilon} \sum_{j=1}^d \Delta_{r,j} }
\end{equation}
whenever $ X_{[j]}\Omega^T_{-j}  \in \R^{n \times m^{d-1}}$ are $(\epsilon,0,r)$-PCP sketches of $X_{[j]}$ for each $j\in[d]$. Furthermore, Theorem~\ref{thm:error-between-two-passes} together with our third assumption implies that 
\begin{equation} \label{eqn:term_ii}
  \norm{\mathcal{X}_1 - \mathcal{X}_2}{2}  \leq e^{\epsilon} \norm{\mathcal{X} - \mathcal{X}_2}{2}.  
\end{equation} Using \eqref{eqn:term_i} and \eqref{eqn:term_ii} in \eqref{eqn:triangle_inequlaity_main} now yields the desired result.
\end{proof}

We now have need for the lemma that links the AE property to the JL and OSE properties, and thus provides the necessary machinery to fully account for row bounds that guarantee with high probability that the recovered tensor satisfies the stated bound in Theorem~\ref{thm:onepass_kron_nondist}.

\begin{lem} \label{lem:affineembdbnd} Let $B\in\R^{n\times N}$ and suppose that $Q \in \R^{n\times r}$, $n\geq r$, has orthonormal columns.  If $\Phi \in \R^{m\times n}$ is an $\left(\frac{1}{2},\delta,r \right)$-OSE for $Q$, and has the $\left(\frac{\epsilon}{2\sqrt{r}},\delta \right)$-AMM property for $Q^T$ and $(I-Q Q^T )B $, then $\Phi$ will be an $(\epsilon, r, N)$-AE for the matrices $Q$ and $B$
with probability at least $1-2\delta$.
\end{lem}

\begin{proof}
Denote  $\tilde{Y} :=(\Phi Q)^{\dagger}\Phi B $ and $Y' :=Q^T B $, the solutions to the sketched and un-skechted linear least square problems given by minimizing  $\norm{\Phi B-\Phi Q Y }{F}$ and $\norm{B - QY}{F}$, respectively, with respect to $Y$.  Whenever $\Phi$ is a $(1/2,\delta,r)$-OSE for $Q$ we know that $\Phi Q$ will be full-rank.  It follows that $(\Phi Q)^{T} \Phi (B-Q \tilde{Y}) = 0$ will then also hold because $\left(\Phi Q\right)^{\dagger} = \left[\left(\Phi Q\right)^T \left(\Phi Q\right)\right]^{-1}\left(\Phi Q\right)^T$.  As a consequence, 
\begin{align*}
 Q^T \Phi^T \Phi Q (\tilde{Y} - Y') 
&= Q^T \Phi^T \Phi Q (\tilde{Y} - Y') + (\Phi Q)^{T} \Phi (B - Q \tilde{Y}) \\
&= Q^T \Phi^T \Phi (Q \tilde{Y} - QY' + B - Q \tilde{Y}) \\
&= Q^T \Phi^T \Phi (B - QY').
\end{align*}
When the approximate matrix multiplication property also holds we will now have that
\begin{align}
\norm{(\Phi Q)^T \Phi Q (\tilde{Y} - Y')}{F} &= \norm{Q^T \Phi^T \Phi (B - QY') }{F} ~=~ \norm{Q^T \Phi^T \Phi (I - QQ^T)B }{F} \nonumber \\
 &\leq \frac{\epsilon}{2\sqrt{r}} \norm{Q^T}{F}\norm{(I - QQ^T)B}{F} \label{equ:DiffboundYtildeY'}\\
  &= \frac{\epsilon}{2} \norm{(I - QQ^T)B}{F} ~=~\frac{\epsilon}{2} \norm{B - QY'}{F}. \nonumber
\end{align}

Furthermore, whenever $\Phi$ is a $\left(\frac{1}{2},\delta,r \right)$-OSE for the column space of $Q$, all the eigenvalues of $Q^T \Phi ^T \Phi Q - I$ will be within the interval $[-1/2,1/2]$.  Thus, we can bound its operator norm $\norm{Q^T \Phi ^T \Phi Q - I}{} \leq 1/2.$
We may now combine this operator norm bound with \eqref{equ:DiffboundYtildeY'} and see that
\begin{align*}
\norm{  \tilde{Y} - Y'}{F} &=
\norm{ (\tilde{Y} - Y') - (\Phi Q)^T \Phi Q  (\tilde{Y} - Y') +(\Phi Q)^T \Phi Q (\tilde{Y} - Y') }{F} \\
&\leq \norm{(\Phi Q)^T \Phi Q (\tilde{Y} - Y')}{F}+\norm{\left[(\Phi Q)^T \Phi Q - I \right] (\tilde{Y} - Y')}{F} \\
  &\leq \frac{\epsilon}{2}  \norm{(I - QQ^T)B}{F} + \norm{\left[(\Phi Q)^T \Phi Q - I \right]}{}\norm{ \tilde{Y} - Y'}{F} \\
    &\leq \frac{\epsilon}{2}  \norm{(I - QQ^T)B}{F} + \frac{1}{2}\norm{\tilde{Y} - Y'}{F}.
\end{align*}
Rearranging the inequality above while noting the invariance of the Frobenius norm to multiplication by a matrix with orthogonal columns, we learn that 
$$\norm{\tilde{Y} - Y'}{F} = \norm{ Q( \tilde{Y} - Y')}{F} \leq \epsilon \norm{(I - QQ^T)B}{F}.$$

To finish, we may now apply the triangle inequality to see that
\begin{align*}
    \norm{\tilde{Y}}{F} & \leq \norm{\tilde{Y} - Y'}{F} + \norm{Y'}{F} \\
    &\leq \epsilon \norm{(I - QQ^T)B}{F} +  \norm{Q^T B }{F} \\ 
    &\leq \epsilon \norm{(I-QQ^T)}{}\norm{B}{F} +  \norm{Q^T}{}\norm{ B }{F} \\ 
    &= (1+\epsilon) \norm{ B }{F}.
\end{align*}
In addition, we note that taking a union bound over the two necessary OSE and AMM conditions establishes the stated probability guarantee.
\end{proof}

We are now prepared to state how a particular choice of distribution used to generate our measurement matrices as well as the leave-one-out measurement type (Kronecker or Khatri-Rao) can satisfy the error bound \eqref{eqn:onepassbound} with high probability.  Note that below Algorithms~\ref{alg:kron_1pass} and \ref{alg:khat_1pass} refer to the specialization of Algorithm~\ref{alg:loo_one_pass_prime} to the type of leave-one-out measurement (Kronecker or Khatri-Rao, respectively).

\begin{thm}[Error bound for one-pass Kronecker-structured sub-gaussian measurements]  \label{thm:onepass_kron}
Suppose $\mathcal{X}$ is a $d$-mode tensor with side length $n$. Let $\epsilon > 0$, $\delta \in (0,\frac{1}{3})$, $r\in[n]$. Furthermore, let
\begin{enumerate}
    \item $\Omega_{(i,i)}\in\R^{n \times n}$ be arbitrary full-rank matrices, 
    \item $\Omega_{(i,j)}\in\R^{m\times n}$ for $i\neq j$ be random matrices with mutually independent, mean zero, variance $m^{-1}$, sub-gaussian entries with
   \[ m\geq \max \left\{\frac{C_1 r (d-1)^2}{\epsilon^2} \ln \left(\frac{n^{d} d^2}{\delta}\right),\frac{C_2 (d-1)^2 }{\epsilon^2} \ln \left(\left( \frac{141}{\epsilon} \right)^r \frac{n^{d-2}d^2}{\delta} \right)  \right\}, {\rm and}\]
    \item $\Phi_i \in\R^{m_c \times n}$ be random matrices with mutially independent, mean zero, variance $m^{-1}$, sub-gaussian entries with
    \begin{equation}
    \label{equ:mc_subgauss}
    m_c\geq \max \left\{C_3 \ln\left(\frac{(94)^{r} d}{\delta}\right),\frac{C_4 r d^2}{\epsilon^2} \ln\left( \frac{ d (r + n^{d-1})^2}{\delta} \right) \right\}.
    \end{equation}
\end{enumerate} 
Then $\mathcal{X}_1 = [\![ \mathcal{H}, Q_1,\dots, Q_d ]\!]$ the output of Algorithm \ref{alg:kron_1pass} (i.e., Algorithm~\ref{alg:loo_one_pass_prime} specialized to Kronecker sub-gaussian measurements $\mathcal{B}_i$ and $\mathcal{B}_c$), 
will satisfy \eqref{eqn:onepassbound} with probability at least $1-3\delta$.
\end{thm}

\begin{proof}  We verify that the requirements of Theorem~\ref{thm:onepass_kron_nondist} are satisfied.  Note that:
\begin{enumerate}
    \item $\Omega_{(i,i)}\in\R^{n \times n}$ are full-rank matrices by assumption.
    
    \item  We have from Corollary~\ref{lem:measurementsdefinepcp} where $\delta \gets \frac{\delta}{d}$ that $X_{[j]}\Omega^T_{-j} $ is a $(\epsilon,0,r)$-PCP sketch of $X_{[j]}$ for all $j \in [d]$ with probability at least $1-\delta$ when $\Omega_{(i,j)}\in\R^{m\times n}$ are independent sub-gaussian random matrices with
\[
m\geq \max \left\{\frac{C_1 r (d-1)^2}{\epsilon^2} \ln \left(\frac{n^{d} d^2}{\delta}\right),\frac{C_2 (d-1)^2 }{\epsilon^2} \ln \left(\left( \frac{141}{\epsilon} \right)^r \frac{n^{d-2}d^2}{\delta} \right)  \right\}.
\]

    \item A substitution of $\epsilon \gets \frac{1}{2}$, $\delta\gets \frac{\delta}{d}$ into Corollary~\ref{cor:subgOSE} yields,
\begin{equation} \label{eqn:mc_1}
    m_c \geq C_3 \ln\left(\frac{(94)^{r} d}{\delta}\right)
\end{equation}
in order to ensure $\Phi_i$  is $(\frac{1}{2}, \frac{\delta}{d},r)$-OSE. Using Corollary~\ref{cor:AMMsetexist}, where $\epsilon \gets \frac{\epsilon}{2d\sqrt{r}}$, $\delta \gets  \frac{\delta}{d}$ and noting that the matrices $Q_i$ and $X'_i$ as in  \eqref{eqn:unfolding_core_diff} have are $r\times n$ and $n\times n^{i-1} r^{d-i}$, respectively, we have that when
\begin{equation} \label{eqn:mc_2}
    m_c \geq \frac{C_4 r d^2}{\epsilon^2} \ln\left( \frac{ 2 d (r + n^{d-1})^2}{\delta} \right)
\end{equation}
then $\Phi_i$ has the  $(\frac{\epsilon}{2d\sqrt{r}}, \frac{\delta}{d })$-AMM property for each $i \in [d]$. Lemma~\ref{lem:affineembdbnd} now shows how the OSE and AMM properties ensure that $\Phi_i$ has the desired AE property for each $i \in [d]$ with probability at least $1 - 2 \delta/ d$. The union bound now implies that the third requirement of Theorem~\ref{thm:onepass_kron_nondist} will hold with probability at least $1-2 \delta$.
\end{enumerate}

Taking a maximum over \eqref{eqn:mc_1} and \eqref{eqn:mc_2} after simplifying and adjusting constants then yields \eqref{equ:mc_subgauss}. 
A final union bound over the failure probabilities for requirements of $\Omega_{(i,j)}$ and $\Phi_i$ now yields the result.
\end{proof}

The following runtime analysis demonstrates that instances of Algorithm~\ref{alg:loo_one_pass_prime} can indeed recover low-rank approximations of $d$-mode tensors of side length $n$ in $o(n^d)$-time.  As a result, one can see that Algorithm~\ref{alg:loo_one_pass_prime} is effectively a sub-linear time recovery algorithm for a large class of low Tucker-rank tensors.

\begin{thm}
\label{thm:kron_runtime}
Suppose that $m < n < m^{d-1}$ and that $m > C m_c$ for some absolute constant $C > 0$.  Then, Algorithm \ref{alg:kron_1pass}, when given the $B_i$ and $\mathcal{B}_c$ measurements as input, runs in $\mathcal{O}(d n^2 m^{d-1})-time$ and requires the storage of $dnm^{d-1} + m_c^{d}$ measurement tensor entries, and at most $d(d-1)mn + n^2 + dm_cn$ total measurement matrix entries.
\end{thm}

\begin{proof}
Inside the factor matrix recovery loop of Algorithm \ref{alg:recover_factors}, called by Algorithm \ref{alg:kron_1pass}, the two main sub-tasks are to solve the linear system $\Omega_{(i,i)} F_{i} = B_{i}$ and to compute a truncated SVD, $F_{i} = U\Sigma V^T$. Solving the linear system can be accomplished using $QR$-factorization via Householder orthogonalization. Doing so requires $\frac{4}{3}n^3$ floating point operations to compute the factorization of $\Omega_{(i,i)}$, and $2 n^2 m^{d-1}$ operations to form $Q^T B_{i}$ and $n^2 m^{d-1}$ operations to solve $RF_{i} = Q^T B_{i}$ via back substitution.  The complexity of computing the SVD is $\mathcal{O} (n m^{d-1} \min (n,m^{d-1}))$. Therefore the factor recovery loop has overall complexity
\[
\mathcal{O} (d n^2 m^{d-1})
\]
if we assume that $n < m^{d-1}$. 

Next we consider the core recovery loop. for the first iteration of the loop of Algorithm \ref{alg:recover_core}, we must form $\Phi_1 Q_1 $ at a cost of $\mathcal{O}(m_c n r)$. Next we solve a linear system $\Phi_1 Q_1 H = B_{1}$ at a cost of $\mathcal{O}(2m_c r^2 - \frac{2}{3}r^3 +3m_c^{d} r )$. The first iteration dominates the complexity, since subsequent solves use a smaller right hand side formed from solutions from the previous iterations. Furthermore, if we assume $m_c = \mathcal{O}(m)$, we have a core recovery loop with complexity
\[
\mathcal{O}(d m^d r).
\]Thus overall the recovery algorithm has $\mathcal{O} (d n^2 m^{d-1})$ complexity. In the situation where $\Omega_{(i,i)}=I_n$ then the computation of the SVDs in the factor recovery step dominates the run-time of the algorithm. Clearly the size of the measurement tensors are $nm^{d-1}$ per factor and $m_c^d$ for the core, which yields the space complexity of the measurements. 
\end{proof}

One of the advantages of the structure of the argument in Theorem~\ref{thm:onepass_kron_nondist} is that once it is known how to ensure a given random matrix will satisfy the JL property, we can (with the help of Lemma~\ref{lem:affineembdbnd}) account for how to assemble related measurement operators that satisfy the error bound \eqref{eqn:onepassbound} in a straightforward way.  For example, using Theorem 3.1 in \cite{Krahmer2011} along with bounds appearing in that work on the sketching dimension, we have that a sub-sampled and scrambled Fourier matrix is a $(\epsilon,\delta,p)$-JL of vectors in $\mathbbm{R}^n$ provided that
\[
m\geq \frac{C}{\epsilon^2}\log\left(\frac{p}{\delta}\right) \log^4 n.
\]
Using such existing results one can easily update Theorem~\ref{thm:onepass_kron} to instead use sub-sampled and scrambled Fourier measurement matrices $\Omega_{(i,j)}$ and $\Phi_i$ instead of matrices with independent sub-gaussian entries.

Furthermore, it need not be the case that the distribution is the same for each component map $\Omega_{(i,j)}$ or $\Phi_i$. It is important only that each map satisfies the JL primitive for arbitrary sets. We are free to choose a measurement map for a particular mode to suit some other purpose. For example, in the case that the side lengths of the tensor are unequal, we may prefer to choose a map that admits a fast matrix-vector multiply in order to economize run-time for the modes which are long, and on smaller modes, choose maps which have better trade-offs for quality of approximation in terms of $m$ (e.g., we may prefer dense sub-gaussian random matrices for these modes).

On the other hand, a measurement ensemble which does not rely on the Kronecker product of components, like that in Theorem \ref{thm:khatasJL} does not admit this sort of reasoning.  Mixing measurement maps of different kinds in that case has no clear advantage, and indeed, may even serve to undermine the advantages.  Nonetheless, 

Theorem~\ref{thm:onepass_kron_nondist} can still be used to provide Khatri-Rao structured leave-one-out measurement results.  Verifying the requirements of Theorem~\ref{thm:onepass_kron_nondist} for Khatri-Rao measurements using Theorem~\ref{thm:khat_loo_is_pcp} yields the following result. 

\begin{thm}[Error bound for one-pass Khatri-Rao structured sub-gaussian measurements] \label{thm:onepass_khat}
Suppose $\mathcal{X}$ is a $d$-mode tensor with side length $n$, Let $\epsilon > 0$, $\delta \in (0,\frac{1}{3})$, $r\in[n]$. Furthermore, let
\begin{enumerate}
    \item $\Omega_{(i,i)}\in\R^{n \times n}$ be full-rank matrices of any kind,
    \item $\Omega_{(i,j)}\in\R^{\tilde{m}\times n}$ for $i\neq j$ be random matrices with mutually independent, mean zero, variance $\tilde{m}^{-1}$, sub-gaussian entries with
   \begin{align*}   \tilde{m}\geq C^{d-1}  \max\biggl\{ \epsilon^{-2} \ln \left( \frac{\left(\frac{141}{\epsilon}\right)^r 2d}{\delta} \right) ,~  \epsilon^{-1}\left(\ln \left( \frac{\left( \frac{141}{\epsilon}\right)^r 2d}{ \delta}\right)\right)^{d-1},  \\ 
    ~\frac{r}{\epsilon^{2}} \ln \left(\frac{34 n^2 d}{\delta} \right),~ \frac{r}{\epsilon}\left(\ln \left(\frac{34n^2d}{ \delta}\right)\right)^{d-1}\biggr\}, \end{align*}
   and
    \item $\Phi_i \in\R^{m_c \times n}$ be random matrices with mutually independent, mean zero, variance $m^{-1}$, sub-gaussian entries with
    \[m_c\geq \max \left\{C_3 \ln\left(\frac{(94)^{r} d}{\delta}\right),\frac{C_4 r d^2}{\epsilon^2} \ln\left( \frac{ d (r + n^{d-1})^2}{\delta} \right) \right\}.\]
\end{enumerate} 
Then $\mathcal{X}_1 = [\![ \mathcal{H}, Q_1,\dots, Q_d ]\!]$ the output of Algorithm \ref{alg:khat_1pass} (i.e., Algorithm~\ref{alg:loo_one_pass_prime} specialized to Khatri-Rao sub-gaussian leave-one-out measurements $\mathcal{B}_i$ and Kronecker measurements $\mathcal{B}_c$)
 will satisfy \eqref{eqn:onepassbound} with probability at least $1-3\delta$.
\end{thm}
\begin{proof} The proof is again based on verifying the 
requirements of Theorem~\ref{thm:onepass_kron_nondist}.
\begin{enumerate}
    \item The first requirement is satisfied by assumption.
    \item The row requirement for the $\Omega_{(i,j)}$ follows from an application of Theorem \ref{thm:khat_loo_is_pcp} with $\delta \gets \delta/d$.  As a result of doing so, we learn that the second requirement will be satisfied with probability at least $1 - \delta$ after a union bound.
    \item The third requirement is satisfied identically to the argument in the proof of Theorem~\ref{thm:onepass_kron}.
\end{enumerate}
The proof now concludes identically to the proof of Theorem~\ref{thm:onepass_kron}.
\end{proof}

Comparing the Theorem~\ref{thm:kron_runtime}, we note, e.g., that Theorem~\ref{thm:onepass_khat} will require the storage of $d \tilde{m} n+m_c^d$ measurement tensor entries (where the $\tilde{m}$ here is from the second condition of Theorem~\ref{thm:onepass_khat}).  Comparing this to the 
$dnm^{d-1}+m_c^d$ measurements from Theorem~\ref{thm:kron_runtime} (where $m$ here is as in 
Theorem~\ref{thm:onepass_kron}) one can see that there are parameter regimes where Khatri-Rao structured measurements will lead to a smaller overall measurement budget.

\section{Numerical Experiments}

In this section we present numerical results that support our theoretical contributions and address practical trade-offs involved in the different choices for measurement type. (Code and resulting data are available at \burl{https://github.com/cahaselby/leave_one_out_recovery}.) Unless otherwise specified, the tensors in the experiments are random three-mode cubic tensors, with side length $n=300$ and rank $r=10$. We use the following procedure (same as in \cite{Sun2019}) to generate low-rank tensors; the core's entries are uniformly and independently drawn from $[0,1]$ and the factors are formed by first sampling a standard normal distribution for each entry and then normalized and made orthogonal using $QR$-factorization. The data points in the plots are the mean of $100$ independent trials. The parameter $m$ refers to the sketching dimension for maps $\Omega_{(i,j)}\in \R^{m\times n}$ used in recovering factor $Q_i$. For the left-out mode we remove the need to solve the full $n\times n$ linear system by setting $\Omega_{(i,i)}$ to be the $n\times n$ identity matrix. The parameter $m_c$ refers to the sketching dimension for $\Phi_{i}\in \R^{m_c \times n}$ which are used in recovering the core $\mathcal{H}$. 

In experiments with noise, the additive Gaussian noise tensor $\mathcal{N}$ is scaled according to the desired signal to noise ratio and added to the true, (low-rank) tensor $\mathcal{X}_0$. That is, $\mathcal{X} = \mathcal{X}_0+\mathcal{N} $ is the observed, noisy tensor. 

Signal to noise ratio (SNR) is calculated as $10 \log_{10}\left(\norm{\mathcal{X}}{2}/\norm{\mathcal{X}_0 - \mathcal{X}}{2}\right)$. Relative error is calculated as $(\|\hat{\mathcal{X}} - \mathcal{X}
\|_2)/\norm{\mathcal{X}_0}{2}$ where $\hat{\mathcal{X}}$ is the full estimated tensor. 

\subsection{Recovering low-rank tensors}

In this first simple experiment, we fix the signal to noise ratio at $30$ decibels (dB) and vary the sketching dimension $m$ to show the dependence on the accuracy of our estimate on the number of measurements. For each $m$ we set $m_c = 2m$. Rank truncation is fixed at $r=10$, which matches the rank of the true, noiseless tensor $\mathcal{X}_0$, see Figure~\ref{fig:base_result}. For the plot (b) in the figure, we have the maximum principal angle among the three estimated factors and true factors $(Q_i, U_i)$ in degrees, see \cite{Knyazev2002}. Note, there is no straightforward way to plot the relative error which is due to the factor estimates vs. the core estimate, because the decomposition will in general not be unique. However, since principal angle is invariant to non-singular transformations, plot (b) provides empirical evidence that the factor estimates alone are improving with sketching dimension. We note that for these low-rank tensors with noise, we are able to fit at or below the level of noise (relative error of $0.001$) easily - evidently finding good rank $10$ approximations to the (full-rank) noisy tensor $\mathcal{X}$. 
\begin{figure}[ht]
\
\centering
\includegraphics[width=1\textwidth]{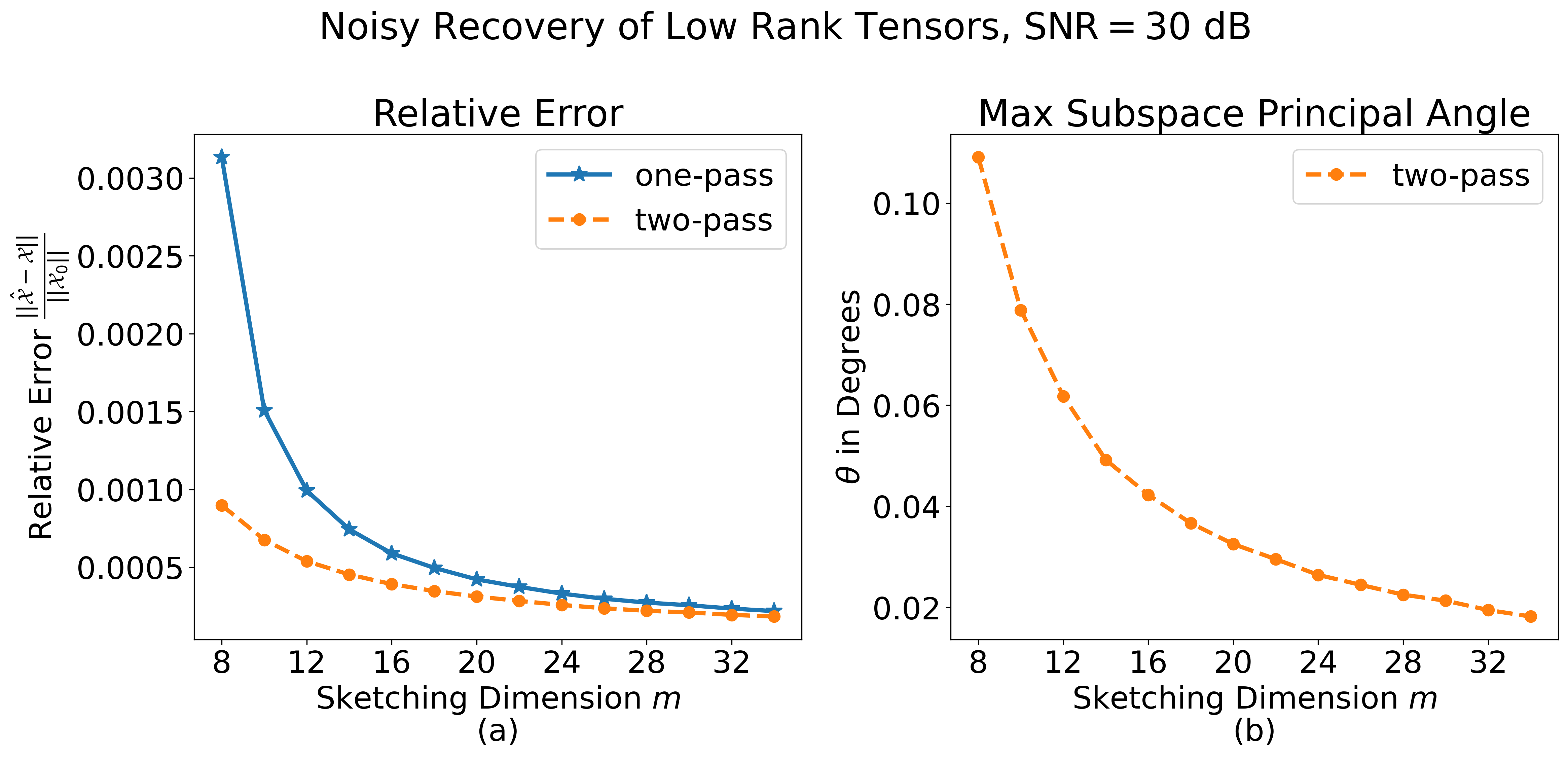}
\caption{Error plots for different sketching dimensions and a fixed SNR of 30 dB and a fixed rank truncation of $10$. Plot (a) compares relative errors for both one-pass and two-pass recovery. Plot (b) shows the maximum principal angle among all estimated sub-spaces and the true factor matrices.}
\label{fig:base_result}
\end{figure}

This perhaps surprising result motivated us to try the method on a class of tensors in which we could be more certain about what quality of rank $10$ approximation is achievable. In our second set of experiments, we examine performance on super-diagonal tensors with tail decay. Since we are truncating to rank 10, this tail can be thought of as structured, deterministic noise. These are tensors where all values are zero except for those on the diagonal, and where all values on the diagonal for indices larger than $r=10$, we have some decay in their magnitude. In particular we have two types, exponential tail decay in plot (a), where
\begin{equation} \label{eqn:super_exp}
\mathcal{X}_{ijk} = \begin{cases} 1 & i=j=k \in [r+1] \\
10^{-1(i - r)} & i=j=k \in [r+2,n] \\
0 & \text{otherwise}
\end{cases}
\end{equation}
and polynomial tail decay in plot (b),
\begin{equation}\label{eqn:super_poly}
   \mathcal{X}_{ijk} = \begin{cases} 1 & i=j=k \in [r+1] \\
(i-r)^{-1} & i=j=k \in [r+2,n] \\
0 & \text{otherwise.}
\end{cases} 
\end{equation}

These highly constrained tensors are clearly not low-rank, however it is reasonable to suppose that a recovery algorithm for a given rank truncation would output an estimate that is close to the leading $r$ terms of the diagonal. The residual in that case will simply be the norm of the tail-sum $\sqrt{\sum_{i=r+1}^n \mathcal{X}_{iii}^2 }$, which we have included as the red horizontal line in Figure~\ref{fig:super_diagonal}.

\begin{figure}[ht]
\
\centering
\includegraphics[width=1\textwidth]{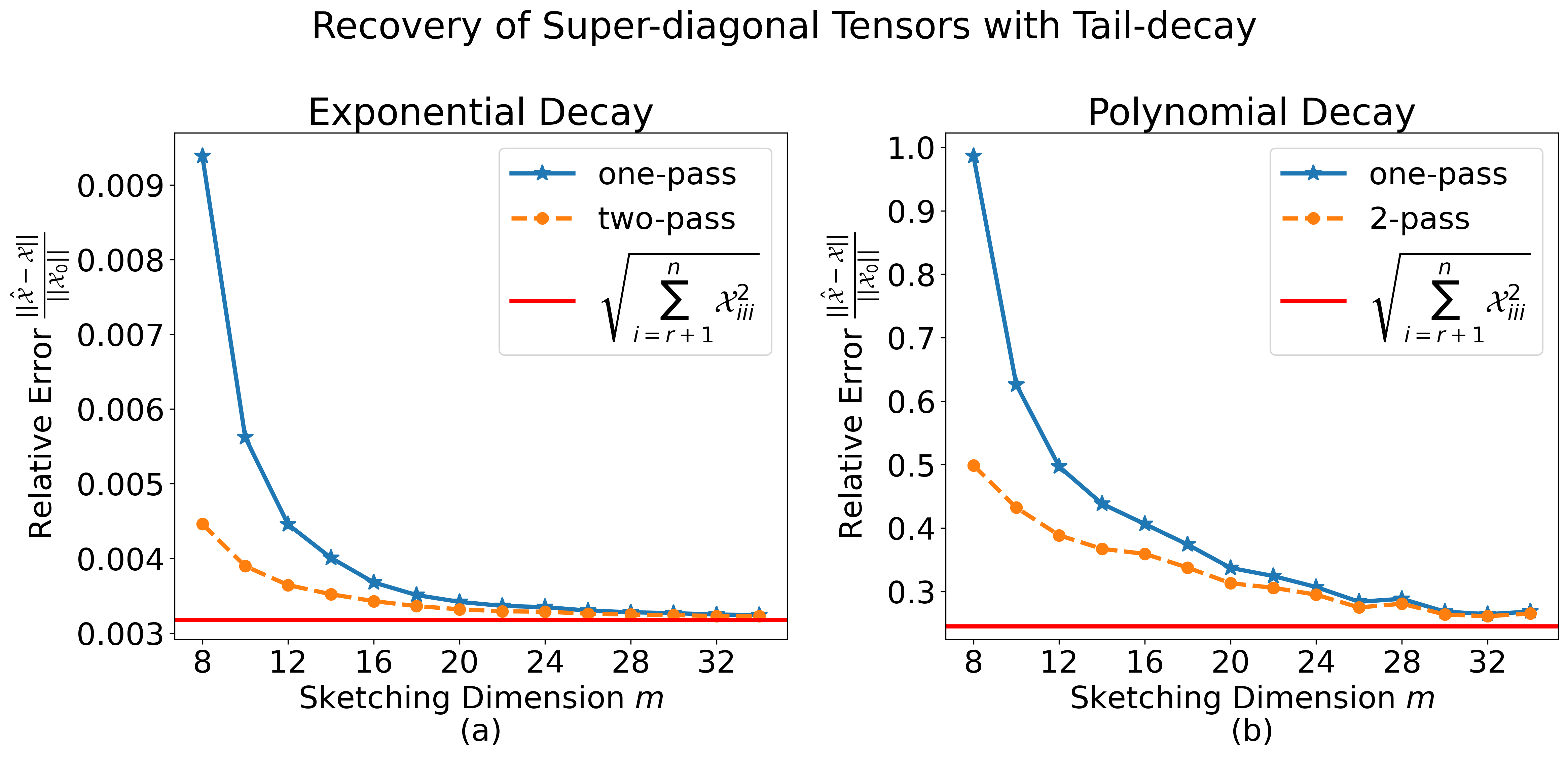}
\caption{Error plots for different sketching dimensions in the noiseless setting with a fixed rank truncation of 10. Plot (a) compares relative errors for both one-pass and two-pass recovery for super-diagonal tensors where the diagonal entries have exponential decay of type \eqref{eqn:super_exp}, and plot (b) super-diagonal tensors with polynomial decay of the type \eqref{eqn:super_poly}}
\label{fig:super_diagonal}
\end{figure}

\subsection{Allocating Core and Factor Measurements}
 One question raised by our error analysis is how to weigh the error contribution between the tasks of estimating the factor matrices and estimating the core. In other words, for a given total measurement budget, how should we allocate between the two tasks if we wish to decrease overall relative error? In the following experiment (see Figure~\ref{fig:corevsfactor}) we find the relative error under various noise levels for pairs of sketching dimensions $(m,m_c)$. We compare pairs $(13,12)$ and $(11,36)$ and $(8,48)$. These choices of sketching dimensions were chosen since they have nearly equal overall compression ratios of 0.57\%, 0.58\%, 0.62\% respectively, however they vary considerably on whether they emphasize measurements to be used in estimating the factors or the core of the tensor. Note that the two-pass error, which relies only on the factor matrix estimates is naturally best when the factor sketches are larger, i.e. the $m=13$ case. However the relative error of the recovered tensor in the one-pass setting is more than ten times better when more of the total measurement budget is allocated to estimate the core as shown in Figure~\ref{fig:corevsfactor}. This shows that in some situations it is preferable to allocate more resources to obtain measurements for the core than the factors, up to some threshold. For example in Figure~\ref{fig:corevsfactor}, the rank of the true signal is 10, and going below this dimension for the factor sketches does correspond with no longer improving on the accuracy in terms of the trade-off between $m$ and $m_c$. 

\begin{figure}[ht]
\centering
\includegraphics[width=1\textwidth]{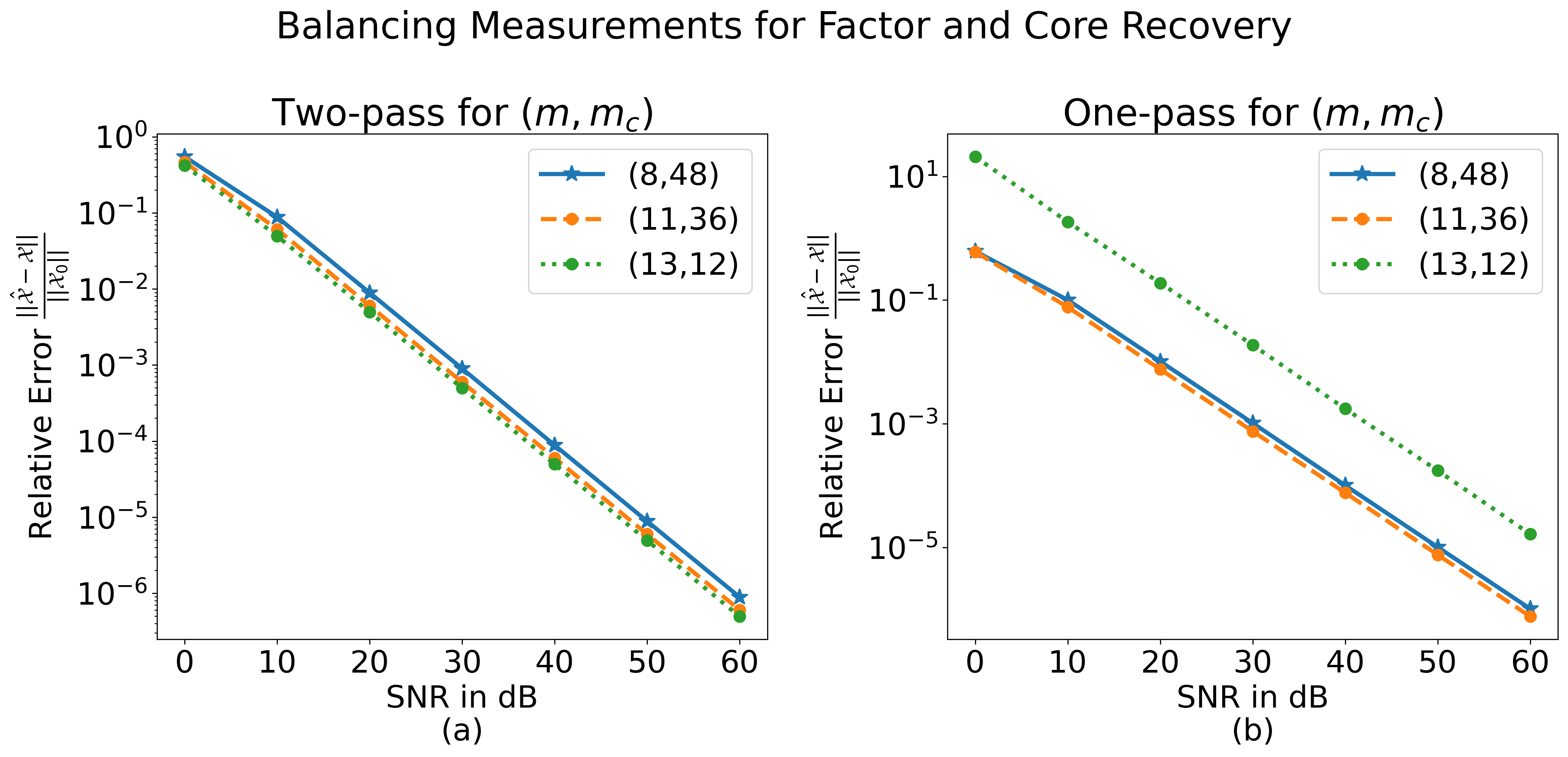}
\caption{Relative error plots at different signal to noise ratios for two-pass and one-pass recovery of noisy low-rank tensors. Ordered pairs indicates the choice for sketching dimensions $(m,m_c)$.}
\label{fig:corevsfactor}
\end{figure}

\subsection{Error bounds apply to sub-gaussian measurement matrices}

In this next experiment we demonstrate, in a similar manner as done in Figure 1 in \cite{Sun2019}, that recovery performance of Algorithm~\ref{alg:khat_1pass} does not vary greatly for different choices of types of sub-gaussian measurement matrices. What is different from that earlier work is that the measurement ensembles are all Kronecker structured. Plotted in Figure~\ref{fig:meas_types} are relative errors for Gaussian (g), sparse uniform from $[-1,0,1]$ with weights $\frac{1}{6},\frac{2}{3},\frac{1}{3}$ (sp0), sub-sampled randomized Fourier transforms as in \cite{Woolfe2008} (rfd)  and a mixed measurement ensemble that uses Gaussian-RFD-sparse measurements where we vary by mode which measurement type is used, which is a scenario that is practically and theoretically not well suited for the Khatri-Rao structured measurement operators used in \cite{Sun2019}. 

\begin{figure}[ht]
\centering
\includegraphics[width=1\textwidth]{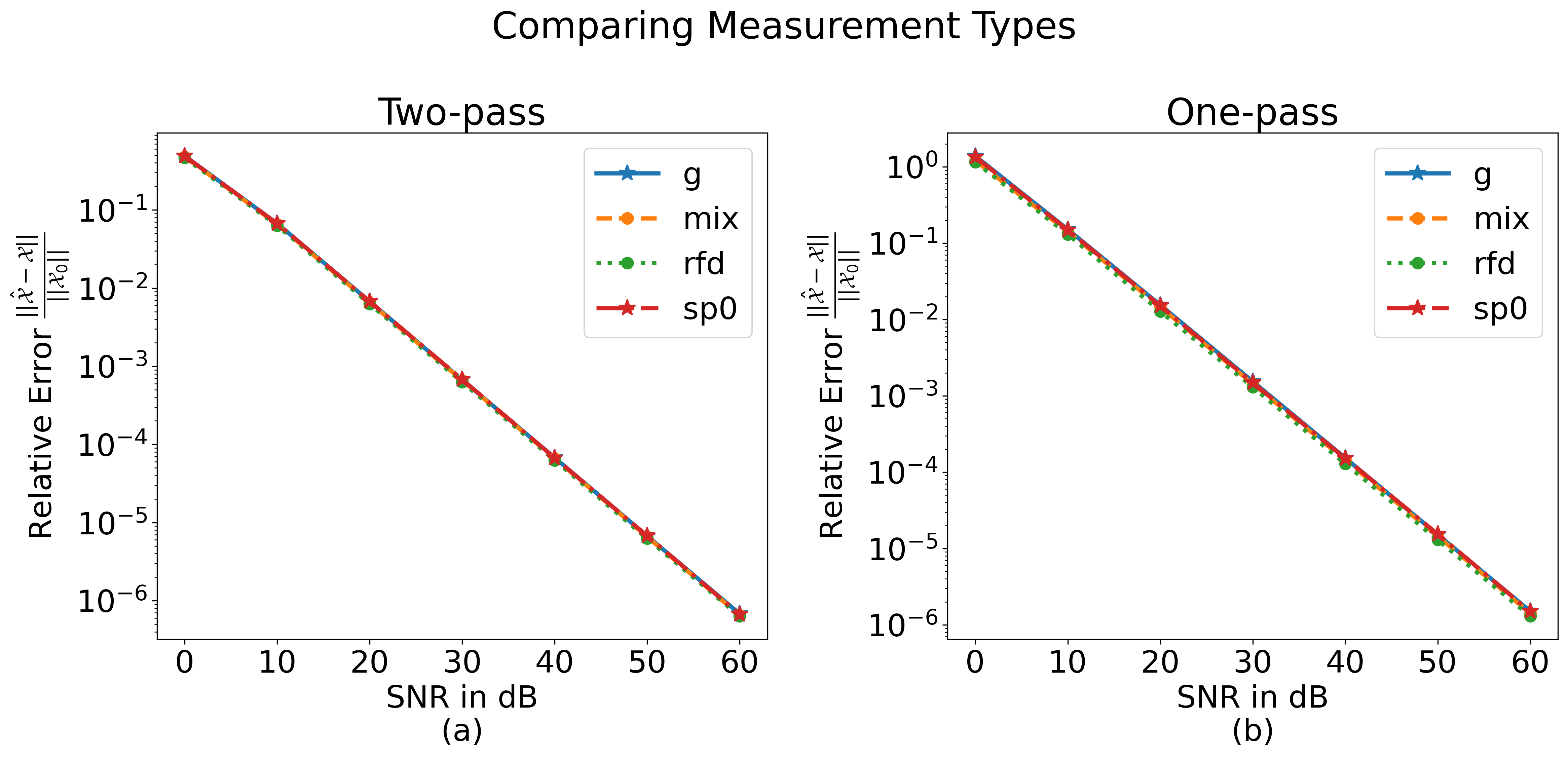}
\caption{Relative errors for one-pass, two-pass for Kronecker measurement ensembles made up of different kinds of sub-gaussian random matrices. The legend g, sp0, rfd, mix correspond to Gaussian, sparse, sub-sampled random Fourier transform, and a mixture of all three for the measurement ensembles.}
\label{fig:meas_types}
\end{figure}

\subsection{Comparison to Khatri-Rao}

\begin{figure}[ht]
\centering
\includegraphics[width=1\textwidth]{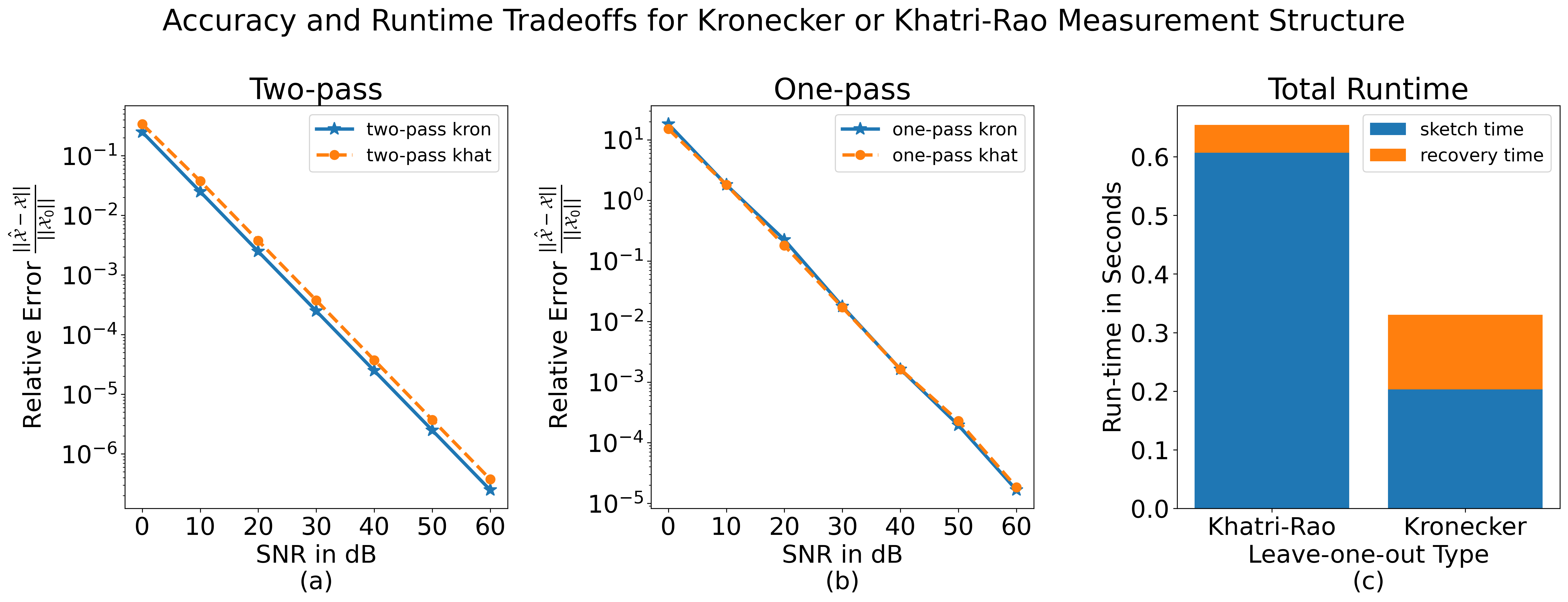}
\caption{Relative error comparison for Khatri-Rao and Kronecker-structured measurements. Right hand figure shows the average time for sketching and recovery phases of the algorithm.}
\label{fig:khatvskron}
\end{figure}
This set of experiments demonstrates that the sketching phase will dominate the run-time of  Algorithm~\ref{alg:loo_one_pass_prime} regardless of the choice of leave-one-out type, however the Kronecker-structured measurements are able to generate more measurements for a fixed number of operations as compared to Khatri-Rao structured measurements, see Figure~\ref{fig:khatvskron}. This means that it is possible to achieve similar or better performance using strictly modewise measurements and in less overall time as problems grow in size with respect to total number of tensor elements; i.e. both number of modes and length of those modes. In Figure~\ref{fig:khatvskron} for the Kronecker-structured measurements we sketch to $m=25$, for the Khatri-Rao ensemble we sketch pairs of modes to $225$. We see that the Kronecker measurements perform incrementally better in terms of relative error but at less than half the overall run-time. Sketching times are about five times faster for the Kronecker-structured measurements as compared to the Khatri-Rao. Note that this does trade speed for space - the total number of entries in the leave one out measurements is nearly three times larger for the Kronecker-structured measurements versus the Khatri-Rao, i.e. sketches $B_i$ as per \eqref{eq:kronecker-sketch} and \eqref{eqn:loo_khat} have sizes $300\times 25^2$, and $300 \times 225$ respectively.

\subsection{Application to Video Summary Task}

As a practical demonstration, we consider the same video summary task first described in \cite{Malik2018} and again in \cite{Sun2019}. In this demonstration, the video is taken with a camera in a fixed position. The video is a nature scene and a person walks in front of the camera at two different time points in the second half of the video. The first 100 and the last 193 frames are removed since they include setup that results in small shifts of the camera. The entire video has been converted to grayscale. This yields a three mode tensor of dimensions $2200 \times 1080 \times 1980$ which has a size of about 41 GB when stored as an array of doubles. We wish to identify the parts of the scene that include the person walking and distinguish them from the relatively static scene elsewhere. As discussed in \cite{Sun2019}, there is a third salient time varying feature in this particular video, which is the light intensity of the scene, since at around frame 940 the scene darkens. Furthermore there are changes in the light intensity as the camera automatically adjusts after the person walks in and out of frame. For this reason, we cluster the frames using three centers, rather than two.

In all cases, we use $k$-means to cluster the frames, however we assign features to frames in four different ways:
\begin{enumerate}
    \item Using the sketch $B_1 \in \R^{2200 \times 20^2}$, as in \eqref{equ:BRKSfactors} that leaves out the time dimension, then clustering using $k$-means on the rows of the unfolding of the sketch along the first, temporal mode. 
    \item Unfolding the temporal mode of the reconstructed tensor using a one-pass set of measurements, i.e. $\left(X_1\right)_{[1]}\in\R^{2200 \times 2138400}$ (recall that $\mathcal{X}_1$ denotes the output of Algorithm~\ref{alg:loo_one_pass_prime}).
    \item Unfolding the reconstructed tensor using a two-pass scenario, $\left(X_2\right)_{[1]}\in\R^{2200 \times 2138400}$ (recall that $\mathcal{X}_2$ denotes the output of Algorithm~\ref{alg:kron_2pass}). 
    \item Estimated temporal factor matrix $U_1 \in \R^{2200\times 20}$ (see in Algorithm~\ref{alg:loo_one_pass_prime}). 
\end{enumerate}
As we can see in Figures~\ref{fig:frame_classes}, \ref{fig:frame_ref}, the sketch alone shows reliable clustering of the main temporal changes in the video, which verifies the observation in \cite{Sun2019} about using the measurements as an effective feature set for clustering, although in that case the measurements were Khatri-Rao structured whereas ours are Kronecker-structured.
The unfoldings of the reconstructed tensor also reliably distinguishes the main parts of the scene. The reconstruction is useful at least to get clusted interpretability. Although certainly natural to wish to cluster on the temporal factor, this method appears inferior to any of the preceding. 

As an added advantage of using the modewise, Kronecker structured measurements, we can in principle select measurement maps for different modes. Gaussian measurement maps theoretically have some advantages over other types in terms of accuracy for a fixed number of measurements, whereas applying RFD or other Fourier-like transforms to modes that have longer fibers would net a better payoff in terms of overall run-time because of the faster matrix-vector multiply permitted by these structured matrices. In this demonstration, we use Gaussian matrices along the spatial modes, and $RFD$ matrices for the temporal mode. 

 In the earlier work \cite{Malik2018}, the authors describe a variant of Tucker-Alternating Least Squares (aka Higher Order Orthogonal Iteration, multi-pass scenario) that employed TensorSketch to produce the necessary measurements used to reconstruct the same video tensor data we have used here. In the subsequent work \cite{Sun2019}, those authors again perform the same task, but use a single-pass approach which fits the framework we have described as Algorithm~\ref{alg:loo_one_pass_prime}, where the measurement matrices are Khatri-Rao structured, and the $\Omega_i$ have entries drawn from standard Gaussian distribution. Furthermore, analysis of the type afforded by Theorem~\ref{thm:onepass_khat} may also explain the discrepancy between the sketching dimensions seen in \cite{Malik2018} and \cite{Sun2019}. Naturally there are several differences between the approaches, but the CountSketch matrices used in TensorSketch operators as shown in \cite{Woodruff2019} have a $\mathcal{O} \left(\frac{1}{\delta}\right)$ dependency in order to ensure the OSE property, whereas the other ensembles, such as dense Gaussians, enjoy an $\mathcal{O}\left(\log \frac{1}{\delta}\right)$ dependence for this parameter.

\begin{figure}[ht]
\centering
  \subfloat[]{\includegraphics[width=1\textwidth]{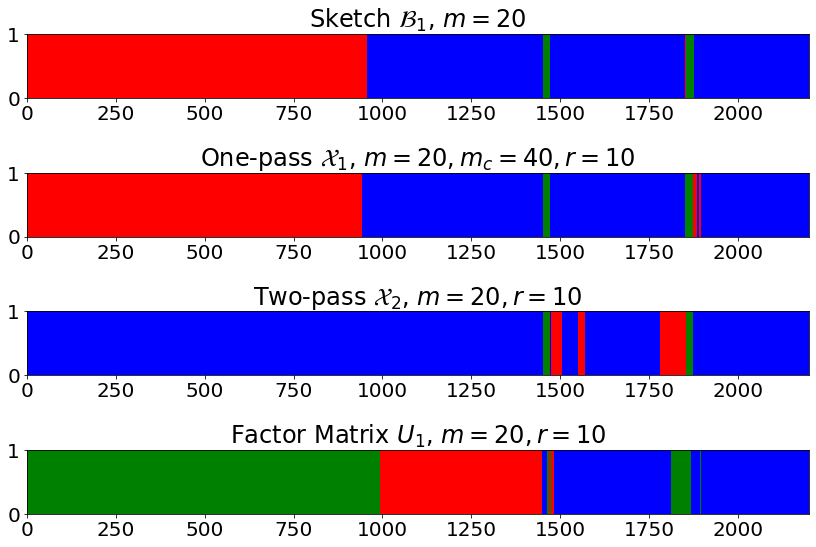} \label{fig:frame_classes}}   \\
  \subfloat[]{\includegraphics[width=1\textwidth]{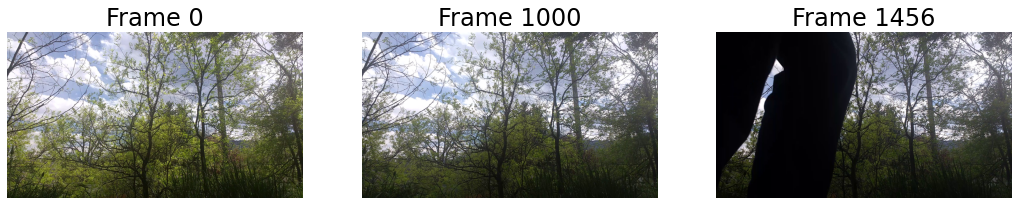} \label{fig:frame_ref}}  
 \caption{(a) Cluster assignments for the 2200 frames in the video. Top run corresponds to using the measurements for the first mode only $B_1$, middle rows use one and two pass approximations of the tensor, and the last row uses the factor matrix $U_1$ for the temporal mode only. We use Gaussian sketching matrices for both spatial modes and the real part of $RFD$ for the temporal mode. Sketching parameters are $m=20$, $m_c = 40$ and rank truncation of $r=10$ in all modes. (b) Three reference frames at 0, 1000 and 1496.} 
 \label{fig:video}
\end{figure}

As was discussed in \cite{Sun2019}, the video is not especially low-rank in practice - in particular along the spatial dimensions in terms of relative error of the reconstruction. However the clusters appear distinct enough that assigning clusters with this summary type of information is still possible.
\begin{figure}[ht]
\centering
\includegraphics[width=1\textwidth]{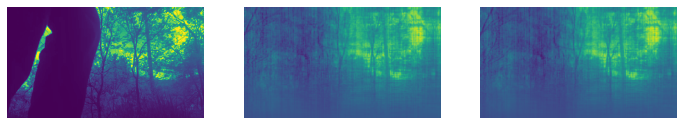}
\caption{The 1456th frame of the grey-scale video is shown for the original, the one-pass, and two-pass reconstructions using sketching dimension of $m=300, m_c = 601$ and $r=50$ for each mode. Although the reconstructions for this choice of sketching dimension and rank truncation are not particularly accurate for this video, nevertheless the reconstructions provide enough information to perform the summary task of clustering the frames into the major temporal changes that occur during the scene.}
\label{fig:video_error}
\end{figure}

\section*{Declarations} 

\textbf{Funding:} Cullen Haselby and Mark A. Iwen were partially supported by NSF DMS 2106472. Deanna Needell and William Swartworth ere partially supported by NSF DMS 2011140 and NSF DMS 2108479. Elizaveta Rebrova was partially supported by NSF DMS 2309685 and NSF DMS 2108479.

\textbf{Conflicts of Interest:} The authors declare no competing interests.
 
\begin{appendices}
\section{Technical Proofs}
\label{appendix:PCPproofs}

Herein we provide proofs for selected results from Section~\ref{sec:RandNumAlgReview}.

\subsection{Proof of Lemma~\ref{lem:fastmatmulbyJLReal}}

Our proof of Lemma~\ref{lem:fastmatmulbyJLReal} will utilize several intermediate lemmas.  Our first lemma concerning the approximate preservation of inner products is a slight generalization of \cite[Corollary 2]{arriaga2006algorithmic}.

\begin{lem}[The JL property implies angle preservation]
	\label{lem:polarizationJL}
	 Let $S\subset \C^n$ with cardinality at most $p$ and $\epsilon\in(0,1)$. If a random matrix $\Omega\in\C^{m\times n}$ has the $(\epsilon / 4,\delta, 4 p^2)$-JL property for 
	\[
		S^{\prime} = \left\{\frac{\vb{x}}{\|\vb{x}\|_2} + \frac{\vb{y}}{\|\vb{y}\|_2}, \frac{\vb{x}}{\|\vb{x}\|_2} - \frac{\vb{y}}{\|\vb{y}\|_2}, \frac{\vb{x}}{\|\vb{x}\|_2} + i\frac{\vb{y}}{\|\vb{y}\|_2}, \frac{\vb{x}}{\|\vb{x}\|_2} - i \frac{\vb{y}}{\|\vb{y}\|_2} ~\big \vert~ \vb{x},\vb{y} \in S\right\},
	\] 
 then
 \begin{equation} \label{eqn:JLforinnerprodonsetdiff}
 |\langle \Omega\vb{x} ,\Omega\vb{y} \rangle - \langle\vb{x},\vb{y}\rangle| \leq \epsilon \|\vb{x}\|_2 \|\vb{y}\|_2~~\forall \vb{x},\vb{y}\in S 
 \end{equation}
 will be satisfied with probability at least $1 - \delta$.
\end{lem}
\begin{proof}
Note that if either $\vb{x} = \vb{0}$ or $\vb{y} = \vb{0}$, then \eqref{eqn:JLforinnerprodonsetdiff} automatically holds because $0 \leq 0$.  Thus, suppose without loss of generality that $\vb{x},\vb{y} \neq 0$. Considering the normalizations $\vb{u} = \frac{\vb{x}}{\|\vb{x}\|_2}, \vb{v} = \frac{\vb{y}}{\|\vb{y}\|_2}$, one can see that the polarization identity implies that
		\begin{align*}
		\abs{\langle\Omega\vb{u},\Omega\vb{v} \rangle - \langle\vb{u},\vb{v} \rangle} &= \abs{ \frac{1}{4} \sum_{\ell = 0}^{3} i^\ell \left(\|\Omega\vb{u}+i^{\ell}\Omega\vb{v}\|_2^2 - \|\vb{u}+i^{\ell}\vb{v}\|_2^2 \right)} \\
		&\leq \frac{1}{4}\sum_{\ell=0}^{3} \frac{\epsilon}{4} \left( \|\vb{u}\|_2 + \|\vb{v}\|_2\right)^2 ~=~ \epsilon
		\end{align*}
 will hold whenever \eqref{equ:JLproperty} holds with $S \leftarrow S'$ and $\epsilon \leftarrow \epsilon / 4$.  The result now follows by renormalizing.  
\end{proof}

\begin{rmk}
\label{Rem:RealPolarizationHelpsalittle}
Note that if $S\subset \R^n$ it suffices for a random matrix $\Omega \in \R^{m \times n}$ to have the $(\epsilon / 2,\delta, 2 p^2)$-JL property for a smaller set $S' \subset \R^n$ in Lemma~\ref{lem:polarizationJL}.  This can be seen by using the real version of the polarization identity instead of the complex version.
\end{rmk}

The next lemma constructs a set $S$ to utilize in Lemma~\ref{lem:polarizationJL} based on two matrices with normalized columns.  The end result is an entrywise approximate matrix multiplication property for the two column-normalized matrices in question.

\begin{lem}[The JL property allows approximate matrix multiplies for unitary matrices] \label{lem:compresseduntrymatmul}
	Let $V\in\C^{n\times p}$ and $U\in\C^{n\times q}$ have unit $\ell^2$-normalized columns. Suppose that $\Omega\in\C^{m \times n}$ satisfies \eqref{eqn:JLforinnerprodonsetdiff} from Lemma \ref{lem:polarizationJL} where $S=\left\{\vb{u}_j \vert \vb{u}_j = U[:,j]\right\}\cup\left\{\vb{v}_j \vert \vb{v}_j = V[:,j]\right\}$. Then
	\[
		\abs{\left(V^*\Omega^* \Omega U - V^* U\right)_{k,j}} \leq \epsilon,\, \textrm{ for all } 1 \leq k \leq p~\textrm{ and }~ 1 \leq j \leq q.
	\]
\end{lem}
\begin{proof}
	Note that $\abs{S} = p + q$. Thus, $\abs{S^{\prime}} \leq 4(p+q)^2$ in Lemma \ref{lem:polarizationJL}.
	Furthermore, 
	\[\Omega V = \left( \begin{matrix}
		\vrule & \vrule & \dots & \vrule \\
		\Omega \vb{v}_1 & \Omega\vb{v}_2 & \dots & \Omega\vb{v}_p \\
		\vrule & \vrule & \dots & \vrule 
	\end{matrix}\right),\, \textrm{ and } \Omega U = \left( \begin{matrix}
		\vrule & \vrule & \dots & \vrule \\
		\Omega \vb{u}_1 & \Omega\vb{u}_2 & \dots & \Omega\vb{u}_q \\
		\vrule & \vrule & \dots & \vrule 
	\end{matrix}\right).\]
	Hence, $\left(\left(\Omega V\right)^* \Omega U\right)_{k,j} = \langle \Omega\vb{u}_j, \Omega\vb{v}_k \rangle$. Therefore, given Lemma \ref{lem:polarizationJL}, for all choices of $k,j$ we have
	\begin{align*}
		\abs{\left(V^*\Omega^* \Omega U - V^* U\right)_{k,j}} & = \abs{\langle \Omega\vb{u}_j,  \Omega \vb{v}_k \rangle - \langle \vb{u}_j, \vb{v}_k \rangle} ~\leq~ \epsilon	\|\vb{v}_k \|_2 \|\vb{u}_j\|_2 ~=~ \epsilon.	
	\end{align*}
\end{proof}

The next lemma constructs a new set $S$ to utilize in Lemma~\ref{lem:polarizationJL} by selecting a well chosen subset of the singular vectors of both $A$ and $B$.  This set will ultimately determine how the finite set $S$ promised by Lemma~\ref{lem:fastmatmulbyJLReal} depends on $A$ and $B$.  As we shall see, it's proven by applying  Lemma~\ref{lem:compresseduntrymatmul} to two unitary matrices provided by the SVDs of $A$ and $B$.

\begin{lem}[The JL property implies the AMM property for arbitrary matrices] \label{lem:fastmatmulbyJL}
Let $A\in\C^{p\times n}$ and $B\in\C^{n\times q}$ have SVDs given by $A = U_1 \Sigma_1 V^*$ and $B =  U \Sigma_2 V^*_2$, and suppose that $\Omega\in\C^{m\times n}$ satisfies the conditions of Lemma~\ref{lem:compresseduntrymatmul} for $U$ and $V$. Then, 
\[
\|A\Omega^* \Omega B - AB\|_F \leq \epsilon \|A\|_F \|B\|_F
\]
\end{lem} 
\begin{proof}
We will expand the quantity of interest according the SVD of $A$ and $B$.  Doing so we see that
	\begin{align*}
		\|A\Omega^* \Omega B - AB\|_F &= \| U_1\Sigma_1 V^* \Omega^* \Omega U \Sigma_2 V_2^* -  U_1\Sigma_1 V^*U \Sigma_2 V_2^*\|_F \\
		&= \| U_1\Sigma_1\left( V^* \Omega^* \Omega U  -   V^*U \right) \Sigma_2 V_2^*\|_F \\	
		&= \| \Sigma_1\left( V^* \Omega^* \Omega U  -   V^*U \right) \Sigma_2 \|_F \\	
		&= \sqrt{\sum_{k=1}^p\sum_{j=1}^q  \left(\Sigma_1\right)_{k,k}^2 \abs{V^* \Omega^* \Omega U  -   V^*U}_{k,j}^2 \left(\Sigma_2\right)_{j,j}^2 }  \\
		&\leq \sqrt{\sum_{k=1}^p\sum_{j=1}^q  \sigma_k (A)^2 \epsilon^2 \sigma_j (B)^2 }  \\
		&= \epsilon \sqrt{\sum_{k=1}^p \sigma_k (A)^2} \sqrt{\sum_{j=1}^q \sigma_j (B)^2 }  \\	
		&= \epsilon \|A\|_F \|B\|_F.	
\end{align*}
\end{proof}

Lemmas \ref{lem:polarizationJL},~\ref{lem:compresseduntrymatmul}, and~\ref{lem:fastmatmulbyJL} now collectively prove the following generalized version of Lemma~\ref{lem:fastmatmulbyJLReal}.

\begin{lem}[The JL property provides the AMM property] \label{lem:fastmatmulbyJLComplex}
Let $A\in\C^{p\times n}$ and $B\in\C^{n\times q}$.  There exists a finite set $S \subset \C^n$ with cardinality $|S| \leq 4(p + q)^2$ (determined entirely by $A$ and $B$) such that the following holds:   If a random matrix $\Omega\in\C^{m\times n}$ has the $(\epsilon/4, \delta, 4(p + q)^2 )$-JL property for $S$, then $\Omega$ will also have the $(\epsilon, \delta)$-AMM property for $A$ and $B$.
\end{lem} 

We will make use of this simple centering result with regards to $L^p$ norms of random variables. 
\begin{lem}
	\label{lem:rawmomentsboundcentered}
Suppose $X$ a real random variable, and let $p\geq 1$, Then,
	\begin{align*}
	   \norm{X-\E[X]}{L^p} \leq 2 \norm{X}{L^p}
	\end{align*}
\end{lem}
\begin{proof}
Let $\mu=\E[X]$. By Jensen's inequality, $\norm{\mu}{L^p} = \abs{\mu} \leq \E[\abs{X}] = \norm{X}{L^1}$ Observe, 
\begin{align*}
    \norm{X - \mu}{L^p} &\leq \norm{X}{L^p} + \norm{\mu}{L^p} \\
     &\leq \norm{X}{L^p} + \norm{X}{L^1} \\
    &\leq 2\norm{X}{L^p}
\end{align*}
Where we have used Minkowski's inequality in the first line, and Jensen's inequality in the third.
\end{proof}

\subsection{Proof of Theorem~\ref{thm:musc2}}
\label{Appsec:FixingMusco2020} 

A similar proof appears in support of \cite[Theorem 2]{Musco2020}, which was itself simplified from earlier work \cite{Chowdhury2019}.  We reproduce the proof here for completeness, and to clarify details.  We begin by restating Theorem~\ref{thm:musc2} for ease of reference.

\begin{thm}[Restatement of Theorem~\ref{thm:musc2}]
Let $X\in \R^{n\times N}$ of rank $\tilde{r} \leq \min \{ n, N \}$ have the full SVD $X = U\Sigma V^T$, and let $V_{r'}\in\R^{N\times r'}$ denote the first $r'$ columns of $V \in \R^{N \times N}$ for all $r' \in [N]$.  Fix $r \in [n]$ and consider the head-tail split $X = X_r + X_{\setminus r}$.  If $\Omega \in\R^{m\times N}$ satisfies 
\begin{enumerate}
    \item subspace embedding property \eqref{equ:JLsubspaceEmbedding} with $\epsilon \leftarrow \frac{\epsilon}{3}$ for $A \leftarrow X_{r}^T$,
    \item approximate multiplication property \eqref{equ:AMMdetprop} with $\epsilon \leftarrow \frac{\epsilon}{6\sqrt{\min\{r,\tilde{r} \} } }$ for $A \leftarrow X_{\setminus r}$ and $B \leftarrow V_{\min\{r,\tilde{r} \} }$, 
    \item JL property \eqref{equ:JLproperty} with $\epsilon \leftarrow \frac{\epsilon}{6}$ for $S \leftarrow \Big\{$the $n$ columns of $X_{\setminus r}^T \Big\}$, and 
    \item approximate multiplication property \eqref{equ:AMMdetprop} with $\epsilon \leftarrow \frac{\epsilon}{6\sqrt{r}}$ for $A \leftarrow X_{\setminus r}$ and $B \leftarrow X_{\setminus r}^T$, 
    \end{enumerate}
then $\tilde{X} := X\Omega^T $ is an $(\epsilon,0,r)$-PCP sketch of $X$.
\end{thm}

\begin{proof}

Let $Q\in\R^{n\times r}$ be a an arbitrary matrix with orthonormal columns so that $QQ^T \in \R^{n \times n}$ is an orthogonal projection matrix. It suffices to show that 

\begin{equation} \label{eqn:pcpbound}
    \abs{\norm{(I-QQ^T)X}{F}^2 - \norm{(I-QQ^T)X\Omega^T}{F}^2} \leq \epsilon\norm{(I - QQ^T) X}{F}^2.
\end{equation}
Writing $X$ in terms of its head-tail split, \eqref{eqn:pcpbound} becomes 
\begin{equation} \label{eqn:pcpbound_split}
    \abs{\norm{(I-QQ^T)(X_{r}+ X_{\setminus r})}{F}^2 - \norm{(I-QQ^T)(X_{r}+ X_{\setminus r})\Omega^T}{F}^2} \leq \epsilon\norm{(I - QQ^T) X}{F}^2, 
\end{equation}
where $X_{r}$ has the thin SVD representation $X_{r} = U_{r} \Sigma_r V^T_{r}$ with $U_r$ and $V_r$ containing the first $r$ columns of $U$ and $V$ from the full SVD of $X$, respectively.\footnote{If $r \geq N$ we let  $V_r = V$.} 

Letting $P := (I-QQ^T)$, and using that $\tr{AA^T} = \norm{A}{F}^2$, one may expand the left hand side of \eqref{eqn:pcpbound_split}.  Noting that $X_{\setminus r} X_{r}^T = 0$ and $P = P^T$ while doing so, we can now see that
\begin{align*}&\abs{\norm{(I-QQ^T)(X_{r}+ X_{\setminus r})}{F}^2 - \norm{(I-QQ^T)(X_{r}+ X_{\setminus r})\Omega^T}{F}^2} ~=~\\
    &\abs{\tr\left(P(X_{r} X_{r}^T  +X_{\setminus r}X_{\setminus r}^T)P\right) - \tr\left(P(X_{r} \Omega^T\Omega X_{r}^T  + X_{r}\Omega^T \Omega X_{\setminus r}^T  + X_{\setminus r} \Omega^T \Omega X_{r}^T +X_{\setminus r}\Omega^T \Omega X_{\setminus r}^T) P\right) }.
\end{align*}
Regrouping terms in the last expression while noting the invariance of trace to transposition, one can we can now further see that
\begin{align*}&\abs{\tr\left(P(X_{r} X_{r}^T  +X_{\setminus r}X_{\setminus r}^T)P\right) - \tr\left(P(X_{r} \Omega^T\Omega X_{r}^T  + X_{r}\Omega^T \Omega X_{\setminus r}^T  + X_{\setminus r} \Omega^T \Omega X_{r}^T +X_{\setminus r}\Omega^T \Omega X_{\setminus r}^T) P\right) }\\
    &=~\abs{ \tr\left(P(X_{r} X_{r}^T -  X_{r}\Omega^T \Omega X_{ r}^T)P\right) + 2 \tr \left(P(X_{\setminus r} \Omega^T \Omega X_{r}^T) P\right) + \tr\left(P (X_{\setminus r}X_{\setminus r}^T - X_{\setminus r} \Omega^T \Omega X_{\setminus r}^T)P\right)} \\
    &\leq~\abs{\tr\left(P(X_{r} X_{r}^T -  X_{r}\Omega^T \Omega X_{ r}^T)P\right)} + 2\abs{\tr\left(P(X_{\setminus r} \Omega^T \Omega X_{r}^T) P\right)} + \abs{\tr\left(P (X_{\setminus r}X_{\setminus r}^T - X_{\setminus r} \Omega^T \Omega X_{\setminus r}^T)P\right)}. 
\end{align*}

Looking back at \eqref{eqn:pcpbound_split} in the light of this last computation, one can now see that it suffices to prove that
\begin{align}
&\abs{\tr\left(P(X_{r} X_{r}^T -  X_{r}\Omega^T \Omega X_{ r}^T)P\right)} + 2\abs{\tr\left(P(X_{\setminus r} \Omega^T \Omega X_{r}^T) P\right)} + \abs{\tr\left(P (X_{\setminus r}X_{\setminus r}^T - X_{\setminus r} \Omega^T \Omega X_{\setminus r}^T)P\right)}   \nonumber\\
&\leq~ \epsilon \norm{ P X }{F}^2 
\label{equ:NewConditiontoProve}
\end{align}
holds to establish the desired result.  Going forward we will therefore aim to prove \eqref{equ:NewConditiontoProve} by proving each of the following three bounds:  
\begin{itemize}
\item[(a)] $\abs{\tr\left(P(X_{r} X_{r}^T -  X_{r}\Omega^T \Omega X_{ r}^T)P\right)} ~\leq~ \frac{\epsilon}{3} \norm{ P X }{F}^2$, 
\item[(b)] $\abs{\tr\left(P(X_{\setminus r} \Omega^T \Omega X_{r}^T) P\right)} ~\leq~ \frac{\epsilon}{6} \norm{ P X }{F}^2$, and
\item[(c)] $\abs{\tr\left(P (X_{\setminus r}X_{\setminus r}^T - X_{\setminus r} \Omega^T \Omega X_{\setminus r}^T)P\right)} ~\leq~ \frac{\epsilon}{3} \norm{ P X }{F}^2$.
\end{itemize}
Proving (a) -- (c) will establish \eqref{equ:NewConditiontoProve}, thereby completing the proof.\\

\underline{\textbf{Proof of Bound (a):}}  Using again that $\tr{AA^T} = \tr{A^TA} = \norm{A}{F}^2$ and that $P = P^T$ we have
\begin{align*}
     \abs{\tr\left(P(X_{r} X_{r}^T -  X_{r}\Omega^T \Omega X_{ r}^T)P\right)} = \abs{\norm{ X^T_{r}P}{F}^2 - \norm{\Omega X^T_{r}P}{F}^2}.
\end{align*}
Applying Lemma~\ref{lem:subembedimpliesfrobound} in the light of assumption \ref{thm:musc2_assump1}, we now have that
\begin{align*}
     \abs{\tr\left(P(X_{r} X_{r}^T -  X_{r}\Omega^T \Omega X_{ r}^T)P\right)} &\leq \frac{\epsilon}{3} \norm{ X^T_{r}P}{F}^2 ~=~ \frac{\epsilon}{3} \norm{ V_r V_r^T X^T P}{F}^2\\
     &\leq \frac{\epsilon}{3} \norm{ X^T P}{F}^2 ~=~ \frac{\epsilon}{3} \norm{ P X}{F}^2
\end{align*}
as desired.\\

\underline{\textbf{Proof of Bound (b):}}  Using the invariance of trace to both transposition and permutations, as well as that $P^T = P = P^2$, we can see that
\begin{align*}
\abs{\tr\left(P(X_{\setminus r} \Omega^T \Omega X_{r}^T) P\right)} = \abs{\tr\left(P(X_{r} \Omega^T \Omega X_{\setminus r}^T) P\right)} = \abs{\tr\left(P X_{r} \Omega^T \Omega X_{\setminus r}^T \right)}.
\end{align*}
Recalling the full SVD $X = U\Sigma V^T$, we note that if $\tilde{r}:= {\rm rank}(X) < \min \{n,N \}$ we can remove the last $n - \tilde{r}$ columns of $U$, the last $N - \tilde{r}$ columns of $V$, and the last $n - \tilde{r}$ rows and $N - \tilde{r}$ columns of $\Sigma$ to form the thin SVD $X = \tilde{U}\tilde{\Sigma} \tilde{V}^T$ with $\tilde{U} \in \mathbbm{R}^{n \times \tilde{r}}$, $\tilde{\Sigma} \in \R^{\tilde{r} \times \tilde{r}}$, and $\tilde{V} \in \mathbbm{R}^{N \times \tilde{r}}$.  Having done so we note that $\tilde{\Sigma}$ will be invertable and that $\tilde{U} \tilde{U}^T X_{r} = X_{r}$ so that we may write
\begin{align*}
\abs{\tr\left(P(X_{\setminus r} \Omega^T \Omega X_{r}^T) P\right)} &= \abs{\tr\left(P X_{r} \Omega^T \Omega X_{\setminus r}^T \right)}\\
&= \abs{\tr\left(P \tilde{U} \tilde{U}^T X_{r} \Omega^T \Omega X_{\setminus r}^T \right)}\\
&= \abs{\tr\left( P \tilde{U} \tilde{\Sigma} \tilde{V}^T \tilde{V} \tilde{\Sigma}^{-1} \tilde{U}^T X_{ r} \Omega^T \Omega X_{\setminus r}^T\right)}\\
&= \abs{\tr \left( \left( P \tilde{U} \tilde{\Sigma} \tilde{V}^T \right) \left(\tilde{V} \tilde{\Sigma}^{-1} \tilde{U}^T X_{ r} \Omega^T \Omega X_{\setminus r}^T \right) \right)}\\
&= \abs{\tr \left( \left( P X \right) \left(\tilde{V} \tilde{\Sigma}^{-1} \tilde{U}^T X_{ r} \Omega^T \Omega X_{\setminus r}^T \right) \right)}.
\end{align*}

Recall now that $\langle A,B \rangle_{F} := \tr(A B^T)$ is an inner product on matrices with $\| A \|_{F} = \sqrt{\tr(A A^T)}$.  Hence, we may apply the Cauchy–Schwarz inequality to our last expression to see that 
\begin{align*}
\abs{\tr\left(P(X_{\setminus r} \Omega^T \Omega X_{r}^T) P\right)} &= \abs{\tr \left( \left( P X \right) \left(\tilde{V} \tilde{\Sigma}^{-1} \tilde{U}^T X_{ r} \Omega^T \Omega X_{\setminus r}^T \right) \right)}\\
&\leq \| P X \|_F \left\| \tilde{V} \tilde{\Sigma}^{-1} \tilde{U}^T X_{ r} \Omega^T \Omega X_{\setminus r}^T \right\|_F\\
&= \| P X \|_F \left\| \tilde{\Sigma}^{-1} \tilde{U}^T X_{ r} \Omega^T \Omega X_{\setminus r}^T \right\|_F.
\end{align*}
Expanding $X_{r}$ in terms of its thin SVD representation $X_{r} = U_{r} \Sigma_r V^T_{r}$ we now have that 
\begin{align*}
\abs{\tr\left(P(X_{\setminus r} \Omega^T \Omega X_{r}^T) P\right)} &\leq \| P X \|_F \left\| \tilde{\Sigma}^{-1} \tilde{U}^T X_{ r} \Omega^T \Omega X_{\setminus r}^T \right\|_F\\
&= \| P X \|_F \left\| \tilde{\Sigma}^{-1} \tilde{U}^T U_{r} \Sigma_r V^T_{r} \Omega^T \Omega X_{\setminus r}^T \right\|_F\\
&= \| P X \|_F \left\| V^T_{\min \{ r, \tilde{r} \}} \Omega^T \Omega X_{\setminus r}^T \right\|_F\\ 
&= \| P X \|_F \left\| X_{\setminus r} \Omega^T \Omega V_{\min \{ r, \tilde{r} \}} \right\|_F.
\end{align*}

Finally, we may now use that $X_{\setminus r} V_{\min \{ r, \tilde{r} \}} = X(I_N - V_r V_r^T) V_{\min \{ r, \tilde{r} \}} = X \left(V_{\min \{ r, \tilde{r} \}} - V_{\min \{ r, \tilde{r} \}} \right) = 0$ to see that
\begin{align*}
\abs{\tr\left(P(X_{\setminus r} \Omega^T \Omega X_{r}^T) P\right)} &\leq \| P X \|_F \left\| X_{\setminus r} \Omega^T \Omega V_{\min \{ r, \tilde{r} \}} \right\|_F\\
&= \| P X \|_F \left\| X_{\setminus r} \Omega^T \Omega V_{\min \{ r, \tilde{r} \}} - X_{\setminus r} V_{\min \{ r, \tilde{r} \}} \right\|_F\\
&\leq \| P X \|_F \left( \frac{\epsilon}{6\sqrt{\min \{ r, \tilde{r} \}}} \norm{X_{\setminus r}}{F}\norm{V_{\min \{ r, \tilde{r} \}}}{F} \right),
\end{align*}
where we have utilized assumption \ref{thm:musc2_assump2} in the last inequality.  We are now finished after using that $\norm{V_{\min \{ r, \tilde{r} \}}}{F} = \sqrt{\min \{ r, \tilde{r} \}}$, and noting that $\norm{X_{\setminus r}}{F} \leq \| P X \|_F$ holds for all rank $n - r$ or greater orthogonal projections $P$ by the definition of $X_r$.\\

\underline{\textbf{Proof of Bound (c):}} Again using the invariance of trace to permutations as well as $P = P^2 = I-QQ^T$ we have that
\begin{align*}
\abs{\tr\left(P (X_{\setminus r}X_{\setminus r}^T - X_{\setminus r} \Omega^T \Omega X_{\setminus r}^T)P\right)} &= \abs{\tr\left(P (X_{\setminus r}X_{\setminus r}^T - X_{\setminus r} \Omega^T \Omega X_{\setminus r}^T)\right)}\\
&= \abs{\tr\left( (I-QQ^T) (X_{\setminus r}X_{\setminus r}^T - X_{\setminus r} \Omega^T \Omega X_{\setminus r}^T)\right)}\\
& \leq \abs{\tr\left( X_{\setminus r}X_{\setminus r}^T - X_{\setminus r} \Omega^T \Omega X_{\setminus r}^T\right)}\\ 
&~~~~~~~~~~~~~~~~+ \abs{\tr\left( QQ^T (X_{\setminus r}X_{\setminus r}^T - X_{\setminus r} \Omega^T \Omega X_{\setminus r}^T)\right)}\\
&\leq \abs{ \left\| X_{\setminus r}^T \right\|_F^2 - \left\| \Omega X_{\setminus r}^T \right\|_F^2} + \left\| QQ^T \right\|_F \left\| X_{\setminus r}X_{\setminus r}^T - X_{\setminus r} \Omega^T \Omega X_{\setminus r}^T \right\|_F,
\end{align*}
where we have again utilized Cauchy–Schwarz in the last inequality.  Utilizing assumption \ref{thm:musc2_assump3}, the first term just above can be bounded by $\frac{\epsilon}{6} \| X_{\setminus r}^T \|_F^2 = \frac{\epsilon}{6} \| X_{\setminus r}\|_F^2$ using an argument analogous to the proof of Lemma~\ref{lem:subembedimpliesfrobound}.  Doing so we see that
\begin{align*}
\abs{\tr\left(P (X_{\setminus r}X_{\setminus r}^T - X_{\setminus r} \Omega^T \Omega X_{\setminus r}^T)P\right)} &\leq \frac{\epsilon}{6} \| X_{\setminus r} \|_F^2 + \left\| QQ^T \right\|_F \left\| X_{\setminus r}X_{\setminus r}^T - X_{\setminus r} \Omega^T \Omega X_{\setminus r}^T \right\|_F\\
&= \frac{\epsilon}{6} \| X_{\setminus r} \|_F^2 + \sqrt{r} \left\| X_{\setminus r}X_{\setminus r}^T - X_{\setminus r} \Omega^T \Omega X_{\setminus r}^T \right\|_F.
\end{align*}

Finally, we may now employ assumption \ref{thm:musc2_assump4} to bound the second term just above.  Doing so we obtain that
\begin{align*}
\abs{\tr\left(P (X_{\setminus r}X_{\setminus r}^T - X_{\setminus r} \Omega^T \Omega X_{\setminus r}^T)P\right)} &\leq \frac{\epsilon}{6} \| X_{\setminus r} \|_F^2 + \sqrt{r} \left\| X_{\setminus r}X_{\setminus r}^T - X_{\setminus r} \Omega^T \Omega X_{\setminus r}^T \right\|_F\\
&\leq \frac{\epsilon}{6} \| X_{\setminus r} \|_F^2 + \sqrt{r} \frac{\epsilon}{6 \sqrt{r}} \| X_{\setminus r} \|_F^2\\
&= \frac{\epsilon}{3} \| X_{\setminus r} \|_F^2.
\end{align*}
To conclude we note again that $\norm{X_{\setminus r}}{F} \leq \| P X \|_F$ holds for all rank $n - r$ or greater orthogonal projections $P$ by the definition of $X_r$.
\end{proof}

\section{Algorithms}
\label{appendix:algs}
There are four tasks relevant in Algorithm \ref{alg:loo_one_pass_prime} where by making difference choices for these tasks, we get variants of the algorithm. Two of the tasks are related to the measurement process: sketching the tensor in order to produce leave-one-out measurements, and also sketching the tensor without leaving out a mode in order to produce measurements useful in computing the core. The second two tasks are then recovering the factor matrices, and recovering the core.

Note, if we permit a second pass on the tensor $\mathcal{X}$ after estimating factor matrices $Q_i$, then core measurements are not necessary, and the optimal core given these factor matrices can be found by computing
\[
\mathcal{G} = \mathcal{X}\times_1 Q_1^T \times_2 Q_2^T \dots \times_d Q_d^T
\]
See section 4.2 in \cite{Kolda2009}. This then is the computation used in two-pass versions of the algorithm; it requires first to compute the factor matrices from measurements, and then apply these factors modewise to the original tensor; no separate measurement tensor for the core is required. Crucially however, this relies on a second access to the data; for which we are computing a HOSVD. 

In this appendix then, we break apart Algorithm \ref{alg:loo_one_pass_prime} into these tasks, and show how combining different choices for these tasks produces the variants of the algorithm considered.

Within the pseudo-code, ``unfold'' refers to the operation of taking a tensor and flattening it into a matrix of the size listed by arranging the specified mode's fibers as the columns. The operation ``fold'' is the inverse of this, taking a matrix as viewed as an unfolding along the specified mode and reshaping it into a tensor of the given dimensions. 

Algorithm \ref{alg:kron_sketching} takes measurement matrices (e.g. sub-gaussian random matrices) and the data tensor $\mathcal{X}$ and produces a set of leave-one-out measurements which can be viewed as tensors or as flattened matrices. That is, it is practical to apply the measurement matrices along the modes and obtain a tensor of measurements with $d-1$ modes each of length $m$ and one mode of length $n$, see Figure~\ref{fig:brks_diagram} for a schematic depiction. The measurement process takes slices of the tensor and maps them to (smaller) slices with (mostly) shorter edge lengths. Or we can conceive of the measurements as a matrix by unfolding the measurement tensor along the mode that is uncompressed and thus obtaining a matrix of size $n\times m^{d-1}$, one such matrix for each mode $i\in [d]$.  It is Kronecker structured because of the correspondence of modewise products of tensors with matrices to matrix products of unfoldings of a tensor with matrices, see for example \eqref{eqn:mat_tucker}.

\begin{algorithm}[hbt!]
\SetKwComment{Comment}{\# }{}
\SetKwInOut{Input}{input}\SetKwInOut{Output}{output}%
\caption{Leave-One-Out Kronecker Sketching}	\label{alg:kron_sketching}
 \Input
  {\par $\Omega_{(i,j)}$ where $\text{rank}(\Omega_{(i,i)}) = n$ and $\Omega_{(i,j)} \in \R^{m \times n}$ for $i,j \in [d]$, $i\neq j$
  \par $\mathcal{X}$ a $d$ mode tensor with side lengths $n$}
\Output{$B_{i}$ for $i\in[d]$ }
\For{$i\in[d]$}{
    \par $\mathcal{B}_{i} \gets \mathcal{X}\times_1 \Omega_{(i,1)} \times_2 \Omega_{(i,2)} \times_3 \dots \times_{d} \Omega_{(i,d)}$
    
    \Comment{Now unfold the measurement tensor so the mode-$i$ fibers are columns, size $n\times m^{d-1}$}
    \par $B_{i} \gets \text{unfold} (\mathcal{B}_i,n\times m^{d-1},\text{mode}=i)$
}
\end{algorithm}
Alternatively, we can use Algorithm \ref{alg:khat_sketching} which also produces a set of leave-one-out measurements of the tensor $\mathcal{X}$ which can be viewed as a matrix. It is Khatri-Rao structured because the measurement matrix applied to the unfolding is formed using Khatri-Rao products. Note, unlike Algorithm~\ref{alg:kron_sketching}, there is not necessarily any natural way to view the measurements $B_i$ as tensors of $d$ modes - the sketching process in this case takes slices of the tensor and maps them to vectors and we gather these into matrices $B_i$ of size $n \times m$ for each $i\in[d]$.

\begin{algorithm}[hbt!]
\SetKwComment{Comment}{\# }{}
\SetKwInOut{Input}{input}\SetKwInOut{Output}{output}%
\caption{Leave-One-Out Khatri-Rao Sketching}	\label{alg:khat_sketching}
 \Input
  {\par $\Omega_{(i,j)}$ where $\text{rank}(\Omega_{(i,i)}) = n$ and $\Omega_{(i,j)} \in \R^{m \times n}$ for $i,j \in [d]$, $i\neq j$
  \par $\mathcal{X}$ a $d$ mode tensor with side lengths $n$}
\Output{$B_{i}$ for $i\in[d]$ }
\For{$i\in[d]$}{
    $B_i \gets \Omega_{(i,i) }X_{[i]} \left( \Omega_{(i,1)} \sbullet \Omega_{(i,2)} \sbullet \dots \sbullet \Omega_{i,i-1} \sbullet \Omega_{i,i+1} \sbullet \dots \sbullet \Omega_{i,d}  \right)^T$

}
\end{algorithm}

In order to estimate the core, we wish another, independent set of measurements. These are obtained in the manner as Algorithm \ref{alg:kron_sketching}, only now we are permitted to compress all the modes, producing a tensor of $d$ modes all with side length equal to $m_c$.

\begin{algorithm}[hbt!]
\SetKwComment{Comment}{\# }{}
\SetKwInOut{Input}{input}\SetKwInOut{Output}{output}%
\caption{Core Sketching}	\label{alg:core_sketching}
 \Input
  {\par $\Phi_{i}\in \R^{m_c \times n}$, $i\in[d]$
  \par $\mathcal{X}$ a $d$ mode tensor with side lengths $n$}
\Output{$\mathcal{B}_{c}$}
    \par $\mathcal{B}_{c} \gets \mathcal{X}\times_1 \Phi_{1} \times_2 \Phi_{2} \times_3 \dots \times_{d} \Phi_{d}$
\end{algorithm}

Next, we describe the procedure which takes as input $(a)$ leave-one-out measurements of $\mathcal{X}$ (from Algorithm \ref{alg:kron_sketching} or \ref{alg:khat_sketching}), $(b)$ the full-rank sensing matrices applied to the uncompressed modes of $\mathcal{X}$, and $(c)$ our desired target rank vector of $\vb{r}$, and then outputs a factor matrix $Q_i$ for each mode.  In the case of exact arithmetic, no rank truncation, and no noise, this exactly recovers the factors (see \eqref{eqn:factorerrorbound}).

\begin{algorithm}[hbt!]
\SetKwComment{Comment}{\# }{}
\SetKwInOut{Input}{input}\SetKwInOut{Output}{output}%
\caption{Recover HOSVD Factors from Leave-One-Out Measurements}	\label{alg:recover_factors}
 \Input
  { \par $B_{i} \in R^{n\times m^{d-1}}$ for $i\in[d]$ measurements that leave mode $i$ uncompressed
  \par $\Omega_{(i,i)}\in R^{n\times n}$ for $i \in [d]$
  \par $\vb{r} = (r, \dots, r)$ desired rank for HOSVD
 }
\Output{$ Q_1, \dots, Q_d \in \R^{n\times r}$ }
\Comment{Factor matrix recovery}
\For{$i\in[d]$}{
\Comment{Solve $n\times n$ linear system}
\par Solve $\Omega_{(i,i)}F_{i} = B_{i}$ for $F_i$ 

\Comment{Compute SVD and keep the top $r$ singular vectors}
\par $ U,\Sigma, V^T \gets \text{SVD}(F_{i})$
\par $Q_i \gets U_{:,:r}$
}
\end{algorithm}

Lastly, we consider the task of obtaining the core of the HOSVD of the data tensor. The two ways described in section \ref{sec:MainResBRKS} are to either compute the core using a second access to the data tensor - in which case this is a matter of applying the transpose of the factor matrices from Algorithm \ref{alg:recover_factors} to the data tensor. This is detailed in Algorithm \ref{alg:kron_2pass}.
\begin{algorithm}[hbt!]
\SetKwComment{Comment}{\# }{}
\SetKwInOut{Input}{input}\SetKwInOut{Output}{output}%
\caption{Compute HOSVD Core with Second Access}	\label{alg:compute_core}
 \Input
  {
  \par $ Q_1, \dots, Q_d \in \R^{n\times r}$ computed factor matrices
   \par $\mathcal{X}$ a $d$ mode tensor with side lengths $n$.
 }
 \Output{$ \mathcal{G}$ a $d$ mode tensor with side lengths $r$}
 
\par $\mathcal{G} = \mathcal{X}\times_1 Q_1^T \times_2 Q_2^T \dots \times_d Q_d^T$
\end{algorithm}

In the scenario in which a second access to the tensor is not desired, instead we obtain the core of the HOSVD by solving a linear system involving the measurement operators and the factor matrices, see \eqref{eqn:core_linear_solve}. This is equivalent to solving the linear system a mode at a time as detailed in Algorithm \ref{alg:brks_tucker_1pass_core}, a method of practical value because it does not require as much working memory.

\begin{algorithm}[hbt!]
\SetKwComment{Comment}{\# }{}
\SetKwInOut{Input}{input}\SetKwInOut{Output}{output}%
\caption{Recover HOSVD Core from Measurements}	\label{alg:brks_tucker_1pass_core}
 \Input
  {
  \par $\mathcal{B}_{c}$ a $d$ mode tensor with side lengths $m_c$
  \par $\Phi_i \in\R^{m_c \times n}$
  \par $ Q_1, \dots, Q_d \in \R^{n\times r}$ computed factor matrices
 }
 \Output{$ \mathcal{H}$ }
\For{$i\in[d]$}{
\Comment{unfold measurements, mode-$i$ fibers are columns, size $m_c\times r^{(i-1)}m_c^{d-1-(i-1)}$}
\par $H \gets \text{unfold}(\mathcal{B}_c,m_c\times r^{(i-1)}m_c^{d-1-(i-1)},\text{mode}=i)$ 

\Comment{Undo the mode-$i$ measurement operator and factor's action by finding least square solution to $m_c\times r$ over-determined linear system }
\par Solve $\Phi_i Q_i H_{\text{new}} = H$ for $H_{new}$

\Comment{reshape the flattened partially solved core into a tensor}

\par $\mathcal{B}_c \gets \text{fold}(H_{\text{new}},\underbrace{r\times r\dots \times r}_{i}\underbrace{\times m_c\times \dots \times m_c}_{d-i},\text{mode}=i)$

\Comment{Each iteration $m_c\to r$ in $i$th mode}
} 
\end{algorithm}

A third possibility is to ``re-use'' leave-one-out measurements to compute the core. Theoretically this involves new dependencies on the errors introduced by estimating the factors and errors introduced by recovering the core that are not addressed in any of our main results. Practically however, it would be desirable to avoid having to produce the core measurement tensor, and empirically on synthetic data the overall error is not effected. 

\begin{algorithm}[hbt!]
\SetKwComment{Comment}{\# }{}
\SetKwInOut{Input}{input}\SetKwInOut{Output}{output}%
\caption{Recover HOSVD Core from Recycled Leave One Out Measurements}	\label{alg:recover_core}
 \Input
  {
  \par $B_j$ for a fixed $j$.
  \par $\Omega_{j,i}$ for each $i\in[d]$
  \par $ Q_1, \dots, Q_d \in \R^{n\times r}$ computed factor matrices
 }
 \Output{$ \mathcal{H}$ }

\For{$i\in[d]$}{
\Comment{unfold measurements, mode-$i$ fibers are columns, size $m\times r^{(i-1)}m^{d-1-(i-1)}$}
\par $H \gets \text{unfold}(\mathcal{B}_j,m\times r^{(i-1)}m^{d-1-(i-1)},\text{mode}=i)$ 

\Comment{Undo the mode-$i$ measurement operator and factor's action by finding least square solution to $m\times r$ over-determined linear system }
\par Solve $\Omega_{(j,i)} Q_i H_{\text{new}} = H$ for $H_{new}$

\Comment{reshape the flattened partially solved core into a tensor}

\par $\mathcal{B}_j \gets \text{fold}(H_{\text{new}},\underbrace{r\times r\dots \times r}_{i}\underbrace{\times m_c\times \dots \times m_c}_{d-i},\text{mode}=i)$

\Comment{Each iteration $m\to r$ in $i$th mode or $n \to n$ when $i=j$}
} 
\end{algorithm}

We are now able to state the variants of the general algorithm.

\begin{algorithm}[hbt!]
\SetKwComment{Comment}{\# }{}
\SetKwInOut{Input}{input}\SetKwInOut{Output}{output}%
\caption{Recover HOSVD Kronecker One Pass}	\label{alg:kron_1pass}
\Comment{Obtain measurements for factors with Alg.~\ref{alg:kron_sketching} }
 \par $B_1, B_2, \dots, B_d \gets \text{Leave-One-Out Kronecker Sketching} (\left\{ \Omega_{(i,j)}\right\}_{i,j \in [d]}, \mathcal{X})$
 
 \Comment{Obtain measurement for core with Alg.~\ref{alg:core_sketching}}
 \par $\mathcal{B}_c \gets \text{Core Sketching} (\left\{ \Phi_i\right\}_{i\in[d]}, \mathcal{X})$
 
 \Comment{Estimate factor matrices using Alg.~\ref{alg:recover_factors}}
 \par $Q_1, Q_2, \dots, Q_d \gets$ Recover HOSVD Factors from 
 
 Leave-One-Out Measurements $(\left\{ \Omega_{(i,i)}\right\}_{i \in [d]},\left\{ B_{i}\right\}_{i \in [d]}, (r,\dots, r))$
 
 \Comment {Estimate core using Alg.~\ref{alg:recover_core}}
 \par $\mathcal{H} \gets \text{Recover HOSVD Core from Measurements} (\left\{ \Phi_i\right\}_{i \in [d]}, \mathcal{B}_c)$
  
 \Output{$ \hat{\mathcal{X}} = [\![ \mathcal{H}, Q_1,\dots Q_d ]\!]$ }
\end{algorithm}

\begin{algorithm}[hbt!]
\SetKwComment{Comment}{\# }{}
\SetKwInOut{Input}{input}\SetKwInOut{Output}{output}%
\caption{Recover HOSVD Kronecker Two Pass}	\label{alg:kron_2pass}
\Comment{Obtain measurements for factors with Alg.~\ref{alg:kron_sketching}}
 \par $B_1, B_2, \dots, B_d \gets \text{Leave-One-Out Kronecker Sketching} (\left\{ \Omega_{(i,j)}\right\}_{i,j \in [d]}, \mathcal{X})$
 
 \Comment{Estimate factor matrices using Alg.~\ref{alg:recover_factors}}
 \par $Q_1, Q_2, \dots, Q_d \gets$ Recover HOSVD Factors from 
 
 Leave-One-Out Measurements$(\left\{ \Omega_{(i,i)}\right\}_{i \in [d]},\left\{ B_{i}\right\}_{i \in [d]}, (r,\dots, r))$
 
 \Comment{Compute core using Alg.~\ref{alg:compute_core}}
 \par $\mathcal{H} \gets \text{Compute HOSVD Core with Second Access} (\left\{ Q_i \right\}_{i \in [d]}, \mathcal{X})$
  
 \Output{$ \hat{\mathcal{X}} = [\![ \mathcal{H}, Q_1,\dots Q_d ]\!]$ }
\end{algorithm}

\begin{algorithm}[hbt!]
\SetKwComment{Comment}{\# }{}
\SetKwInOut{Input}{input}\SetKwInOut{Output}{output}%
\caption{Recover HOSVD Khatri-Rao One Pass}	\label{alg:khat_1pass}
\Comment{Obtain measurements for factors with Alg.~\ref{alg:khat_sketching}}
 \par $B_1, B_2, \dots, B_d \gets \text{Leave-One-Out Khatri-Rao Sketching} (\left\{ \Omega_{(i,j)}\right\}_{i,j \in [d]}, \mathcal{X})$
 
 \Comment{Obtain measurements for core with Alg.~\ref{alg:core_sketching}}
 \par $\mathcal{B}_c \gets \text{Core Sketching} (\left\{ \Phi_i\right\}_{i\in[d]}, \mathcal{X})$
 
 \Comment{Estimate factor matrices using Alg.~\ref{alg:recover_factors}}
 \par $Q_1, Q_2, \dots, Q_d \gets$ Recover HOSVD Factors from 
 Leave-One-Out Measurements $(\left\{ \Omega_{(i,i)}\right\}_{i \in [d]},\left\{ B_{i}\right\}_{i \in [d]}, (r,\dots, r))$
 
 \Comment {Estimate core using Alg.~\ref{alg:recover_core}}
 \par $\mathcal{H} \gets \text{Recover HOSVD Core from Measurements} (\left\{ \Phi_i\right\}_{i \in [d]}, \mathcal{B}_c)$
  
 \Output{$ \hat{\mathcal{X}} = [\![ \mathcal{H}, Q_1,\dots Q_d ]\!]$ }
\end{algorithm}

\begin{algorithm}[hbt!]
\SetKwComment{Comment}{\# }{}
\SetKwInOut{Input}{input}\SetKwInOut{Output}{output}%
\caption{Recover HOSVD Khatri-Rao Two Pass}	\label{alg:khat_2pass}
\Comment{Obtain measurements for factors with Alg.~\ref{alg:khat_sketching}}
 \par $B_1, B_2, \dots, B_d \gets \text{Leave-One-Out Khatri-Rao Sketching} (\left\{ \Omega_{(i,j)}\right\}_{i,j \in [d]}, \mathcal{X})$
 
  \Comment{Estimate factor matrices using Alg.~\ref{alg:recover_factors}}
 \par $Q_1, Q_2, \dots, Q_d \gets$ Recover HOSVD Factors from
 Leave-One-Out Measurements $(\left\{ \Omega_{(i,i)}\right\}_{i \in [d]},\left\{ B_{i}\right\}_{i \in [d]}, (r,\dots, r))$
 
  \Comment{Compute core using Alg.~\ref{alg:compute_core}}
 \par $\mathcal{H} \gets \text{Compute HOSVD Core with Second Access} (\left\{ Q_i \right\}_{i \in [d]}, \mathcal{X})$
  
 \Output{$ \hat{\mathcal{X}} = [\![ \mathcal{H}, Q_1,\dots Q_d ]\!]$ }
\end{algorithm}
\end{appendices}

\clearpage


\bibliography{biblio}

\end{document}